\documentclass[12pt]{article}
\usepackage{amsmath}
\usepackage{amsfonts}
\usepackage{amssymb}
\usepackage{amsthm}
\usepackage{nicefrac}
\usepackage{mathrsfs}
\usepackage{cite}

\usepackage{authblk}

\usepackage{csquotes}
\usepackage{color}
\usepackage{colortbl}
\usepackage{enumitem}
\usepackage{fontenc}
\usepackage{accents}
\usepackage{float}
\usepackage{subfigure}
\usepackage{wrapfig}
\usepackage{multicol}

\usepackage{nccmath}
\usepackage{cancel}
\usepackage[letterpaper,textwidth=6.5in]{geometry} 

\usepackage[pdftex]{hyperref}
\usepackage{datetime}

\newcommand{\x}{{\rm x}}

\newcommand{\Lie}{\ensuremath{\pounds}} 
\newcommand{\R}{\ensuremath{\mathbb{R}}} 
\newcommand{\N}{\ensuremath{\mathbb{N}}} 

\newcommand{\dbar}{\,{\mathchar'26\mkern-12mud}\,}

\usepackage[pdftex]{graphicx}
\usepackage{tikz}

\newtheorem{thm}{Theorem}
\newtheorem{cor}{Corollary}
\newtheorem{lemma}{Lemma}

\newtheorem{conj}{Conjecture}

\theoremstyle{definition}
\newtheorem{defn}{Definition}

\begin{document}
	\title{The Hadamard condition on a Cauchy surface and the renormalized stress-energy tensor}
\author{Benito A. Ju\'arez-Aubry$^{1,2}$\thanks{{\tt benito.juarezaubry@york.ac.uk}}, Bernard S. Kay$^{2}$\thanks{{\tt bernard.kay@york.ac.uk}}, Tonatiuh Miramontes$^{1,3}$\thanks{{\tt tonatiuhmiramontes@tec.mx}} \\ and Daniel Sudarsky$^{1}$\thanks{\tt sudarsky@nucleares.unam.mx}}
\affil{$^{1}$Departamento de Gravitaci\'on y Teor\'ia de Campos, Instituto de Ciencias Nucleares, Universidad Nacional Aut\'onoma de M\'exico, A. Postal 70-543, Mexico City 045010, Mexico}
\affil{$^{2}$Department of Mathematics, University of York, York YO10 5DD, UK}
\affil{$^{3}$ Departamento de Ciencias, Instituto Tecnol\'ogico y de Estudios Superiores de Monterrey, Campus Santa Fe, Av. Carlos Lazo 100, Santa Fe, La Loma, \'Alvaro Obreg\'on, Mexico City 01389, Mexico}

\date{Revised version: \today \\ Published in JCAP \textbf{10}, 002 (2024)
doi:10.1088/1475-7516/2024/10/002}
	\maketitle
\begin{abstract}
	Given a Cauchy surface in a curved spacetime and a suitably defined quantum state on the CCR algebra of the Klein-Gordon quantum field on that surface, we show, by expanding the squared spacetime geodesic distance and the `$U$' and `$V$' Hadamard coefficients (and suitable derivatives thereof) in sufficiently accurate covariant Taylor expansions on the surface that the renormalized expectation value of the quantum stress-energy tensor on the surface is determined by the geometry of the surface and the first 4 time derivatives of the metric off the surface, in addition to the Cauchy data for the field's two-point function. This result has been anticipated in and is motivated by a previous investigation by the authors on the initial value problem in semiclassical gravity, for which the geometric initial data corresponds {\it a priori} to the metric on the surface and up to 3 time derivatives off the surface, but where it was argued that the fourth derivative can be obtained with aid of the field equations on the initial surface.
\end{abstract}

\section{Introduction}

Semiclassical gravity is a useful approximate theory which describes a regime lying between that where classical general relativity provides  a  good  effective  description,   and that  where  a  full  quantum gravity theory becomes  necessary. In this classical-quantum mixed theory, matter fields are treated as quantum fields that gravitate via the expectation value of their stress-energy tensor, sourcing geometry in the semiclassical Einstein field equations. It is expected that this theory accurately describes regimes in which gravity and quantum effects are important, but for which  scales of energy densities  or  curvatures   are far removed from the Planck scale.\footnote{Often one adds the  caveat   that  for   semiclassical  gravity to  be reliable,  the quantum  uncertainties  in the  energy momentum tensor  ought  to  be  small \cite{KuoFord},  however,  when trying to make   such a statement precise,   serious  obstacles  are encountered.  For instance,  the  case of the Minkowski  vacuum illustrates the fact that   requiring the uncertainties to  be  small  in comparison   with  the  corresponding  expectation values, does not  seem to be a universally desirable criterion. We therefore make no further comments in this direction.} 

In recent years, substantial mathematical advances have been made in order to understand the structural properties of semiclassical gravity, especially in symmetry-reduced situations or by studying simpler semiclassical models \cite{DiezTejedor:2011ci, Eltzner:2010nx, Gottschalk:2018kqt, GottRothSiem, Janssen:2022vux, Juarez-Aubry:2019jbw, Juarez-Aubry:2019jon, Juarez-Aubry:2020uim, Juarez-Aubry:2021tae, Juarez-Aubry:2021abq, JKMS2023:Programmatic, Juarez-Aubry:2023kvl, Meda:2020smb, Meda:2021zdw, Meda:2022hdh, Pinamonti:2010is, Pinamonti:2013wya, Pinamonti:2013zba, Sanders:2020osl}. We highlight \cite{JKMS2023:Programmatic}, in which the authors of this paper have put forth a number of conjectures on the well-posedness of the full (non-symmetry-reduced problem) in the case that matter is modelled by a scalar field satisfying the covariant Klein-Gordon equation.

Central to our conjectures in \cite{JKMS2023:Programmatic} is the \emph{surface Hadamard} condition. This condition characterizes the Hadamard property of a quantum state in terms of its behavior on an initial surface. It is worth emphasizing here, this condition is a ``semiclassical gravity" notion and not a ``quantum field theory in curved spacetime'' notion. (The precise definition is given in  \cite{JKMS2023:Programmatic} and depends on a prior notion of ``preliminary surface Hadamard'' which is a quantum field theory in curved spacetime notion but which will not be particularly relevant to the developments in this paper.)

But suffice it here to say that a state defined on the CCR algebra of a Riemannian 3-manifold, $({\cal C},\, h_{ab})$ together with three further (classical) symmetric two-tensors on that manifold will satisfy the surface Hadamard condition if the 3-metric, $h_{ab}$ and those three further symmetric tensors may be regarded as the restrictions to an initial surface of the metric of a solution to the semiclassical Einstein equations and also of the first 3 time-derivatives of the metric off that surface, while the state may be identified with the state in which the expectation value is taken on the right hand side of those equations.

One of the main conjectures in \cite{JKMS2023:Programmatic} is then (roughly) that for an initial such state and an initial such set of classical initial data that together satisfy the surface Hadamard condition, one can obtain a unique Hadamard solution to semiclassical gravity with that data. Note that the surface Hadamard condition subsumes the usual constraint equations of semiclassical gravity, since data that satisfies it must constitute {\it bona fide} initial data.

It turns out that this surface Hadamard condition involves a tower of further constraints on the initial data. A priori, this tower is infinite but presumably finite cut-offs of the tower will yield states (cf.\ \cite[Footnote 5]{JKMS2023:Programmatic}) whose 2-point functions differ from those of Hadamard states by continuous functions with finitely-many continuous derivatives. In any case, the minimal expectation is that reasonable initial data should be sufficiently {\it close} to being Hadamard, such that the expectation value of the renormalized stress-energy tensor on the initial value surface can be defined. Let us note, in passing, that this is a minimal requirement in order even for the usual constraint equations to make sense. 

The purpose of the present paper is to address the question of {\it how much geometric data is required on an initial surface in order to renormalize the stress-energy tensor  intrinsic to that initial surface.}   While the context for the ``surface Hadamard'' notion and for the initial value conjectures in \cite{JKMS2023:Programmatic} is semi-classical gravity and not quantum field theory in curved spacetime, this specific question can be analyzed in a fixed curved spacetime context and that is what we aim to do in this paper, taking the initial surface to be a Cauchy surface, $\mathcal{C}$, in a given fixed spacetime. We will discuss where the results of the present paper feed into those in \cite{JKMS2023:Programmatic} in Section \ref{sec:SurfaceHadamard} below.

What makes this question especially difficult is that the Hadamard form for the two-point function involves the (square of the) geodesic distance between two points and even when both those points are in a given initial surface, this depends, in general, on the geometry off the surface since the spacetime geodesics connecting two points in the surface will, in general, leave the surface.

This paper analyses this issue in detail and explores a path to deal with it, generating a recipe that yields the renormalized expectation value of the stress-energy tensor on $\mathcal{C}$ in terms of the intrinsic geometry of $\mathcal{C}$ and a choice of initial state on the CCR algebra on $\mathcal{C}$. This ``surface-defined" stress-energy tensor coincides with the restriction of the covariantly-defined {\`a la} Hadamard stress-energy tensor in the spacetime  to the surface $\mathcal{C}$. The strategy will be to obtain asymptotic expansions of the spacetime covariant quantities in the definition of the Hadamard condition (the spacetime geodesic distance, and the Hadamard coefficients) in terms of the spatial geodesic distance, the induced metric, $h_{ab}$, extrinsic curvature, $K_{ab} = (1/2) \pounds h_{ab} $, and derivatives of the extrinsic curvature off the surface, up to a  sufficiently high order, so that the spacetime Hadamard expresion  is  approximated well  enough to yield the correct  value of the stress-energy tensor on the surface.
 
The paper is organized as follows: Section \ref{sec:TRen} briefly reviews the elements of the renormalization for the stress-energy tensor of the Klein-Gordon field in curved spacetimes using Hadamard subtraction, and also serves the purpose of introducing the notation for the paper.  In Section \ref{sec:InitialData} we study how the Hadamard condition can be characterized on a Cauchy surface in terms of intrinsic geometric quantitites on the surface, namely its induced metric, $h_{ab}$, its extrinsic curvature, $K_{ab}$, and the bi-function, $\sigma$, which is defined to be half the square of the geodesic distance  intrinsic to the surface between pairs of points in the surface, together with suitable derivatives of all these geometric quantities both in the surface and normal to it.   We study 3-dimensional covariant Taylor expansions in terms of $\sigma$ and its derivatives in the surface, of the Hadamard coefficients $U$ and $V$ and of half the squared \emph{spacetime} geodesic distance, $\Sigma$, between pairs of points in the surface, to a desirable order of approximation that is sufficiently good so that the expectation value of the renormalized stress-energy tensor can be defined on the surface.

To achieve this, in Section \ref{sec:orderEstimates} we define a general notion of order-$n$ approximation for bi-scalars and show that if the approximation order for $\Sigma$, $U$ and $V$  are $6$, $4$ and $2$, respectively, then the renormalized stress-energy tensor computed using those approximations is exact. Then, in Section \ref{sec:ApproxGen} we present approximations that satisfy the conditions from Section \ref{sec:orderEstimates}. Namely, we first provide covariant Taylor series expansions for $U$ and $V$ in terms of $\Sigma$ and its derivatives.  Subsequently, we present  a Taylor expansion for $\Sigma$ and its first two normal derivatives in terms of $\sigma$ and its derivatives in the surface, yielding effectively expansions for $U$, $V$ (as well as $\Sigma$) in terms of $\sigma$ (and its derivatives).   We compute these expansions in Sections \ref{sec:HypSigma} and \ref{sec:SigmaDerivatives}.  We remark that the results for the derivatives of $\Sigma$ presented here fill an important gap in the arguments given in our previous paper \cite{JKMS2023:Programmatic} which makes conjectures about the well-posedness of the initial value problem in semiclassical gravity  (with a Klein-Gordon field) for initial states satisfying the surface Hadamard condition. The results that we find confirm that the minimal geometric data required to define the expectation value of the renormalized stress-energy tensor on an initial Cauchy surface are the induced metric on the surface and its first four derivatives off the surface as anticipated in \cite{JKMS2023:Programmatic}. In Section \ref{sec:SurfaceHadamard} we give more details on how those results feed into the arguments in \cite{JKMS2023:Programmatic}.  In Section \ref{sec:Discuss} we discuss these results and relate them with the general programme of semiclassical gravity presented in \cite{JKMS2023:Programmatic}. Some technical details and a brief review of the $3+1$ formalism we employ are included in the appendix.

We will be limiting our considerations to a globally hyperbolic spacetime $(M,g_{ab})$ foliated by Cauchy surfaces $\mathcal{C}_T$. We use abstract index notation for indices $\,_{abc\dots}$, and follow the conventions of \cite{WaldGR}. We will also make extensive use of Synge's coincidence-limit notation, whereby the coincidence limit of a bi-scalar $A$ is expressed as
\begin{equation}
	[A](\x) = \lim\limits_{\x'\to\x} A(\x,\x').
\end{equation}
 
 \section{The expectation value of the stress-energy tensor} 
 \label{sec:TRen}

The classical stress-energy tensor, $T_{ab}$, plays an essential role in General Relativity, representing the effects of matter on the spacetime metric that characterizes the \textit{gravitational field}. It acts as the \textit{source} in Einstein's equation,
	\begin{equation}
		R_{ab} - \frac{1}{2}R g_{ab} - \Lambda g_{ab} = 8 \pi G T_{ab}. \label{EinsteinEq}
	\end{equation} 
In semiclassical gravity, one replaces the right hand side of Eq. \eqref{EinsteinEq} by the expectation value of the quantum stress-energy tensor of matter fields evaluated in a suitable quantum state, $\omega$, which we denote by $\omega(T_{ab})$, and could, e.g., arise, when the quantum field is represented on a Hilbert space, $\cal H$, as $\langle\psi | \hat T_{ab}\psi\rangle$ for some vector $\psi$ in $\cal H$, or as $\mathrm{tr}(\hat \rho \hat T_{ab})$ for some density operator, $\hat\rho$, on $\cal H$. This results in the semiclassical Einstein equation,
\begin{equation}
	R_{ab} - \frac{1}{2}R g_{ab} - \Lambda g_{ab} = 8 \pi G \omega(T_{ab}). \label{SemiclassEinsteinEq}
\end{equation} 

An important technical aspect one must face is that in order to make sense of $\omega(T_{ab})$ (and therefore of Eq.\ \eqref{SemiclassEinsteinEq}) it is necessary to perform a renormalization procedure, since a na\"ive evaluation of the expectation value of a \textit{stress-energy observable} involves expectation values of products of operator-valued distributions at the same point and these are infinite or ill-defined. In the present paper we shall confine our attention to a matter model consisting of a single real scalar field obeying the covariant Klein-Gordon equation
\begin{equation}
	(g^{ab}\nabla_a \nabla_b - m^2 - \xi R ) \phi = 0, \label{KleinGordon}
\end{equation}
where $m$ is the field mass and $\xi$ a curvature coupling constant
and, for this,  we shall use the Hadamard renormalization procedure.  First we define the notion of  a Hadamard state $\omega$. This is characterized by a local condition on its  two-point function $G^+(\x, \x') = \omega(\phi(\x)\phi(\x'))$ (see \cite{KayWald} for more details). Namely, in a convex normal neighbourhood the two-point function $G^+$ is the sum of a smooth term, $W \in C^\infty(M \times M)$, and a bi-distribution, $H^\ell: C_0^\infty(M \times M) \to \mathbb{C}$, whose singular structure coincides with that of a locally constructed Hadamard parametrix of the Klein-Gordon operator. More specifically, the Wightman function of a Hadamard state admits the local representation
\begin{equation}
	\omega(\phi(\x)\phi(\x')) = H^\ell (\x, \x') + W (\x, \x'),  \label{TwoPointHad}
\end{equation}
where $H^\ell (\x, \x')$ has the form
\begin{align}
	H^\ell (\x, \x') &= \lim\limits_{\epsilon \to 0^+}\frac{1}{8\pi^2}\left(\frac{ U(\x, \x') }{\Sigma(\x,\x') + 2 i \epsilon (T(\x)-T(\x'))+\epsilon^2} \right.\nonumber\\
	&\qquad \left. + V(\x, \x') \log\left(\frac{\Sigma(\x,\x')+ 2i \epsilon (T(\x)-T(\x'))+\epsilon^2 }{\ell ^2}\right) \right), \label{HadamardSing}
\end{align}
where $T$ is an arbitrary time function, $\Sigma(\x,\x')$ is half the squared geodesic (spacetime) distance from $\x$ to $\x'$, $U$ and $V$ are smooth, symmetric bi-scalar functions, which are non-vanishing in the coincidence limit, and $\ell$ is a constant introduced for dimensional reasons that can be interpreted as a renormalization parameter. It has been shown that the Hadamard property ensures that the state satisfies a generalized positive-energy condition known as the micro-local spectrum condition \cite{Radzikowski1996}.

The bi-scalar functions $U$ and $V$ are determined by the so-called Hadamard recursion relations (see e.g. \cite{DecaniniFolacci}) as asymptotic series in $\Sigma$, and are independent of the details of the state. It can be shown \cite{DeWittBrehme} that $U(\x,\x') = \Delta^{1/2}(\x,\x')$, where
\begin{equation}
	\Delta(\x,\x') = - \frac{\operatorname{det}(-\nabla_\mu \nabla_{\nu'} \Sigma(\x,\x'))}{\sqrt{-g(\x)}\sqrt{-g(\x')}}, \label{VanVleckMoretteClosed}
\end{equation}
is the van Vleck-Morette determinant, with (un)primed derivative indices in \eqref{VanVleckMoretteClosed} acting in the (first) argument, and here $g$ is the determinant of the spacetime metric.

 Therefore, in the context of quantum field theory on a fixed curved spacetime, the singular structure for these states is known in terms of the spacetime metric, the mass and curvature-coupling parameters of the quantum field, allowing
one to construct and directly subtract $H^\ell$ in order to obtain the renormalized quadratic quantities in the stress-energy tensor.

The expectation value of the renormalized stress-energy tensor is then (here we use the formalism of \cite{DecaniniFolacci} which is technically different from the earlier formalism of \cite{Wald78}.)
\begin{equation}
	\omega(T_{ab} (x)) = \lim\limits_{\x' \rightarrow \x} \left(\mathscr{T}_{ab} \left[G^+ (\x,\x') - H^\ell(\x,\x') \right] - \frac{1}{8 \pi} g_{ab} V_1(\x, \x')\right) + \Theta_{ab} (\x), \label{TRenWald}
\end{equation}
where $\mathscr{T}_{ab}$ is given by
\begin{align}
	\mathscr{T}_{ab} &= (1-2\xi) g_{b}\,^{b'}\nabla_a \nabla_{b'} +\left(2\xi - \frac{1}{2}\right) g_{ab}g^{cd'} \nabla_c \nabla_{d'} - \frac{1}{2} g_{ab} m^2 \nonumber\\
	&\quad + 2\xi \Big[ - g_{a}\,^{a'} g_{b}\,^{b'} \nabla_{a'} \nabla_{b'} + g_{ab} g^{c d}\nabla_c \nabla_d + \frac{1}{2}G_{ab} \Big], \label{PointSplitOp}
\end{align}
with $g_a{}^{a'}$ the parallel-transport propagator \cite{DeWittBrehme}. Further, $V_1$ is the second term in the Hadamard (asymptotic) expansion
\begin{equation}
 	V(\x, \x') = \sum_{n=0}^\infty V_n (\x,\x') \Sigma^n(\x, \x'), \label{vSigma}
 \end{equation}
and in the limit equals
\begin{align}
	\lim_{\x'\to\x} V_1(\x,\x') &= \frac{1}{8} m^4 + \frac{1}{4} \left(\xi - \frac{1}{6}\right) m^2 R - \frac{1}{24}\left(\xi - \frac{1}{5} \right) \Box R \nonumber \\
	& \quad + \frac{1}{8} \left(\xi - \frac{1}{6} \right)^2 R^2 - \frac{1}{720} R_{ab} R^{ab} + \frac{1}{720} R_{abcd} R^{abcd}.
	\label{v1c}
\end{align}

The term $\Theta$ is a local ambiguity of the form $\alpha m^4 g_{ab} + \beta m^2 G_{ab} + \gamma H_{ab}$ with $H_{ab}$ a linear combination of terms coming from the scalar actions (cf.\  \cite[section 6.2]{birrell_davies_1982}) with quadratic curvature Lagrangian functions\footnote{A term coming from the Lagrangian $L_3= R_{abcd} R^{abcd}$ is not linearly independent by the Gauss-Bonnet formula.} $L_1=R^2$ and $L_2=R_{ab} R^{ab}$. Absorbing the $\alpha$ and $\beta$ coefficients by suitable renormalizations of the cosmological and Newton constants, we can write
\begin{align} \label{GeneralTheta}
	\Theta_{ab} = & (2 \alpha_1+\alpha_2) \nabla_a \nabla_b R-\frac{1}{2}\left(4 \alpha_1+\alpha_2\right)\Box R g_{ab}-\alpha_2 \Box R_{ab} \nonumber \\
	& + \frac{\alpha_2}{2} R^{cd}R_{cd} g_{ab} +\frac{\alpha_1}{2}R^2 g_{ab} -2 \alpha_1 R R_{ab}-2 \alpha_2 R^{cd}R_{cadb},
\end{align}
where $\alpha_1$ and $\alpha_2$ are dimensionless, arbitrary parameters,which to the best of our knowledge are not yet fixed by experiments. Note that if we chose the values of $\alpha_1$ and $\alpha_2$ to get rid of the fourth order terms in the semiclassical Einstein equations at a fixed renormalisation scale, $\ell$, any change of scale, $\ell \to \ell'$, in the Hadamard renormalisation prescription will yield a term of the general form of $\Theta_{ab}$.

Let us note that, if the spacetime metric is already given, the only input required from quantum theory to compute $\omega(T_{ab})$ up to a local, geometric term, is $G^+ (\x,\x')$  and its first two derivatives. In particular, if one is restricted to a Cauchy surface $\mathcal{C}$, and Gaussian coordinates adapted to it such that the normal direction is parallel to the coordinate vector $(\partial_t)^a$, this translates into the need to have been given initial data for the field in terms of the bi-scalar functions $G^+ (\x,\x')$, $\partial_t G^+(\x,\x')$, $\partial_{t'}G^+(\x,\x')$ and $\partial_{t}\partial_{t'}G^+(\x,\x')$, defined on $\mathcal{C}$.

\section{The Hadamard condition on an initial surface} \label{sec:InitialData}

In this section we construct a sufficiently precise approximation of the Hadamard condition on an initial Cauchy surface, so that we are able to define the expectation value of the renormalized stress-energy tensor on the surface in terms of initial data. The strategy is to obtain suitable approximations of the geodesic distance and the Hadamard coefficients in terms of intrinsic objects defined on the initial surface viewed as a Riemannian manifold, namely in terms of the initial Riemannian metric tensor on the surface, the surface geodesic distance (which differs from the spacetime geodesic distance in general), together with the extrinsic curvature and higher order derivatives of the metric off the surface. This allows one to obtain a suitable bi-distribution on the initial surface that coincides sufficiently well with the induced Hadamard bi-distribution on the surface, and that can be used for the renormalization purposes.

The main result of the section is stated below in Theorem \ref{th:SigmaSurfaceApprox}, which provides the sought after definition of the expectation value of the renormalized stress-energy in terms of certain bi-scalars $\tilde \Sigma$, $\tilde U$ and $\tilde V$ together with bi-scalars corresponding to their derivatives off the surface. We then proceed to explicitly construct these bi-scalars in terms of the intrinsic surface quantities in the subsequent subsections. These results put together achieve the goal of the section.

Furthermore, these results will help us build a bridge between the construction presented in this paper and the notion of the surface Hadamard condition introduced in our previous work \cite{JKMS2023:Programmatic}.
 See the remarks before Conjecture 1 in  section \ref{sec:SurfaceHadamard}.

\subsection{Estimates for approximate initial data} \label{sec:orderEstimates}

The Hadamard coefficients and half the squared geodesic distance $\Sigma(\x,\x')$ in a background spacetime can be computed in a convex normal neighbourhood, in principle, if we are given the full metric description of that spacetime.

However, it is necessary to analyse explicitly how many \textit{time} derivatives of the metric given on a Cauchy surface $\mathcal{C}$ need to be known in order to compute the two-point function of the field and its time-derivatives off the surface to a sufficient degree of accuracy that the stress-energy tensor on  $\mathcal{C}$ can be calculated.

Our approach will be to derive a covariant Taylor expansion of the spacetime half squared geodesic distance, $\Sigma(\x,\x')$, in terms of the half-squared geodesic distance, $\sigma(\x, \x')$, intrinsic to the surface $\mathcal{C}$ (viewed as a 3-dimensional Riemannian manifold in its own right) and its derivatives in the surface, and then to use this to compute Taylor series for the Hadamard coefficients, in terms of their complete $3+1$ projections, so that every object appearing in the resulting prescription, is given explicitly in terms of the metric and its derivatives on the \textit{initial surface} $\mathcal{C}$.

But first, it is necessary to estimate how good such approximations need to be in order to compute $\omega(T_{ab}(\x))$ at $\mathcal{C}$. That is, we propose that instead of $U(\x,\x')$, $V(\x,\x')$ and $\Sigma(\x,\x')$, we construct corresponding \textit{approximate} bi-scalar functions $\tilde{U}(\x,\x')$, $\tilde V (\x,\x')$ and $\tilde{\Sigma}(\x,\x')$ to create an approximate Hadamard singular parametrix such that substracting it from the two point function at the surface yields a bi-scalar function regular enough to compute a renormalized expectation value of the stress-energy tensor at the initial surface by the prescription \eqref{TRenWald} with the Hadamard part substituted by said approximation. 

Let us first define the general notion of an approximation for a bi-scalar function:
\begin{defn}
	Let $A\in C^m (M\times M)$ a bi-scalar function, where $m\in \N$. We say $B\in C^n(M\times M)$ is an order-$n$ split-point approximation (we will often  omit the ``split-point" conditional and simply say ``an order-$n$ approximation") of $A$, with $n \leq m$, if and only if for every $k \in \mathbb{N}_0$, $0\leq k\leq n$, 
	\begin{equation}
		\lim\limits_{\x'\rightarrow \x} \nabla_{a_k} \dots \nabla_{a_1} B(\x,\x') = \lim\limits_{\x'\rightarrow \x} \nabla_{a_k} \dots \nabla_{a_1} A(\x,\x'). \label{ApproximationDerivatives}
	\end{equation}

If such order-$n$ approximation can be identified as a truncated covariant Taylor series of $A$ of the form
\begin{align}\label{CovTaylorGen}
	B(\x,\x') =& \overset{0}{A}(\x) + \overset{1}{A}^b(\x) \nabla_{b} \Sigma (\x,\x') + \frac{1}{2!}\overset{2}{A}^{b_2 b_1}(\x) \nabla_{b_2} \Sigma (\x,\x')\nabla_{b_1} \Sigma (\x,\x')+ \dots \nonumber \\
	&\, + \frac{1}{n!} \overset{n}{A}^{b_{n}\dots b_1}(\x)  \nabla_{b_n} \Sigma (\x,\x') \dots \nabla_{b_1} \Sigma (\x,\x'),
\end{align}
we refer to it as an order-$n$ covariant Taylor truncation of $A$.
\end{defn}

As usual, we can obtain expressions for the coefficients of the Taylor expanison by taking into account the fundamental equation satisfied by $\Sigma(\x,\x')$ \cite{DeWittBrehme},
\begin{equation}
	g^{ab} (\nabla_a \Sigma) (\nabla_b \Sigma) = 2 \Sigma, \label{pdeSigma}
\end{equation}
which lead to the vanishing of the coincidence limits for $\Sigma$, $\nabla_a \Sigma$ and $\nabla_a \nabla_b \nabla_c \Sigma$, and the nonvanishing limit (see appendix \ref{AppSigma})
\begin{equation}\label{SigmaLimits02}
	\lim\limits_{\x'\rightarrow \x} \nabla_a \nabla_b \Sigma(\x,\x') = g_{ab}(\x).
\end{equation}

Then, the first coefficients of the Taylor expansion above are
\begin{subequations}\label{CovTaylorCoefs02}
	\begin{align}
		\overset{0}{A}(\x) &= \lim\limits_{\x'\rightarrow \x} A(\x,\x'),\\
		\overset{1}{A}_{a}(\x) &= \left(\lim\limits_{\x'\rightarrow \x} \nabla_a A(\x,\x')\right) - \nabla_a \overset{0}{A}(\x),\\
		\overset{2}{A}_{a_1 a_0}(\x) &= \left(\lim\limits_{\x'\rightarrow \x} \nabla_{a_1} \nabla_{a_0} A\right) (\x)- \nabla_{a_1} \nabla_{a_0} \overset{0}{A}(\x)-2 \nabla_{(a_1} \overset{1}{A}_{a_0)}(\x),\\
		\overset{3}{A}_{a_2 a_1 a_0} (\x) &= \left(\lim\limits_{\x'\rightarrow \x} \nabla_{(a_2} \nabla_{a_1} \nabla_{a_0)} A\right) (\x)-\nabla_{(a_2} \nabla_{a_1} \nabla_{a_0)} \overset{0}{A}(\x) -3 \nabla_{(a_2} \nabla_{a_1}\overset{1}{A}_{a_0)}(\x) \nonumber \\
		&\qquad - 3 \nabla_{(a_2} \overset{2}{A}_{a_1 a_0)}(\x) . \label{Taylor3}
	\end{align}
\end{subequations}

To recapitulate our obective in terms of the definitions we have just introduced, we want to determine the general order of approximation for $\Sigma(\x,\x')$, $U(\x,\x')$ and $V(\x,\x')$ that permits to construct an approximate Hadamard parametrix good enough for computing the exact renormalized stress-energy tensor on the surface by the prescription \eqref{TRenWald}.

Consider the result in \cite{DecaniniFolacci} for the limit of \eqref{TRenWald}, 
\begin{align}
	\omega(T_{ab}(\x)) = \frac{1}{2(2\pi)^2} &\Big(-\overset{2}{w}_{ab}+\frac{1}{2}(1-2\xi) \nabla_{a}\nabla_{b}\overset{0}{w} + g_{ab} \left(\xi -\frac{1}{4}\right) \Box \overset{0}{w} +\xi \overset{0}{w} R_{ab}\nonumber\\
	&\quad - g_{ab}[V_1]\Big) + \Theta_{ab}, \label{TrenDecaFol}
\end{align}
where $\overset{0}{w}$ and $\overset{2}{w}_{ab}$ are the corresponding terms for the covariant Taylor expansion
\begin{equation}
	W(\x,\x') = \overset{0}{w}(\x) + \overset{1}{w}^{a} (\x) \nabla_a \Sigma(\x,\x') + \overset{2}{w}^{ab} (\x) \nabla_{a} \Sigma(\x,\x') \nabla_{b} \Sigma(\x,\x') +  O(\Sigma^{3/2}), \label{TaylorW}
\end{equation}
with 
\begin{equation}
	W(\x,\x') \equiv G_{\psi}^+(\x,\x') - H^{\ell}(\x,\x'). \label{WDef}
\end{equation}

As we have already anticipated, our approach will be to propose approximations for the coefficients $U(\x,\x')$ and $V(\x,\x')$ of the Hadamard term \eqref{HadamardSing} as well as half the squared geodesic distance $\Sigma(\x,\x')$, and we will show that there is a sufficient condition on the order of approximation for each of these objects so that the Hadamard property is approximated sufficiently well that the renormalization procedure using it will give the exact stress-energy tensor. We begin by proving the following
\begin{thm} \label{Th:OrderEstimates}
	Let $\x$ and $\x'$ be points contained in a convex normal neighborhood $\mathcal{D}\subset M$, where the Hadamard singular part is given by \eqref{HadamardSing} and half the square geodesic distance $\Sigma(\x,\x')$ as well as the Hadamard coefficients $U(\x,\x')$, $V(\x,\x')$ are well defined, and let $\tilde{\Sigma}(\x,\x')$, $\tilde{U}(\x,\x')$ and $\tilde{V}(\x,\x')$ be corresponding approximations of order $n_{\Sigma}\geq 6$, $n_{U}\geq 4$ and $n_{V}\geq 2$, respectively. Then,
	\begin{equation}\label{NullDeltaH}
		\left[\nabla_a \nabla_b \left(H^\ell - \tilde{H}^\ell\right)\right] = 0,
	\end{equation}
	with
	\begin{equation}
		\tilde{H}^\ell (\x,\x') \equiv \frac{1}{2(2\pi)^2} \left(\frac{\tilde U(\x,\x')}{\tilde{\Sigma}(\x,\x')} + \tilde{V} (\x,\x') \log \left(\frac{\tilde{\Sigma}(\x,\x')}{\ell^2}\right) \right).\label{ApproximateHadamard}
	\end{equation}
\end{thm}
\begin{proof}
	Let us consider a coordinate chart $\{x^\mu\}$ on $M$ so that we deal with coordinate components, i.e., 
	\begin{equation}
		E_{\mu\nu} \equiv \nabla_\mu \nabla_\nu \left(H^\ell - \tilde{H}^\ell\right). \label{Emunu}
	\end{equation}
	Define the following bi-scalars,
	\begin{subequations}\label{deltas}
		\begin{align}
			\,^\Sigma \delta (\x,\x')&\equiv \Sigma(\x,\x') - \tilde{\Sigma}(\x,\x'),\\
			\,^U \delta(\x,\x') &\equiv U(\x,\x') - \tilde{U}(\x,\x'),\\
			\,^V \delta(\x,\x') &\equiv V(\x,\x') - \tilde{V}(\x,\x'),
		\end{align}
	\end{subequations}
	that represent a measure of the \textit{error} corresponding to each approximate function. Then, by construction, all up to (including) the second derivative of $\,^V \delta(\x,\x')$, the fourth derivative of $\,^U \delta(\x,\x')$ and the sixth derivative of $\,^\Sigma \delta (\x,\x')$ will vanish in the coincidence limit. To explicitly implement these properties, we invert the definitions \eqref{deltas} so that we make the substitutions
	\begin{subequations}\label{Invdeltas}
		\begin{align}
			\tilde{\Sigma}(\x,\x') &= \Sigma(\x,\x') - \,^\Sigma \delta (\x,\x'),\\
			\tilde{U}(\x,\x') &= U(\x,\x') - \,^U \delta(\x,\x'),\\
			\tilde{V}(\x,\x') &\equiv V(\x,\x') - \,^V \delta(\x,\x'),
		\end{align}
	\end{subequations}
	in \eqref{Emunu}, obtaining 
	\small{
		\begin{equation} 
			E_{\mu\nu} = \mathcal{E}_{\mu\nu} + \left(\log \left(\frac{\Sigma -(\,^\Sigma\delta)}{l^2}\right)-\log \left(\frac{\Sigma }{l^2}\right)\right) \nabla _{\mu }\nabla _{\nu }V-\log \left(\frac{\Sigma -(\,^\Sigma\delta)}{l^2}\right) \nabla _{\mu }\nabla _{\nu }(\,^V\delta), \label{DDdelta}
	\end{equation}}
	where 
	\small{
		\begin{align}
			\mathcal{E}_{\mu\nu} \equiv& \Bigg\lbrace (\,^\Sigma\delta)^3 \Big(-2 \Delta^{1/2} \nabla _{\mu }\Sigma \nabla _{\nu }\Sigma - \Sigma ^2 (\nabla _{\mu }\nabla _{\nu }\Delta^{1/2}+\nabla _{\mu }\Sigma \nabla _{\nu }V+\nabla _{\nu }\Sigma \nabla _{\mu }V+V \nabla _{\mu }\nabla _{\nu }\Sigma)\nonumber\\
			&\qquad\qquad +\Sigma (\nabla _{\mu }\Sigma \nabla _{\nu }\Delta^{1/2}+\Delta^{1/2} \nabla _{\mu }\nabla _{\nu }\Sigma +\nabla _{\nu }\Sigma (\nabla _{\mu }\Delta^{1/2}+V \nabla _{\mu }\Sigma))\Big)\nonumber\\
			&+\Sigma (\,^\Sigma\delta)^2 \Big(\Sigma ^2 \big(\nabla _{\mu }\nabla _{\nu }(\,^{U}\delta)+\nabla _{\nu }(\,^\Sigma\delta) (\nabla _{\mu }V-\nabla _{\mu }(\,^V\delta))+\nabla _{\mu }(\,^\Sigma\delta) (\nabla _{\nu }V-\nabla _{\nu }(\,^V\delta))\nonumber\\
			&\qquad\qquad\qquad +\nabla _{\mu }\Sigma (\nabla _{\nu }(\,^V\delta)+2 \nabla _{\nu }V)+\nabla _{\nu }\Sigma (\nabla _{\mu }(\,^V\delta)+2 \nabla _{\mu }V)\nonumber\\
			&\qquad\qquad\qquad +(V-(\,^V\delta)) \nabla _{\mu }\nabla _{\nu }(\,^\Sigma\delta)+((\,^V\delta)+2 V) \nabla _{\mu }\nabla _{\nu }\Sigma +2 \nabla _{\mu }\nabla _{\nu }\Delta^{1/2}\big)\nonumber\\
			&\qquad\qquad-3 \Sigma \big(\nabla _{\mu }\Sigma \nabla _{\nu }\Delta^{1/2}+\Delta^{1/2} \nabla _{\mu }\nabla _{\nu }\Sigma +\nabla _{\nu }\Sigma (\nabla _{\mu }\Delta^{1/2}+V \nabla _{\mu }\Sigma)\big)\nonumber\\
			&\qquad\qquad+6 \Delta^{1/2} \nabla _{\mu }\Sigma \nabla _{\nu }\Sigma \Big)\nonumber\\
			&+ \Sigma ^2 (\,^\Sigma\delta)\Big(-\Sigma ^2 \big(2 \nabla _{\nu }(\,^\Sigma\delta) (\nabla _{\mu }V-\nabla _{\mu }(\,^V\delta))+2 \nabla _{\mu }(\,^\Sigma\delta) (\nabla _{\nu }V-\nabla _{\nu }(\,^V\delta))\nonumber\\
			&\qquad\qquad\qquad +2 \nabla _{\mu }\nabla _{\nu }(\,^{U}\delta)+2 (V-(\,^V\delta)) \nabla _{\mu }\nabla _{\nu }(\,^\Sigma\delta)\nonumber\\
			&\qquad\qquad\qquad +\nabla _{\mu }\Sigma (2 \nabla _{\nu }(\,^V\delta)+\nabla _{\nu }V)+\nabla _{\nu }\Sigma (2 \nabla _{\mu }(\,^V\delta)+\nabla _{\mu }V)\nonumber\\
			&\qquad\qquad\qquad +(2 (\,^V\delta)+V) \nabla _{\mu }\nabla _{\nu }\Sigma +\nabla _{\mu }\nabla _{\nu }\Delta^{1/2}\big)\nonumber\\
			&\qquad\qquad+\Sigma \big(\nabla _{\mu }(\,^\Sigma\delta) (\nabla _{\nu }\Delta^{1/2}-\nabla _{\nu }(\,^{U}\delta))+(\Delta^{1/2}-(\,^{U}\delta)) \nabla _{\mu }\nabla _{\nu }(\,^\Sigma\delta)\nonumber\\
			&\qquad\qquad\qquad+\nabla _{\nu }(\,^\Sigma\delta) (-\nabla _{\mu }(\,^{U}\delta)+((\,^V\delta)-V) (\nabla _{\mu }(\,^\Sigma\delta)-\nabla _{\mu }\Sigma)+\nabla _{\mu }\Delta^{1/2})\nonumber\\
			&\qquad\qquad\qquad+\nabla _{\nu }\Sigma (\nabla _{\mu }(\,^{U}\delta)+(V-(\,^V\delta)) \nabla _{\mu }(\,^\Sigma\delta)+((\,^V\delta)+2 V) \nabla _{\mu }\Sigma +2 \nabla _{\mu }\Delta^{1/2})\nonumber\\
			&\qquad\qquad\qquad+\nabla _{\mu }\Sigma (\nabla _{\nu }(\,^{U}\delta)+2 \nabla _{\nu }\Delta^{1/2})+((\,^{U}\delta)+2 \Delta^{1/2}) \nabla _{\mu }\nabla _{\nu }\Sigma \big)\nonumber\\
			&\qquad\qquad-6 \Delta^{1/2} \nabla _{\mu }\Sigma \nabla _{\nu }\Sigma \Big)\nonumber\\
			&+\Sigma ^3 \Big(2 \big((\,^{U}\delta)-\Delta^{1/2}\big) \nabla _{\mu }(\,^\Sigma\delta) (\nabla _{\nu }(\,^\Sigma\delta)-\nabla _{\nu }\Sigma)\nonumber\\
			&\qquad\qquad\qquad +2 \nabla _{\mu }\Sigma \big((\Delta^{1/2}-(\,^{U}\delta)) \nabla _{\nu }(\,^\Sigma\delta)+(\,^{U}\delta) \nabla _{\nu }\Sigma \big)\nonumber\\
			&\qquad\qquad+\Sigma \big(\nabla _{\mu }(\,^\Sigma\delta) (\nabla _{\nu }(\,^{U}\delta)-\nabla _{\nu }\Delta^{1/2})+((\,^{U}\delta)-\Delta^{1/2}) \nabla _{\mu }\nabla _{\nu }(\,^\Sigma\delta)\nonumber\\
			&\qquad\qquad\qquad +\nabla _{\nu }(\,^\Sigma\delta) (\nabla _{\mu }(\,^{U}\delta)-((\,^V\delta)-V) (\nabla _{\mu }(\,^\Sigma\delta)-\nabla _{\mu }\Sigma)-\nabla _{\mu }\Delta^{1/2})\nonumber\\
			&\qquad\qquad\qquad -\nabla _{\nu }\Sigma (\nabla _{\mu }(\,^{U}\delta)+(V-(\,^V\delta)) \nabla _{\mu }(\,^\Sigma\delta)+(\,^V\delta) \nabla _{\mu }\Sigma)\nonumber\\
			&\qquad\qquad\qquad -\nabla _{\nu }(\,^{U}\delta) \nabla _{\mu }\Sigma -(\,^{U}\delta) \nabla _{\mu }\nabla _{\nu }\Sigma \big)\nonumber\\
			&\qquad\qquad+\Sigma ^2 (\nabla _{\mu }\nabla _{\nu }(\,^{U}\delta)+\nabla _{\nu }(\,^\Sigma\delta) (\nabla _{\mu }V-\nabla _{\mu }(\,^V\delta))+\nabla _{\mu }(\,^\Sigma\delta) (\nabla _{\nu }V-\nabla _{\nu }(\,^V\delta))\nonumber\\
			&\qquad\qquad\qquad +(V-(\,^V\delta)) \nabla _{\mu }\nabla _{\nu }(\,^\Sigma\delta)+\nabla _{\nu }(\,^V\delta) \nabla _{\mu }\Sigma +\nabla _{\mu }(\,^V\delta) \nabla _{\nu }\Sigma \nonumber\\
			&\qquad\qquad\qquad\qquad +(\,^V\delta) \nabla _{\mu }\nabla _{\nu }\Sigma)\Big) \Bigg\rbrace \Bigg/ \left\lbrace \Sigma ^3 (\,^\Sigma\delta-\Sigma)^3 \right\rbrace.
	\end{align}} 
	Note that $\mathcal{E}_{\mu\nu}$ is already arranged so that both numerator and denominator vanishes in the coincidence limit. Assume we take said limit to be along a trajectory with tangent vector $\xi^a$. Let us introduce the notation
	\begin{equation}
		\nabla_\xi^n f(\x) \equiv \xi^{a_n} \nabla_{a_n} (\xi^{a_{n-1}} \nabla_{a_{n-1}}(\dots \xi^{a_{1}} \nabla_{a_{1}}f(\x) \dots)), \label{dirder}
	\end{equation}
	for the $n^{\text{th}}$ directional derivative along the vector field $\xi^a$. Then we can successively use L'H\^opital rule with derivatives $\nabla_\xi^n$ on $\mathcal{E}_{\mu\nu}$ to obtain, at the eleventh stage,
	\begin{align}\label{LHopitalLimit}
		[\mathcal{E}_{\mu\nu}] = -\frac{1}{90}\Big\lbrace &2 [\nabla_\xi^6\,\,^\Sigma\delta](g_{\mu \nu}-6 \xi_{\mu } \xi_{\nu }) +24 \xi_{(\mu } [\nabla_\xi^5 \nabla _{\nu) } \,^\Sigma\delta]-15 [\nabla_\xi^4 \nabla _{\mu } \nabla _{\nu } \,^\Sigma\delta]\nonumber\\
		&-15 [\nabla_\xi^4\,\,^{U}\delta] (g_{\mu \nu} - 4 \xi_{\mu } \xi_{\nu }) -120 \xi_{(\mu } [\nabla_\xi^3 \nabla_{\nu) } \,^{U}\delta] +90 [\nabla_\xi^2 \nabla _{\mu }\nabla _{\nu }\,^{U}\delta] \nonumber\\
		&+90 [\nabla_\xi^2\,\,^V\delta] (g_{\mu \nu}-2 \xi_{\mu } \xi_{\nu })+360 \xi_{(\mu } [\nabla_\xi\nabla _{\nu) } \,^V\delta]\Big\rbrace,
	\end{align}
	which vanishes in the coincidence limit, given that the numerator comprises terms multiplying limits of derivatives of $\delta$'s of order $n_{\Sigma}\geq 6$, $n_{U}\geq 4$ and $n_{V}\geq 2$. Note that this limit is unique and trajectory-independent given that it vanishes for all $\xi^a$ thanks to the order of approximation of $\tilde{\Sigma}$, $\tilde{U}$ and $\tilde V$. 
	
	Putting together the logarithmic terms in \eqref{DDdelta} that multiply $\nabla _{\mu }\nabla _{\nu }V$, we have
	\begin{equation}
		\left[\log \left(\frac{\Sigma -(\,^\Sigma\delta)}{\Sigma}\right) \nabla _{\mu }\nabla _{\nu }V\right] = \log \left(\left[\frac{\Sigma -(\,^\Sigma\delta)}{\Sigma}\right]\right) [\nabla _{\mu }\nabla _{\nu }V] = 0,
	\end{equation}
	by direct application of L'H\^opital rule to the argument of the logarithm. The remaining logarithmic term in \eqref{DDdelta} vanishes given that $[\nabla_{\mu }\nabla _{\nu }(\,^V\delta)]$ vanishes for $n_{V}\geq 2$. Therefore,
	\begin{equation}
		[E_{\mu\nu}] = 0.
	\end{equation}
\end{proof}

\begin{cor} \label{Th:ExactD2W}
	Let $\x,\x' \in \mathcal{D}\subset M$, $\tilde{\Sigma}(\x,\x')$, $\tilde{U}(\x,\x')$, $V(\x,\x')$ and $\tilde{H}^\ell$ as in Theorem \ref{Th:OrderEstimates}, and
	\begin{equation}
		\tilde{W}(\x,\x') \equiv G_{\psi}^+(\x,\x')-\tilde H^\ell(\x,\x'),  \label{TildeW}
	\end{equation}
	then,
	\begin{subequations}
		\begin{align}
			\overset{0}{w}(\x)&= [\tilde{W}](\x), \label{W0approx}\\
			\overset{2}{w}_{ab}(\x)&= [\nabla_a \nabla_b \tilde{W}](\x),\label{W2approx}
		\end{align}
	\end{subequations}
	where $\overset{0}{w}(\x)$ and $\overset{2}{w}_{ab}(\x)$ are the zero and second order Taylor coefficients of the regular part of the Wightman two point function, cf. \eqref{WDef} and \eqref{TaylorW}.
\end{cor}
\begin{proof}
	Consider the \textit{difference}
	\begin{align*}
		\,^W\delta(\x,\x') &\equiv W(\x,\x') - \tilde{W}(\x,\x')\\
		&= G_{\psi}^+(\x,\x')- H^\ell(\x,\x') - G_{\psi}^+(\x,\x') + \tilde H^\ell(\x,\x') \\
		&= \tilde H^\ell(\x,\x') -H^\ell(\x,\x').
	\end{align*}
	Putting this in terms of $(\,^\Sigma \delta)$, $(\,^U\delta)$ and $(\,^V\delta)$ by means of \eqref{Invdeltas}, we have
	\begin{equation}
		\,^W\delta = \frac{1}{2(2\pi)} \left(\frac{(U-(\,^U\delta))\Sigma-U(\Sigma-(\,^\Sigma\delta))}{\Sigma(\Sigma - (\,^\Sigma\delta))}+ V \log\left(\frac{\Sigma - (\,^\Sigma\delta)}{\Sigma}\right) - (\,^V\delta) \log\left(\frac{\Sigma-(\,^\Sigma\delta)}{\ell^2}\right)\right).
	\end{equation}
	In the coincidence limit, the logarithmic terms vanish just like in Theorem \ref{Th:OrderEstimates}, while the quotient term requires to apply L'H\^opital rule four times to obtain
	\begin{equation}
		[\,^W\delta] = \frac{1}{2(2\pi)^2} \left(\frac{6 [\nabla_\xi^2\,^U\delta]-[\nabla_\xi^4 \,^\Sigma\delta]}{-6}\right)=0. \label{DeltaW}
	\end{equation}
	Therefore, for $\x\in\mathcal{D}$,
	\begin{equation}
		[\tilde{W}] (\x)= \overset{0}{w}(\x).
	\end{equation}
	Now consider
	\begin{equation*}
		\nabla_a \nabla_b W(\x,\x') - \nabla_a \nabla_b \tilde{W} (\x,\x') = \nabla_a \nabla_b \left(\tilde H^\ell(\x,\x') -H^\ell(\x,\x')\right).
	\end{equation*}
	Then, by Theorem \ref{Th:OrderEstimates}, in the coincidence limit for all $\x' \in \mathcal{D}$,
	\begin{equation}
		\overset{2}{w}_{ab}(\x) - [\nabla_a \nabla_b \tilde{W}] (\x) = 0,
	\end{equation}
	where it follows \eqref{W2approx}.
\end{proof}

With these results at hand, it is straightforward to prove the following

\begin{cor}\label{Cor:ExactRSET}
	Let $\x$ and $\x'$ be points contained in a convex normal neighborhood $\mathcal{D}\subset M$, and $\tilde{\Sigma}(\x,\x')$, $\tilde{U}(\x,\x')$ and $\tilde{V}(\x,\x')$ be corresponding approximations of order $n_{\Sigma}\geq 6$, $n_{U}\geq 4$ and $n_{V}\geq 2$ for half the squared geodesic distance $\Sigma(\x,\x')$ and the Hadamard coefficients $U(\x,\x')$ and $V(\x,\x')$, respectively. Let $\tilde{H}^\ell(\x,\x')$ and $\tilde{W}(\x,\x')$ biscalars given by \eqref{ApproximateHadamard} and \eqref{TildeW}, respectively. Then, the renormalized stress-energy tensor in $\mathcal{D}$ is \textit{exactly} given by
	\begin{align}
		\omega(\tilde T_{ab}(\x)) \equiv \frac{1}{2(2\pi)^2} &\Big(-\overset{2}{\tilde{w}}_{ab}+\frac{1}{2}(1-2\xi) \nabla_{a}\nabla_{b}\overset{0}{\tilde{w}} + g_{ab} \left(\xi -\frac{1}{4}\right) \Box \overset{0}{\tilde{w}} +\xi \overset{0}{\tilde{w}} R_{ab}\nonumber\\
		&\quad - g_{ab}[V_1]\Big) + \Theta_{ab}, \label{TrenByApprox1}
	\end{align}
	where
	\begin{subequations}
		\begin{align}
			\overset{0}{\tilde{w}}(\x) &\equiv [\tilde{W}](\x),\\
			\overset{2}{\tilde{w}}_{ab}(\x) &\equiv [\nabla_a \nabla_b \tilde{W}](\x).
		\end{align}
	\end{subequations}
\end{cor}
\begin{proof}
By taking the difference in expectation values between the exact and approximate stress energy tensor, and using the formula \eqref{TrenDecaFol}
we have that
	\begin{align*}
		 \omega(T_{ab}(\x))-\omega(\tilde{T}_{ab}(\x)) = \frac{1}{2(2\pi)^2} &\Big(\overset{2}{\tilde{w}}_{ab}-\overset{2}{w}_{ab}+\frac{1}{2}(1-2\xi) \nabla_{a}\nabla_{b}(\overset{0}{w}-\overset{0}{\tilde{w}}) \\
		&\quad + g_{ab} \left(\xi -\frac{1}{4}\right) \Box( \overset{0}{w}-\overset{0}{\tilde{w}}) +\xi R_{ab}(\overset{0}{w}-\overset{0}{\tilde{w}}) \Big),
	\end{align*}
	where, by Corollary \ref{Th:ExactD2W}, we have
	\begin{equation}
			\omega(T_{ab}(\x)) = \omega(\tilde{T}_{ab}(\x)).
	\end{equation}
\end{proof}

Note that by demmanding the approximations $\tilde{\Sigma}$, $\tilde{U}$ and $\tilde V$ to be of respective orders $n_{\Sigma}\geq 6$, $n_{U}\geq 4$ and $n_{V}\geq 2$, we are imposing a condition that ensures that in the most general case, the exact expectation value for the renormalized stress-energy tensor is obtained, which might be too stringent. That is because in special situations, for example where some symmetry is present, lower order approximations for $\Sigma$, $U$ and $V$ might still yield a correct result for the expectation value of the renormalized stress-energy tensor. However, as we do not wish to restrict consideration to any particular situation, we take this to be a minimal condition for the general case.

Now we proceed to implement this in the context of initial data.

To do so, we need to write these results in terms of data on a given spacelike surface $\mathcal{C}$ so we will employ the $3+1$ formalism (see Appendix \ref{App:3p1}). As initial data for the field will be given as two point functions defined on $\mathcal{C}$, the following notion is useful,
\begin{defn}
	Let $A\in C^m(\mathcal{C}\times \mathcal{C})$ be a surface bi-scalar function. We say $B\in C^n(\mathcal{C}\times \mathcal{C})$ is an order-$n$ surface approximation of $A$, with $n \leq m$, if and only if for every $k\in \N$, $0\leq k\leq n$, 
	\begin{equation}
		\lim\limits_{\x'\rightarrow \x} D_{a_k} \dots D_{a_1} B(\x,\x') = \lim\limits_{\x'\rightarrow \x} D_{a_k} \dots D_{a_1} A(\x,\x'),
	\end{equation}
	where $D_a$ is the derivative operator associated to the surface induced metric $h_{ab}$ such that $D_a h_{bc}=0$, which is given by
	\begin{equation}
		D_{a} T^{b_k \dots b_1 }\,_{c_l \dots c_1} = h^{b_k}\,_{d_k} \dots h^{b_1}\,_{d_1} h_{c_l}\,^{e_l} \dots h_{c_1}\,^{e_1} h_{a}\,^{f} \nabla_f T^{d_k \dots d_1 }\,_{e_l \dots e_1}.
	\end{equation}
\end{defn}

In particular, $A$ might be the restriction to $\mathcal{C}$ of some bi-scalar function defined on $M\times M$, and we will imply the said restriction whenever the term surface approximation is used for a bi-scalar defined on spacetime. Note that normal (\textit{time}) derivatives are not defined for surface approximations as they are defined only in terms of points within a given spatial surface. The following result will help us to connect the notion of (spacetime) approximation with surface approximations:

\begin{lemma}\label{Lemma:SurfaceReconstruction}
	Let $A\in C^m(M\times M)$ be a bi-scalar function and $B\in C^m(M\times M)$ be an order-$k$ approximation for $A$, with $k \leq m$. Then,  given a time function, $T$, for which $T= T_0$ defines a spatial Cauchy surface, $\mathcal{C}$, with future-directed normal $n^a$, $B$ and its $k$ covariant derivatives, restricted to the spatial surface $\mathcal{C}$ can be reconstructed in terms of surface-approximations for $A$, $\dbar A$, $\dbar^2 A$, \dots, $\dbar^k A$ of order $k, k-1, \dots, 0$, respectively, where the operator $\dbar$ is defined in Eq. \eqref{NormalDerivative}.
\end{lemma}
\begin{proof}
	Consider the $3+1$ expansion of the coincidence limits of the first $k$ covariant derivatives of $A$ (cf. Appendix \ref{App:3p1}):
	\begin{subequations}\label{3p1Coinc}
	\begin{align}
		[\nabla_a A] &= [D_a A + n_a \dbar A] = [D_a A] + n_a [\dbar A],\label{3p1CoincDA}\\
		[\nabla_{a_1} \nabla_{a_0} A] &= [D_{a_1} D_{a_0} A] + K_{a_1 a_0} [\dbar A] + n_{a_0} \left([D_{a_1} \dbar A] + K_{a_1}\,^{b} [D_b A] \right) \nonumber \\
		& \quad + n_{a_1} \left([D_{a_0} \dbar A] + K_{a_0}\,^{b} [D_b A] \right) + n_{a_1} n_{a_0} \left([\dbar^2 A] + u^b [D_b A] \right),  \label{3p1CoincD2A}\\
		&\vdots \nonumber\\
		[\nabla_{a_{k-1}} \dots \nabla_{a_0} A] &= [D_{a_{k-1}} \dots D_{a_0} A] + \dots + n_{a_{k-1}} \dots n_{a_0} \Big( [\dbar^k A] + \dots \Big),
	\end{align}
\end{subequations}
where $K_{ab}$ is the extrinsic curvature and $u_a \equiv -n^b \nabla_b n_a$. Let  $\mathcal{B}$, $\,^1\mathcal{B}$, $\,^2 \mathcal{B}$, \dots, $\,^k \mathcal{B}$  be the surface approximations for $A$, $\dbar A$, $\dbar^2 A$, \dots, $\dbar^k A$ of order $k, k-1, \dots, 0$, respectively. Then, substituting the definition of surface approximation for $\mathcal{B}$, $\,^1 \mathcal{B}$, $\,^2 \mathcal{B}$, \dots, $\,^k \mathcal{B}$, as well as the definition of order $k$ approximation for $B$, we have
	\begin{align*}
		[\nabla_a B] &= [D_a \mathcal{B}] + n_a [\,^1 \mathcal{B}],\\
		[\nabla_{a_1} \nabla_{a_0} B] &= [D_{a_1} D_{a_0} \mathcal{B}] + K_{a_1 a_0} [\,^1 \mathcal{B}] + n_{a_0} \left([D_{a_1} \,^1 \mathcal{B}] + K_{a_1}\,^{b} [D_b \mathcal{B}] \right) \\
		& \quad + n_{a_1} \left([D_{a_0} \,^1 \mathcal{B}] + K_{a_0}\,^{b} [D_b \mathcal{B}] \right) + n_{a_0} n_{a_1} \left([\,^2 \mathcal{B}] + u^b [D_b \mathcal{B}] \right),\\
		&\vdots\\
		[\nabla_{a_{k-1}} \dots \nabla_{a_0} B] &= [D_{a_{k-1}} \dots D_{a_0}\mathcal{B}] + \dots + n_{a_{k-1}} \dots n_{a_0} \Big( [\,^k \mathcal{B}] + \dots \Big).
	\end{align*}
	Therefore, the limits that define $B$ are given in terms of the surface approximations for $A$ and each one of its first $n$ normal derivatives, of orders $n, n-1, \dots, 0$.	
\end{proof}
In other words, an order-$n$ (spacetime) approximation amounts to the information of $n+1$ hypersurface approximations. However, in the case of half the squared geodesic distance $\Sigma$, we will not require all this information in order to compute $\omega(T_{ab})$ on $\mathcal{C}$. We can see this is the case with the help of the following lemma:
\begin{lemma}\label{lemmaSurface}
	Let $\x$ and $\x'$ be points on a spacelike surface $\mathcal{C}$ of the spacetime $(M, g_{ab})$ and contained in a convex normal neighbourhood. Then, one can choose the unique geodesic trajectory along $\mathcal{C}$ going from $\x'$ to $\x$ in order to compute coincidence limits for functions and tensors, so that spacetime directional derivatives coincide  with $\mathcal{C}$-tangent directional derivatives.
\end{lemma}
\begin{proof}
	Let $\xi^a$ be the unit vector in $\mathcal{C}$ tangent to the unique surface-geodesic going from $\x'$ to $\x$. Proceeding by induction, we verify
	\begin{align*}
		\nabla_\xi^1 f &= \xi^{a_1} \nabla_{a_1} f\\
		&= \xi^{b_1} h_{b_1}\,^{a_1} \nabla_{a_1} f\\
		&= \xi^{a_1} D_{a_1} f,
	\end{align*}
	and take the inductive hypothesis to be
	\begin{equation}
		\label{InductiveHyp}
		\nabla_\xi^{n} f =\xi^{a_n} \xi^{a_{n-1}} \dots \xi^{a_1}D_{a_n} D_{a_{n-1}} \dots D_{a_1} f.
	\end{equation}
	
	The identity for surface-geodesic vectors\footnote{Cf. equation \eqref{3p1GeodesicId} in Appendix \ref{App:3p1}.},
	\begin{equation}
		\xi^a \nabla_a \xi^b = (K_{ac}\xi^a \xi^c)n^b,
	\end{equation}
	along with the inductive hypothesis, yield
	\begin{align*}
		\nabla_\xi^{n+1} f &= \xi^{a_{n+1}} \nabla_{a_{n+1}} (\xi^{a_n} \xi^{a_{n-1}} \dots \xi^{a_1}D_{a_n} D_{a_{n-1}} \dots D_{a_1} f)\\
		&= (\xi^{a_{n+1}} \nabla_{a_{n+1}} \xi^{a_n}) \xi^{a_{n-1}} \dots \xi^{a_1} D_{a_n} D_{a_{n-1}} \dots D_{a_1} f \\
		&\quad + \xi^{a_{n}} (\xi^{a_{n+1}} \nabla_{a_{n+1}} \xi^{a_{n-1}}) \xi^{a_{n-2}}\dots \xi^{a_1} D_{a_n} D_{a_{n-1}} \dots D_{a_1} f + \dots \\
		&\quad + \xi^{a_n} \xi^{a_{n-1}} \dots \xi^{a_1} \xi^{a_{n+1}} \nabla_{a_{n+1}} D_{a_n} D_{a_{n-1}} \dots D_{a_1} f \\
		&= (\xi^{a_{n+1}} K_{a_{n+1}c} \xi^c n^{a_n}) \xi^{a_{n-1}} \dots \xi^{a_1} D_{a_n} D_{a_{n-1}} \dots D_{a_1} f \\
		&\quad + \xi^{a_{n}} (\xi^{a_{n+1}} K_{a_{n+1}c} \xi^c n^{a_{n-1}}) \xi^{a_{n-2}}\dots \xi^{a_1} D_{a_n} D_{a_{n-1}} \dots D_{a_1} f + \dots \\
		&\quad + \xi^{a_{n+1}} \xi^{a_n} \xi^{a_{n-1}} \dots \xi^{a_1} \nabla_{a_{n+1}} D_{a_n} \dots D_{a_1} f \\
		&= \xi^{a_{n+1}} \xi^{a_n} \xi^{a_{n-1}} \dots \xi^{a_1} D_{a_{n+1}} D_{a_n} \dots D_{a_1} f.
	\end{align*}
\end{proof}
With these results at hand, we can prove the following result:
\begin{thm}\label{th:SigmaSurfaceApprox}
	Let $\x$ and $\x'$ points contained on a spacelike surface $\mathcal{C}$ of the spacetime $(M,g_{ab})$ and also contained in a convex normal neighborhood $\mathcal{D}$. Let $\tilde{U}(\x,\x')$, $\dbar \tilde{U}(\x,\x')$, $\dbar \dbar' \tilde{U}(\x,\x')$ and $\dbar^2 \tilde{U}(\x,\x')$ be approximations of order $n_{U}\geq 4$, $n_{\dbar U}\geq 3$, $n_{\dbar \dbar' U}\geq 2$, $n_{\dbar^2 U}\geq 2$ respectively, for the Hadamard coefficient $U(\x,\x')$ and its corresponding normal derivatives, as well as for the Hadamard coefficient $V(\x,\x')$ let $\tilde{V}(\x,\x')$, $\dbar \tilde{V}(\x,\x')$, $\dbar \dbar' \tilde{V}(\x,\x')$ and $\dbar^2 \tilde{V}(\x,\x')$ be its corresponding approximations of order $n_{V}\geq 2$, $\dbar n_{V}\geq 1$, $n_{\dbar\dbar' V}\geq 0$ and $n_{\dbar^2 V}\geq 0$, and let $\tilde{\Sigma}$, $\dbar \tilde{\Sigma}$, $\dbar \dbar'\tilde{\Sigma}$ and $\dbar^2 \tilde{\Sigma}$ be respectively order $n_{\Sigma}\geq 6$, $n_{\dbar \Sigma}\geq 5$, $n_{\dbar \dbar' \Sigma} \geq 4$ and $n_{\dbar^2\Sigma} \geq 4$ approximations for $\Sigma(\x,\x')$, $\dbar \Sigma (\x,\x')$, $\dbar \dbar' \Sigma (\x,\x')$ and $\dbar^2 \Sigma(\x,\x')$. Let $\phi(\x)$ a scalar quantum field satisfying the Klein-Gordon equation \eqref{KleinGordon} in a Hadamard state $\psi$ with initial data given by $G^+_\psi(\x,\x')$, $\dbar G^+_\psi (\x,\x')= -\frac{1}{N(\x)}\partial_t G^+_\psi(\x,\x')$, $\dbar' G^+_\psi (\x',\x) = -\frac{1}{N(\x')}\partial_{t'} G^+_\psi(\x,\x')$ and $\dbar \dbar' G^+_\psi (\x',\x)= \frac{1}{N(\x)N(\x')} \partial_{t} \partial_{t'} G^+_\psi(\x,\x')$, and let $\tilde{H}^\ell(\x,\x')$ and $\tilde{W}(\x,\x')$ biscalars given by \eqref{ApproximateHadamard} and \eqref{TildeW}, respectively. Then, the corresponding renormalized stress-energy tensor in $\mathcal{C}$ is \textit{exactly} given by
	\begin{align}
		\omega(T_{ab}(\x)) \equiv \frac{1}{2(2\pi)^2} &\Big(-\overset{2}{\tilde{w}}_{ab}+\frac{1}{2}(1-2\xi) \mathcal{W}_{ab} + g_{ab} \left(\xi -\frac{1}{4}\right) g^{cd}\mathcal{W}_{cd} +\xi \overset{0}{\tilde{w}} R_{ab}\nonumber\\
		&\quad - g_{ab}[V_1]\Big) + \Theta_{ab}, \label{TrenByApprox2}
	\end{align}
	where  
	\begin{subequations}\label{InitialDataTerms1}
		\begin{align}
			\overset{0}{\tilde{w}}(\x) &\equiv [\tilde{W}](\x), \label{TildeWSurf}\\
			\mathcal{W}_{ab}(\x) &\equiv D_{a} D_{b}[\tilde{W}](\x)  + 2 K_{ab} [\dbar \tilde{W}] + 2 n_{(a} \left(2 D_{b)} [\dbar \tilde{W}] + K_{b)}\,^{c} D_c [\tilde{W}] \right) \nonumber \\
			& \quad + 2 n_{a} n_{b} \left([\dbar \dbar' \tilde{W}] + [\dbar^2 \tilde{W}] + u^c D_c [\tilde{W}] \right), \label{CurlyW}\\
			\overset{2}{\tilde{w}}_{ab}(\x) &\equiv [D_{a} D_{b} \tilde{W}] + K_{a b} [\dbar \tilde{W}] + 2 n_{(a} \left([D_{b)} \dbar \tilde{W}] + K_{b)}\,^{c} [D_c \tilde{W}] \right) \nonumber \\
			& \quad + n_{a} n_{b} \left([\dbar^2 \tilde{W}] + u^c [D_c \tilde{W}] \right), \label{Tilde2WSurf}
		\end{align}
		and
		\begin{align}
			u^a &\equiv -n^b\nabla_b n^a, \\
			\tilde{W}(\x,\x')&\equiv G^+_\psi(\x,\x') - \tilde{H}^{\ell}(\x,\x'), \label{SurfaceTildeW}\\
			\dbar \tilde{W} (\x,\x') &\equiv \dbar G^+_\psi(\x,\x') - \dbar \tilde{H}^{\ell}(\x,\x'), \label{dbarTildeW}\\
			\dbar \dbar' \tilde{W} (\x,\x') &\equiv \dbar \dbar' G^+_\psi(\x,\x') - \dbar \dbar' \tilde{H}^{\ell}(\x,\x'), \label{dprimedbarTildeW}\\
			\dbar \tilde{H}^\ell(\x,\x') &\equiv \frac{1}{2(2\pi)^2} \Bigg(\frac{1}{\tilde{\Sigma}(\x,\x')} \left(\dbar \tilde{U}(\x,\x') +\left( \tilde{V}(\x,\x')-\tilde{U}(\x,\x')\right)\dbar\tilde{\Sigma}(\x,\x')\right) \nonumber\\
			&\qquad\qquad\qquad + \dbar \tilde{V}(\x,\x') \log\left(\tilde{\Sigma}(\x,\x')/\ell\right)\Bigg), \\
			\dbar \dbar' \tilde{H}^\ell &\equiv \frac{1}{2(2\pi)^2} \Bigg( -\frac{\dbar'\tilde{\Sigma}(\x,\x')}{\tilde{\Sigma}^2(\x,\x')} \left(\dbar \tilde{U}(\x,\x') +\left( \tilde{V}(\x,\x')-\tilde{U}(\x,\x')\right)\dbar\tilde{\Sigma}(\x,\x')\right) \nonumber\\
			&\qquad\qquad\qquad + \frac{1}{\tilde{\Sigma}(\x,\x')} \Big(\dbar \dbar' \tilde{U}(\x,\x') + \left( \tilde{V}(\x,\x') - \tilde{U}(\x,\x') \right) \dbar \dbar' \tilde{\Sigma}(\x,\x') \nonumber\\
			&\qquad\qquad\qquad\qquad + \left( \dbar'\tilde{V}(\x,\x') - \dbar'\tilde{U}(\x,\x') \right)\dbar\tilde{\Sigma}(\x,\x') + \dbar V(\x,\x') \dbar'\tilde{\Sigma} \Big) \nonumber\\
			&\qquad\qquad\qquad + \dbar \dbar' V(\x,\x') \log\left(\tilde{\Sigma}(\x,\x')/\ell\right) \Bigg), \\ 
			\dbar^2 \tilde{W} (\x,\x') &\equiv D^b D_b \tilde W (\x,\x')+ K (\x)\dbar \tilde{W}(\x,\x') - u^b(\x) D_b \tilde{W}(\x,\x') \nonumber\\
			&\quad -(m^2 + \xi R(\x))\tilde{W}(\x,\x') + 6 V_1(\x,\x')\nonumber\\
			&\quad + 2 \left((D^b V_1(\x,\x'))D_b\Sigma (\x,\x') - \dbar V_1 (\x,\x') \dbar \Sigma(\x,\x')\right). \label{ddbarTildeW}
		\end{align}
	\end{subequations}
\end{thm}
\begin{proof}
	Due to the fact both points of all the biscalar functions lie on $\mathcal{C}$, by lemma \ref{lemmaSurface}, equation \eqref{DeltaW} will read
	\begin{equation}
		[\,^W\delta] = \frac{1}{2(2\pi)^2} \left(\frac{6 [D_\xi^2\,^U\delta]-[D_{\xi}^4 \,^\Sigma\delta]}{-6}\right)=0, \label{DeltaWSurface}
	\end{equation}
	where we have used the notation, for any $f\in C^m(\mathcal{C}\times\mathcal{C})$, $m\geq n$,
	\begin{equation}
		D_\xi^n f \equiv \xi^{a_n} \xi^{a_{n-1}} \dots \xi^{a_1}D_{a_n} D_{a_{n-1}} \dots D_{a_1} f.
	\end{equation}
	 Also, the coincidence limit of
	 \begin{equation}
	 	\Upsilon_a \equiv \nabla_a (H^\ell - \tilde{H}^\ell), \label{upsilon}
	 \end{equation}
	 which involve the L'H\^opital rule applied eight times, obtaining
	 \begin{align}
	 	[\Upsilon_{\mu}] = \frac{1}{30}\Big( 5 [\nabla_\xi^4 \nabla_{\mu} \,^\Sigma \delta ] - 4 \xi_{\mu} [\nabla_\xi^5 \,^\Sigma\delta] + 20 \xi_{\mu} [\nabla_\xi^3 \,^{U}\delta]- 30 [\nabla_\xi^2 \nabla_{\mu} \,^{U}\delta] - 60 \xi_{\mu} [\nabla_\xi \,^V \delta] \Big),
	 \end{align}
	 and by lemma \ref{lemmaSurface}, and expanding in $3+1$ form, reads
	 \begin{align}
	 	[\Upsilon_{\mu}] = \frac{1}{30}\Big( &\, 5 [D_\xi^4 D_{\mu} \,^\Sigma \delta ] - 4 \xi_{\mu} [D_\xi^5 \,^\Sigma\delta] + 5 n_{\mu} [D_\xi^4 \dbar \,^\Sigma \delta ]\nonumber\\
	 	&+ 20 \xi_{\mu} [D_\xi^3 \,^{U}\delta]- 30 [D_\xi^2 D_{\mu} \,^{U}\delta]- 30 n_{\mu} [D_\xi^2 \dbar \,^{U}\delta] - 60 \xi_{\mu} [D_\xi \,^V \delta] \Big), \label{LimitUpsilon}
	 \end{align}
	 Similarly for \eqref{LHopitalLimit}, and using the $3+1$ expansions \eqref{3p1CoincDA} and \eqref{3p1CoincD2A}, we have that
	\begin{align}\label{LHopitalLimitSurface}
		[\mathcal{E}_{\mu\nu}] = -\frac{1}{90}\Big\lbrace &2 [D_\xi^6\,\,^\Sigma\delta](g_{\mu \nu}-6 \xi_{\mu } \xi_{\nu }) +24 \xi_{(\mu } [D_\xi^5 D_{\nu) } \,^\Sigma\delta]-15 [D_\xi^4 D_{\mu } D_{\nu } \,^\Sigma\delta]\nonumber\\
		&+ 24 \xi_{(\mu } n_{\nu) } [ D_\xi^5 \dbar\,^\Sigma\delta] -15 \Big( K_{\mu \nu} [D_\xi^4 \dbar \,^\Sigma\delta] +2 n_{(\mu} \left([D_\xi^4 D_{\nu)} \dbar \,^\Sigma\delta] \right.  \nonumber \\
		&+ \left. K_{\nu)}\,^{b} [D_\xi^4 D_b \,^\Sigma\delta] \right)+ n_{\mu} n_{\nu} \left([D_\xi^4 \dbar^2 \,^\Sigma\delta] + u^b [D_\xi^4 D_b \,^\Sigma\delta] \right) \Big) \nonumber\\
		&-15 [D_\xi^4\,\,^{U}\delta] (g_{\mu \nu} - 4 \xi_{\mu } \xi_{\nu }) -120 \xi_{(\mu } [D_\xi^3 D_{\nu) } \,^{U}\delta] +90 [D_\xi^2 \nabla_{\mu }\nabla _{\nu } \,^{U}\delta] \nonumber\\
		&-120 \xi_{(\mu } n_{\nu)} [D_\xi^3\dbar \,^{U}\delta] + 90 \Big( K_{\mu\nu}[D_\xi^2\dbar\,^{U}\delta] + 2 n_{(\mu} \left( [D_\xi^2 D_{\nu)}\dbar \,^{U}\delta] \right. \nonumber\\
		&+\left. K_{\nu)}\,^b [D_\xi^2 D_b \,^{U}\delta] \right) + n_{\mu}n_{\nu} \left( [D_\xi^2 \dbar^2 \,^{U}\delta] + u^b [D_\xi^2 D_b \,^{U}\delta] \right)\Big) \nonumber\\
		&+90 [D_\xi^2\,\,^V\delta] (g_{\mu \nu}-2 \xi_{\mu } \xi_{\nu })+360 \xi_{(\mu } \left( [D_\xi D_{\nu)} \,^V\delta] + n_{\nu)} [D_\xi \dbar \,^V\delta]\right). 
	\end{align}
	By hypothesis we have an order $6$ surface approximation for $\Sigma$, an order $5$ approximation for $\dbar \Sigma$, an order $4$ approximation for $U$ and a second order apprixmation for $V$, and therefore all of the terms in brackets in \eqref{DeltaWSurface}, \eqref{LimitUpsilon} and \eqref{LHopitalLimitSurface} vanish identically, so it follows that for all $\x \in \mathcal{C}$,
	\begin{subequations}
		\begin{align}
			[\tilde{W}](\x) &= [W](\x), \label{SurfaceEq0w} \\
			[\nabla_a \tilde{W}] (\x)  &= [\nabla W ](\x),\label{SurfaceEq1w}\\
			[\nabla_a \nabla_b \tilde{W}] (\x) &= [\nabla_a \nabla_b W] (\x) ,\label{SurfaceEq2w}
		\end{align}
	\end{subequations}
	with $\tilde{W}(\x,\x')$ defined by \eqref{SurfaceTildeW}, also
	\begin{equation}
		\nabla_a \tilde{W} (\x,\x') \equiv D_a \tilde{W} (\x,\x')  + n_a \dbar \tilde{W} (\x,\x'), \label{3p1TildeW}
	\end{equation}
	where $\dbar \tilde{W} (\x,\x')$ is defined by \eqref{dbarTildeW}, and 
	\begin{align}
		\nabla_a \nabla_b \tilde{W}(\x,\x') \equiv& D_{a} D_{b} \tilde{W}(\x,\x') + K_{a b}(\x) \dbar \tilde{W}(\x,\x') \nonumber \\
		& + 2 n_{(a}(\x) \left(D_{b)} \dbar \tilde{W}(\x,\x') + K_{b)}\,^{c}(\x) D_c \tilde{W}(\x,\x') \right) \nonumber \\
		& + n_{a}(\x) n_{b}(\x) \left(\dbar^2 \tilde{W}(\x,\x') + u^c(\x) D_c \tilde{W}(\x,\x') \right),
	\end{align}
	where $\dbar^2 \tilde{W}$ is defined by \eqref{ddbarTildeW}.
	 Equation \eqref{SurfaceEq0w} implies then that in $\mathcal{C}$,
	 \begin{equation}
	 	\overset{0}{\tilde{w}} = \overset{0}{w}. \label{ApproxW0ID}
	 \end{equation}
	 
	  Comparing \eqref{3p1TildeW} with the $3+1$ expansion of $\nabla_a W$, and with \eqref{SurfaceEq1w}, we have that
	 \begin{equation}
	 	[\dbar \tilde{W}] = [\dbar W].
	 \end{equation}
	Similarly, from the $3+1$ expansion of the field equation satisfied by $W$ (cf. \cite{DecaniniFolacci}),
	\begin{equation}
		(g^{ab}\nabla_a \nabla_b - m^2 - \xi R (\x)) W(\x,\x') = - 6 V_1 (\x,\x') - 2 g^{ab}(\x) (\nabla_b V_1 (\x,\x'))(\nabla_a \Sigma(\x,\x')) + \mathcal{O}(\Sigma),
	\end{equation}
	we can solve for $\dbar^2 W$ to obtain
	\begin{align}
		\dbar^2 W (\x,\x') &= D^b D_b W (\x,\x')+ K (\x)\dbar W (\x,\x') - u^b(\x) D_b W (\x,\x') \nonumber\\
		&\quad -(m^2 + \xi R(\x)) W(\x,\x') + 6 V_1(\x,\x') \nonumber\\
		&\quad + 2 \left((D^b V_1(\x,\x'))D_b\Sigma (\x,\x') - \dbar V_1 (\x,\x') \dbar \Sigma(\x,\x')\right) + \mathcal{O}(\Sigma), 
	\end{align}
	and substracting \eqref{ddbarTildeW} we obtain in the coincidence limit
	\begin{equation}
		[\dbar^2 \tilde{W}] = [\dbar^2 W], \label{LimDbar2W}
	\end{equation}
	which implies that
	\begin{align}
		\overset{2}{w}_{ab} - \overset{2}{\tilde{w}}_{ab} =& [D_a D_b W] - [D_a D_b \tilde{W}] + K_{ab} \left([\dbar W] - [\dbar \tilde{W}]\right) \nonumber\\
		&+ 2 n_{(a} \left([D_{b)}\dbar W] -[D_{b)}\dbar \tilde{W}]  + K_{b)}\,^c\left([D_{c}W]-[D_{c}\tilde{W}]\right)\right) \nonumber\\
		&+ n_a n_b \left([\dbar^2 W] - [\dbar^2 \tilde{W}] + u^c \left([D_c W] - [D_c \tilde{W}]\right)\right),
	\end{align}
	vanishes identically in $\mathcal{C}$, ie., 
	\begin{equation}
		\overset{2}{\tilde{w}}_{ab} = \overset{2}{w}_{ab}. \label{ApproxW2ID}
	\end{equation}
	
	Next, by \eqref{SecondDerPhi3p1} and \eqref{NormalDerivativesExchange}, 
	\begin{align}
		\nabla_{a} \nabla_{b} \overset{0}{w} &= (K_{a b} \dbar \overset{0}{w} + D_{a} D_{b} \overset{0}{w}) + 2 n_{(a} \left(D_{b)} \dbar \overset{0}{w} + K_{b)}\,^{c} D_c \overset{0}{w} \right) + n_{a} n_{b} \left(\dbar^2 \overset{0}{w} + u^c D_c \overset{0}{w} \right), 
	\end{align}
	and subtracting \eqref{CurlyW}, we have
	\begin{align}
		\nabla_{a} \nabla_{b} \overset{0}{w} - \mathcal{W}_{ab} &= D_{a} D_{b} \{\overset{0}{w}-[\tilde{W}]\} + K_{a b} \{\dbar \overset{0}{w}-2 [\dbar \tilde{W}]\}\nonumber\\
		&\quad + 2 n_{(a} \left(D_{b)} \{\dbar \overset{0}{w}-2[\dbar \tilde{W}]\} + K_{b)}\,^{c} D_c \{\overset{0}{w}-[\tilde{W}]\} \right) \nonumber \\
		& \quad + n_{a} n_{b} \left(\dbar^2 \overset{0}{w}-2([\dbar \dbar' \tilde{W}]+[\dbar^2 \tilde{W}]) - u^c D_c [\tilde{W}] + u^c D_c \{ \overset{0}{w} - [\tilde{W}] \} \right),
	\end{align}
	and by \eqref{SurfaceEq0w},
	\begin{align}
		\nabla_{a} \nabla_{b} \overset{0}{w} - \mathcal{W}_{ab} &= \left(K_{a b} + 2 n_{(a} D_{b)}\right) \left(\dbar \overset{0}{w}-2[\dbar \tilde{W}]\right)\nonumber\\
		&\quad + n_{a} n_{b}\left(\dbar^2 \overset{0}{w}-2([\dbar \dbar' \tilde{W}]+[\dbar^2 \tilde{W}])- u^c[D_c \tilde{W}]\right) . \label{DiffDDWW}
	\end{align}
	According to Synge's rule,
	\begin{equation}
		\nabla_a [W] = [\nabla_{a'} W] + [\nabla_{a} W],
	\end{equation}
	and due to the symmetry of $W$, $[\nabla_{a'} W] = [\nabla_a W]$, ie.,
	\begin{equation}
		\nabla_a [W] = 2 [\nabla_a W]. \label{SyngeDW}
	\end{equation}
	Taking the contraction with $-n^a$ yield
	\begin{equation}
		\dbar \overset{0}{w} = \dbar [W] = 2 [\dbar W],
	\end{equation}
	and therefore the first line in \eqref{DiffDDWW} vanishes in $\mathcal{C}$ due to \eqref{SurfaceEq1w}.
	Note that we can write $[\nabla_{a} \nabla_{b'} W]$ in terms of $\overset{0}{w}$ and $\overset{2}{w}_{ab}$ by writing down the Taylor Series expansion for $W$ up to second order,
	\begin{equation}
		W(\x,\x') = \overset{0}{w}(\x) + \overset{1}{w}^c (\x) \nabla_c\Sigma(\x,\x') + \frac{1}{2}\overset{2}{w}^{cd}(\x) \nabla_c\Sigma(\x,\x') \nabla_d\Sigma(\x,\x')+ \mathcal{O} (\Sigma^{3/2}), 
	\end{equation}
	taking a derivative on $\x'$ followed by a derivative on $\x$, and computing the coincidence limit to obtain
	\begin{align}
		[\nabla_{a} \nabla_{b'} W ] &= \nabla_a\overset{1}{w}^c  [\nabla_{c} \nabla_{b'} \Sigma]  + \overset{1}{w}^c [\nabla_{a} \nabla_c \nabla_{b'} \Sigma]  \nonumber\\
		&\qquad + \frac{1}{2} \overset{2}{w}^{cd} \left( [\nabla_c \nabla_{b'}\Sigma] [\nabla_{a}\nabla_d\Sigma] +  [\nabla_{a}\nabla_c\Sigma] [\nabla_d \nabla_{b'}\Sigma] \right).\label{DprimeDW2}
	\end{align}
	Applying Synge's rule on $[\nabla_{a} \Sigma]=0$ yields
	\begin{align*}
		\nabla_b [\nabla_{a} \Sigma]&= [\nabla_{b} \nabla_{a} \Sigma] + [\nabla_{a} \nabla_{b'} \Sigma]\\
		0 &= g_{ba} + [\nabla_{a} \nabla_{b'} \Sigma] ,
	\end{align*}
	ie.,
	\begin{equation}
		[\nabla_{a}\nabla_{b'} \Sigma] = - g_{ab},
	\end{equation}
	and proceeding the same on $[\nabla_{a} \nabla_{c} \Sigma]$ yield
	\begin{align}
		\nabla_{b} [\nabla_{a} \nabla_{c} \Sigma] &=  [\nabla_{b} \nabla_{a} \nabla_{c} \Sigma] +  [\nabla_{a} \nabla_{c} \nabla_{b'}  \Sigma] \nonumber\\
		0 &=  [\nabla_{a}\nabla_{c} \nabla_{b'} \Sigma].
	\end{align}
	Substituting these results back in  \eqref{DprimeDW2} we obtain
	\begin{equation}
		[\nabla_{a} \nabla_{b'} W ] = -\nabla_a\overset{1}{w}_b - \overset{2}{w}_{ab},
	\end{equation}
	and considering that
	\begin{equation}
		\overset{1}{w}_a = - \frac{1}{2} \nabla_a \overset{0}{w},
	\end{equation}
	then
	\begin{equation}
		[\nabla_{a} \nabla_{b'} W ]= \frac{1}{2}\nabla_a\nabla_b \overset{0}{w}- \overset{2}{w}_{ab}.
	\end{equation}
	This in turn means that
	\begin{equation}
		\nabla_a\nabla_b \overset{0}{w} = 2 \left([\nabla_{a} \nabla_{b'} W ] + \overset{2}{w}_{ab}\right).
	\end{equation}
	Then, by contracting twice with $n^a$ we obtain
	\begin{equation}
		\dbar^2 \overset{0}{w} + u^c D_c \overset{0}{w} = 2 \left([\dbar \dbar' W ] + u^c [D_c W] + n^a n^b \overset{2}{w}_{ab}\right).
	\end{equation}
	Taking into account that $[D_c W] = (1/2)D_c \overset{0}{w}$, and that $\overset{2}{w}_{ab}$ has a $3+1$ expansion analog to that of $\overset{2}{\tilde w}_{ab}$ given by \eqref{Tilde2WSurf}, then we have
	\begin{equation}
		\dbar^2 \overset{0}{w} = 2 \left([\dbar \dbar' W ] + [\dbar^2 W] \right) + u^c D_c\overset{0}{w} .
	\end{equation}
	Substituting back in \eqref{DiffDDWW}, we have
	\begin{equation}
		\mathcal{W} = \nabla_a \nabla_b \overset{0}{w}, \label{CurlyWId}
	\end{equation}
	for all points in $\mathcal{C}$.
	Therefore, subtracting \eqref{TrenByApprox2} form \eqref{TrenDecaFol} we have
	\begin{align}
		\omega(T_{ab}(\x)) - \omega(\tilde T_{ab}(\x)) = \frac{1}{2(2\pi)^2} &\Big(\overset{2}{\tilde{w}}_{ab}-\overset{2}{w}_{ab}+\frac{1}{2}(1-2\xi) \left(\nabla_a \nabla_b \overset{0}{w} - \mathcal{W}_{ab}\right) \nonumber\\
		&+ g_{ab} \left(\xi -\frac{1}{4}\right) g^{cd}\left(\nabla_c \nabla_d \overset{0}{w} - \mathcal{W}_{cd}\right) +\xi \left(\overset{0}{w} - \overset{0}{\tilde{w}}\right) R_{ab}\Big),
	\end{align}
	which vanishes by \eqref{ApproxW0ID}, \eqref{ApproxW2ID} and \eqref{CurlyWId}.
\end{proof}

\subsection{Approximations} \label{sec:ApproxGen}

Covariant Taylor series will be considered for $U(\x,\x')$ and $V(\x,\x')$, making use of the coefficients for these expansions that have been previously computed by Decanini \& Folacci\cite{DecaniniFolacci}. For $\Sigma(\x,\x')$, we assume that sufficient geometrical data on the initial surface $\mathcal{C}$ is known, so that we expand $\Sigma(\x,\x')$ as a surface Taylor series. These expansions have the property that tangent derivatives match the corresponding tangent derivatives of $\Sigma(\x,\x')$ on $\mathcal{C}$ up to the number of terms considered in the expansion. However, we cannot, by construction, compute normal derivatives on a surface Taylor expansion. 

Instead, we will use the coincidence limits for $\Sigma(\x,\x')$, $\Sigma_{a_0}(\x,\x')$, \dots, $\Sigma_{a_5 a_4 a_3 a_2 a_1 a_0}(\x,\x')$, $\Sigma_{a_0'}(\x,\x')$, \dots, $\Sigma_{a_5 a_4 a_3 a_2 a_1 a_0'}(\x,\x')$ to compute the limits of tangent derivatives for $\dbar \Sigma(\x,\x')$, $\dbar^2 \Sigma(\x,\x')$, $\dbar' \Sigma(\x,\x')$, and $\dbar \dbar' \Sigma(\x,\x')$ to compute the corresponding surface Taylor expansions $\dbar \tilde{\Sigma}(\x,\x')$, $\dbar^2 \tilde{\Sigma}(\x,\x')$, $\dbar' \tilde{\Sigma}(\x,\x')$, and $\dbar \dbar' \tilde{\Sigma}(\x,\x')$ respectively, which according to \eqref{InitialDataTerms1} of Theorem \ref{th:SigmaSurfaceApprox}, are required to compute the renormalized stress-energy tensor at $\mathcal{C}$.

Initial data for the metric will be given in terms of the induced metric on $\mathcal{C}$, the extrinsic curvature $K_{ab}$, and further \textit{time} derivatives of the metric. We will therefore begin by writing down a $3+1$ expansion of the coefficients of the covariant Taylor series for $U(\x,\x')$ and $V(\x,\x')$, and afterwards we will compute the surface Taylor expansions for $\Sigma(\x,\x')$, $\dbar \Sigma(\x,\x')$, $\dbar' \Sigma(\x,\x')$, $\dbar \dbar' \Sigma(\x,\x')$ and $\dbar^2 \Sigma(\x,\x')$.

\subsubsection{Expansion and surface projections for U(x,x').}
The Taylor expansion for $U(\x,\x')$, up to order fourth order is given by
\begin{align}
	\tilde U (\x,\x') =& \mathcal{U}_0 (\x) + \mathcal{U}_1\,^b (\x) \nabla_b \tilde{\Sigma}(\x,\x') + \frac{1}{2} \mathcal{U}_2\,^{b_1 b_0} (\x) \nabla_{b_1} \tilde{\Sigma}(\x,\x') \nabla_{b_0} \tilde{\Sigma}(\x,\x')\nonumber\\
	&+ \frac{1}{6} \mathcal{U}_3\,^{b_2 b_1 b_0} (\x) \nabla_{b_2} \tilde{\Sigma}(\x,\x')\nabla_{b_1} \tilde{\Sigma}(\x,\x') \nabla_{b_0} \tilde{\Sigma}(\x,\x') \nonumber\\
	&+ \frac{1}{24} \mathcal{U}_4\,^{b_3 b_2 b_1 b_0} (\x) \nabla_{b_3} \tilde{\Sigma}(\x,\x')\nabla_{b_2} \tilde{\Sigma}(\x,\x')\nabla_{b_1} \tilde{\Sigma}(\x,\x') \nabla_{b_0} \tilde{\Sigma}(\x,\x'), \label{TaylorU}
\end{align}
with \cite{DecaniniFolacci},
\begin{subequations}
\begin{align}
	\mathcal{U}_0 &= 1,\\
	\mathcal{U}_1\,_a &= 0,\\
	\mathcal{U}_2\,_{a_1 a_0} &= \frac{1}{6} R_{a_1 a_0},\\
	\mathcal{U}_3\,_{a_2 a_1 a_0} &= \frac{1}{4} \nabla_{(a_2} R_{a_1 a_0)},\\
	\mathcal{U}_4\,_{a_3 a_2 a_1 a_0} &= \frac{3}{10}\nabla_{(a_3} \nabla_{a_2} R_{a_1 a_0)} + \frac{1}{12} R_{(a_3 a_2}R_{a_1 a_0)} + \frac{1}{15} R_{b_1 (a_3 |b_0|a_2 } R^{b_1}\,_{a_1}\,^{b_0}\,_{a_0)}. \label{UTaylor4}
\end{align}
\end{subequations}
Note that the $3+1$ expansion of \eqref{TaylorU} reads
\begin{align}
	\tilde U =& \mathcal{U}_0 + \overset{0}{\mathcal U}_1\,^{b_0} (D_{b_0} \Sigma) - \overset{1}{\mathcal U}\,_1 (\dbar \Sigma) + \frac{1}{2} \overset{0}{\mathcal U}_2\,^{b_1 b_0} (D_{b_1} \Sigma) (D_{b_0} \Sigma) - \overset{1}{\mathcal U}_2\,^{b_1 } (D_{b_1} \Sigma) (\dbar \Sigma) + \frac{1}{2} \overset{3}{\mathcal U}_2 \,(\dbar \Sigma)^2 \nonumber\\
	&+ \frac{1}{6}  \overset{0}{\mathcal U}_3\,^{b_2 b_1 b_0} (D_{b_2} \Sigma) (D_{b_1} \Sigma) (D_{b_0} \Sigma) - \frac{1}{2} \overset{1}{\mathcal U}_3\,^{b_2 b_1} (D_{b_2} \Sigma) (D_{b_1} \Sigma) (\dbar \Sigma) \nonumber\\
	&+ \frac{1}{2} \overset{3}{\mathcal U}_3\,^{b_2} (D_{b_2} \Sigma) (\dbar \Sigma)^2 - \frac{1}{6} \overset{7}{\mathcal U}_3 (\dbar \Sigma)^3 \nonumber\\
	&+\frac{1}{24} \overset{0}{\mathcal U}_4\,^{b_3 b_2 b_1 b_0} (D_{b_3} \Sigma)(D_{b_2} \Sigma) (D_{b_1} \Sigma) (D_{b_0} \Sigma) - \frac{1}{6} \overset{1}{\mathcal U}_4\,^{b_3 b_2 b_1} (D_{b_3} \Sigma)(D_{b_2} \Sigma) (D_{b_1} \Sigma) (\dbar \Sigma) \nonumber\\
	&+\frac{1}{4} \overset{3}{\mathcal U}_4\,^{b_3 b_2} (D_{b_3} \Sigma)(D_{b_2} \Sigma) (\dbar \Sigma)^2 - \frac{1}{6} \overset{7}{\mathcal U}_4\,^{b_3} (D_{b_3} \Sigma)(\dbar \Sigma)^3 + \frac{1}{24} \overset{15}{\mathcal U}_4  (\dbar\Sigma)^4, \label{TildeU}
\end{align}
where the superscript represent the $3+1$ tangent projection label of the corresponding coefficient for $\tilde U$, according to the convention given in appendix \ref{App:3p1}. These projection components are
\begin{subequations}
	\begin{align}
		\overset{0}{\mathcal U}_{1\, a_0} &= 0,\\
		\overset{1}{\mathcal U}_{1} &= 0,
	\end{align}
\end{subequations}
\begin{subequations}
	\begin{align}
		\overset{0}{\mathcal U}_{2\, a_1 a_0} &= \frac{1}{6}\left( \,^{(3)}\mathsf{R}_{a_1 a_0} + K\, K_{a_1 a_0}  + D_{a_1} u_{a_0} - u_{a_1} u_{a_0} - \dbar K_{a_1 a_0} \right),\\
		\overset{1}{\mathcal U}_{2\, a_1 } &= \frac{1}{6}\left( D_{a_1} K - D_b K^b\,_{a_1} \right),\\
		\overset{3}{\mathcal U}_2 &= \frac{1}{6}\left( \dbar K + u^2 - K^{b_1 b_0} K_{b_1 b_0} - D_b u^b \right),
	\end{align}
\end{subequations}
\begin{subequations}
	\begin{align}
		\overset{0}{\mathcal U}_{3\, a_2 a_1 a_0} &= \frac{1}{4} \Big( - D_{(a_2} \dbar K_{a_1 a_0)} + 3 K_{(a_2 a_1} D_{a_0)}K + K D_{(a_0} K_{a_2 a_1)} - 2 u_{(a_2} D_{a_1} u_{a_0)} + D_{(a_2} \,^{(3)}\mathsf{R}_{a_1 a_0)} \nonumber\\
		&\qquad + D_{(a_2} D_{a_1} u_{a_0)} - 2 K_{(a_2 a_1} D^{b_0}K_{a_0) b_0} \Big)\\
	\overset{1}{\mathcal U}_{3\, a_2 a_1} &= \frac{1}{12} \Big(  - \dbar^2 K_{a_2 a_1}  + 3 u^2 K_{a_2 a_1} - 2 K^{b_1 b_0} K_{b_1 b_0} K_{a_2 a_1} + 2 K K_{b_0 (a_2} K^{b_0}\,_{ a_1)} \nonumber\\
	&\qquad + 4 K_{b_0 (a_2} \,^{(3)}\mathsf{R}^{b_0}\,_{a_1)} - 2 K_{b_1 b_0} \,^{(3)}R^{b_1}\,_{(a_1}\,^{b_0}\,_{a_2)} - 3 K^{b_0}\,_{(a_2} u_{a_1)} u_{b_0} + 3 K_{a_2 a_1} \dbar K  \nonumber\\
	&\qquad + K \dbar K_{a_2 a_1}- 2 K_{b_0 (a_2} \dbar K^{b_0}\,_{a_1)} + \dbar \,^{(3)}\mathsf{R}_{a_2 a_1} - 3 u_{(a_2} \dbar u_{a_1)}+ 2 u_{(a_2} D_{a_1)} K \nonumber\\
	&\qquad + u_{b_0} D_{(a_2}K^{b_0}\,_{a_1)} + D_{(a_2}\dbar u_{a_1)} + 2 D_{(a_2} D_{a_1)} K - u^{b_0} D_{b_0} K_{a_2 a_1}  \nonumber\\
	&\qquad - 2 u_{(a_2} D^{b_0} K_{|b_0| a_1)} + 3 K^{b_0}\,_{(a_2} D_{|b_0|}u_{a_1)} - 2 K_{a_2 a_1} D^{b_0} u_{b_0} - 2 D^{b_0} D_{(a_2} K_{a_1) b_0} \Big) \\
	\overset{3}{\mathcal U}_{3\, a_2} &= \frac{1}{12} \Big( 3 D_{a_2} \dbar K - 2 D_{b_0} \dbar K^{b_0}\,_{a_2} + 2 K^{b_1}\,_{a_2} K_{b_1 b_0} u^{b_0}+3 \,^{(3)}\mathsf{R}_{b_0 a_2} u^{b_0}-2 K^{b_0}\,_{a_2} K_{b_1 b_0} u^{b_1}\nonumber\\
	&\qquad\qquad -D^b D_b u_{a_2}+4 u D_{a_2} u -4 K_{b_1 b_0} D_{a_2} K_{b_1 b_0} +6 K_{b_0 a_2} D_{b_0} K \nonumber\\
	&\qquad\qquad -4 K_{b_1 a_2} D_{b_0} K_{b_1 b_0} - 2 u_{a_2} D_{b_0} u_{b_0} \Big)\\
	\overset{7}{\mathcal U}_{3} &= \frac{1}{4} \Big(\dbar^2 K + K^{b_1 b_0} \big(u_{b_0} u_{b_1} - 2 \dbar K_{b_1 b_0}\big) - D_{b_0} \dbar u^{b_0} + 3 u \dbar u \nonumber\\
	&\qquad\qquad + 3 u^{b_0} \big(D_{b_0} K - D_{b_1} K^{b_1}\,_{b_0}\big) - K_{b_1 b_0} D_{b_1} u_{b_0} - K u^2 \Big)
	\end{align}
\end{subequations}
	Note that due to Bianchi identity, $\dbar \,^{(3)}R_{a_3 a_2 a_1 a_0}$ and normal derivatives of its contractions are given in terms of $K_{a_1 a_0}$, $u_{a}$ and their tangent derivatives. (cf. Appendix \ref{App:3p1}, eqs. \eqref{BianchiDbarR}).

\begin{subequations}
\begin{align}
	\overset{0}{\mathcal U}_{4\, a_3 a_2 a_1 a_0} &= \frac{1}{60} \Big(-18 K_{(a_3 a_2} \dbar^2 K_{a_1 a_0)} + 54 K_{(a_3 a_2} K_{a_1 a_0)} \dbar K+ \dbar K_{(a_3 a_2}\big(8 K K_{a_1 a_0)}  \nonumber\\
	&\qquad -8 K^{b_0}\,_{a_1} K_{a_0)b_0} -10 \,^{(3)}\mathsf{R}_{a_1 a_0)}+18 u_{a_1} u_{a_0)}+9 \dbar K_{a_1 a_0)} \big) \nonumber\\
	&\qquad -18 \big(2 K_{b_0 (a_3} \dbar K^{b_0}\,_{a_2} K_{a_1 a_0)} -\dbar \,^{(3)}\mathsf{R}_{(a_3 a_2} K_{a_1 a_0)} +3 \dbar u_{(a_3}  u_{a_2} K_{a_1 a_0)} \nonumber\\
	&\qquad+ (D_{(a_3} u_{a_2}) \dbar K_{a_1 a_0)}\big)+18 \big(K_{(a_3 a_2} D_{a_1}\dbar u_{a_0)}-D_{(a_3} D_{a_2} \dbar K_{a_1 a_0)}\big)\Big) \nonumber\\
	&\qquad+ \,^0\mathfrak{U}_{a_3 a_2 a_1 a_0}\\
	\overset{1}{\mathcal U}_{4\, a_3 a_2 a_1} &= \frac{1}{60}\Big( 9 (-D_{(a_3} \dbar^2 K_{a_2 a_1)} +K D_{(a_3} \dbar K_{a_2 a_1)}) -27 u_{(a_3} D_{a_2} \dbar u_{a_1)} +18 \dbar K D_{(a_3} K_{a_2 a_1)} \nonumber\\
	&\qquad +9 \big(-3 \dbar u_{(a_3)} D_{a_2} u_{a_1)} +D_{(a_3} \dbar \,^{(3)}\mathsf{R}_{a_2 a_1)} +D_{(a_3} D_{a_2} \dbar u_{a_1)} \big)\nonumber\\
	&\qquad+\big(-22 D_{(a_3} K_{b_0 a_2} +13 D_{b_0} K_{(a_3 a_2} \big)\dbar K^{b_0}\,_{a_1)} \nonumber\\
	&\qquad+ \big(13 D_{(a_3} K -4 D^{b_0} K_{b_0 (a_3} \big)\dbar K_{a_2 a_1)}-9 K_{b_0 (a_3} \big(2 D_{a_2} \dbar K_{b_0 a_1)} +D_{b_0} \dbar K_{a_2 a_1)} \big)\nonumber\\
	&\qquad +27 K_{(a_3 a_2} \big(2 D_{a_1)} \dbar K -D_{b_0} \dbar K_{b_0 a_1)} \big)\Big) +\,^1\mathfrak{U}_{a_3 a_2 a_1},\\
	\overset{3}{\mathcal U}_{4\, a_3 a_2} &=\frac{1}{180}\Big( 9 \big(-\dbar^3 K_{a_3 a_2} + 6 K_{a_3 a_2} \dbar^2 K +K \dbar^2 K_{a_3 a_2} \nonumber\\
	&\qquad\qquad -4 \big(K^{b_0}\,_{(a_3} \dbar^2 K_{a_2) b_0} +u_{(a_3} \dbar^2 u_{a_2)} \big)+D_{(a_3} \dbar^2 u_{a_2)} \big) \nonumber\\
	&\qquad-\big(22 K^{b_0}\,_{(a_3} K_{a_2)}\,^{b_1}+86 K_{a_3 a_2} K^{b_1 b_0}+32 \,^{(3)}R^{b_1}\,_{(a_3}\,^{b_0}\,_{a_2)}\big) \dbar K_{b_1 b_0}\nonumber\\
	&\qquad+9 \dbar^2 \,^{(3)}\mathsf{R}_{a_3 a_2}+54 K^{b_0}\,_{(a_3} \dbar \,^{(3)}\mathsf{R}_{a_2) b_0} +18 K^{b_1 b_0} \dbar \,^{(3)}R_{b_1 (a_3 a_2) b_0} \nonumber\\
	&\qquad+\big(72 K_{b_0 (a_3}K_{a_2) b_0} +5 \big(K K_{a_3 a_2}+\,^{(3)}\mathsf{R}_{a_3 a_2}-u_{a_3}u_{a_2} \big)-13 D_{(a_3} u_{a_2)} \big)\dbar K \nonumber\\
	&\qquad+9 \big(-7 K_{b_0 (a_3} u_{a_2)}+2 \big(9 K_{a_3 a_2} u_{b_0}+D_{(a_3} K_{a_2) b_0} -D_{b_0} K_{a_3 a_2} \big)\big)\dbar u^{b_0} \nonumber\\
	&\qquad+9 \big(-3 \big(3 u^{b_0} K_{b_0 (a_3}+\dbar u_{(a_3} \big)+2 \big(D_{(a_3} K -D_{b_0} K^{b_0}\,_{(a_3} \big)\big)\dbar u_{a_2)}\nonumber\\
	&\qquad+\big(-10 K^{b_1}\,_{(a_3} K_{b_1 b_0}+54 \big(K K_{b_0 (a_3}+\,^{(3)}\mathsf{R}_{b_0 (a_3}\big)-22 \big(2 u_{b_0} u_{(a_3} +\dbar K_{b_0 (a_3} \big)\nonumber\\
	&\qquad\quad+53 D_{b_0} u_{(a_3} \big)\dbar K_{a_2)}\,^{b_0} + \big(31 \big(u^2+\dbar K \big)-13 \big(K^{b_1 b_0}K_{b_1 b_0}+D_{b_0} u^{b_0} \big)\big)\dbar K_{a_3 a_2}\nonumber\\
	&\qquad +18 \big(2 u_{(a_3} D_{a_2)} \dbar K -u^{b_0} D_{b_0} \dbar K_{a_3 a_2} \big)+9 \big(K^{b_0}\,_{(a_3} D_{a_2)} \dbar u_{b_0} +5 D_{a_3} D_{a_2} \dbar K \nonumber\\
	&\qquad\quad +u^{b_0} D_{(a_3} \dbar K_{a_2) b_0)} +5 K^{b_0}\,_{(a_3} D_{b_0} \dbar u_{a_2)} -5 K_{a_3 a_2} D_{b_0} \dbar u^{b_0} \nonumber\\
	&\qquad\quad -4 \big(u_{(a_3} D_{b_0} \dbar K^{b_0}\,_{a_2)} +D_{b_0} D_{(a_3} \dbar K^{b_0}\,_{a_2)}\big)\big) \Big)+\,^3\mathfrak{U}_{a_3 a_2},\\
	\overset{7}{\mathcal U}_{4\, a_3} &=+\,^7\mathfrak{U}_{a_3},\\
	\overset{15}{\mathcal U}_{4} &=+\,^{15}\mathfrak{U},
\end{align}
\end{subequations}
where $\,^0\mathfrak{U}_{a_3 a_2 a_1 a_0}$, $\,^1\mathfrak{U}_{a_3 a_2 a_1}$, $\,^3\mathfrak{U}_{a_3 a_2}$, $\,^7\mathfrak{U}_{a_3}$ and $\,^{15}\mathfrak{U}$ are symmetric tensors that do not involve normal derivatives.
With \eqref{TildeU} we can compute up to four tangent derivatives of $\tilde{U}(\x,\x')$ which match in the coincidence limits with those of $U(\x,\x')$. However, according to Theorem \ref{th:SigmaSurfaceApprox}, we also require a third order surface approximation for $\dbar U(\x,\x')$ and second order surface approximations for $\dbar \dbar' U(\x,\x')$ and $\dbar^2 U(\x,\x')$. This is achieved by computing the corresponding normal derivatives on \eqref{TildeU}, where we will require, in addition to the aforementioned expansions for $\dbar \Sigma$, $\dbar \dbar' \Sigma$ and $\dbar^2 \Sigma$, up to two normal derivatives of all the tangent projection components of $\mathcal{U}_{2\,a_1 a_0}$, and one normal derivative of all the tangent projection components of $\mathcal{U}_{3\,a_2 a_1 a_0}$. We require no further normal derivatives of $\mathcal{U}_{4\,a_3 a_2 a_1 a_0}$ because they appear only when computing derivatives of order higher than $4$, which is already above the order or our approximation.

Note that in any case, the higher order of normal derivatives is three normal derivatives of the extrinsic curvature, or in other words, four \textit{time derivatives} of the induced metric, see \eqref{DtHDk}.

\subsubsection{Expansion and surface projections for V(x,x').}

We recall that the biscalar $V(\x,\x')$ admits the Hadamard expansion \eqref{vSigma}. We define an analogue approximation for $\tilde V(\x,\x')$ of the form,
\begin{equation}
	\tilde{V} (\x,\x') = \tilde{V}_0 (\x,\x') + \tilde{V}_1(\x,\x') \tilde{\Sigma} (\x,\x'). \label{TildeVPwr}
\end{equation}

We only define the covariant Taylor expansion for $\tilde{V}_0$, as $V_1$ is only required up to zeroth order,
\begin{align}
	\tilde{V}_0 (\x,\x') &= \mathcal{V}_0 (\x) + \mathcal{V}_1\,^a (\x) \nabla_a \tilde{\Sigma}(\x,\x') + \frac{1}{2}\mathcal{V}_2\,^{a_1 a_0} (\x) \nabla_{a_1} \tilde{\Sigma}(\x,\x') \nabla_{a_0} \tilde{\Sigma}(\x,\x').
\end{align}

In this way, $\tilde{V}$ given by \eqref{TildeVPwr} is a second order approximation for $V(\x,\x')$, as required by Theorem \ref{Th:OrderEstimates}. The required coefficients have been computed in \cite{DecaniniFolacci},
\begin{subequations}
	\begin{align}
		\mathcal{V}_0 &= \frac{1}{2}m^2 + \frac{1}{2} \left(\xi - \frac{1}{6}\right)R,\\
		\mathcal{V}_{1\, a} &= \frac{1}{4}\left(\xi - \frac{1}{6}\right) \nabla_a R,\\
		\mathcal{V}_{2\, a_1 a_0} &= \frac{1}{6}\left(\xi - \frac{3}{20}\right)\nabla_{a_1} \nabla_{a_2} R -\frac{1}{120} \Box R_{a_1 a_0} - \frac{1}{180} R^{b_2 b_1 b_0}\,_{a_1} R_{b_2 b_1 b_0 a_0} - \frac{1}{180} R^{b_1 b_0} R_{b_1 a_1 b_0 a_0} \nonumber \\
		&\qquad + \frac{1}{90} R^{b}\,_{a_1} R_{b a_0} +\frac{1}{12}\left(\xi - \frac{1}{6}\right) R R_{a_1 a_0} + \frac{1}{12} m^2 R_{a_1 a_0}, \label{VTaylor02}\\
		\tilde{V}_1 &= - \frac{1}{24} \left(\xi - \frac{1}{5}\right) \Box R + \frac{1}{720} R^{b_3 b_2 b_1 b_0} R_{b_3 b_2 b_1 b_0} - \frac{1}{720} R^{b_1 b_0} R_{b_1 b_0} + \frac{1}{8}\left(\xi - \frac{1}{6}\right)^2 R^2 \nonumber\\
		&\qquad + \frac{1}{4}\left(\xi - \frac{1}{6}\right)m^2 R + \frac{1}{8} m^4. \label{VTaylor10}
	\end{align}
\end{subequations}
and their corresponding $3+1$ tangent projections are
\begin{align}
	\mathcal{V}_0 = \frac{1}{12} \Big(&6 m^2+(6 \xi -1) \big(K^2-2 u^2+K^{b_1b_0}K_{b_1b_0}+\,^{(3)}\mathscr{R}\big) \nonumber\\
	&\quad +(2-12 \xi ) \dbar K +2 (6 \xi -1) D_{b_0} u_{b_0} \Big),
\end{align}
\begin{subequations}
	\begin{align}
		\overset{0}{\mathcal{V}_{1\, a_0}} = -\frac{1}{24}(6 \xi -1) \Big(& 2 \,^{(3)}\mathsf{R}^{b_0}\,_{a_0} u_{b_0}-2 D^{b_0}D_{b_0} u_{a_0}-2 K D_{a_0} K -D_{a_0} \,^{(3)}\mathscr{R} +4 u D_{a_0} u \nonumber\\
		&\quad-2 K^{b_1b_0} D_{a_0} K_{b_1b_0} +2 D_{a_0} \dbar K \Big)\\
		\overset{1}{\mathcal{V}_{1\,}} = \frac{1}{24}(6 \xi -1) \Big(& 2 K \dbar K +\dbar \,^{(3)}\mathscr{R} +2 \big(K u^2-3 u \dbar u -\dbar^2 K -u^{b_0} D_{b_0} K +u^{b_1} D_{b_0} K_{b_1}\,^{b_0} \nonumber\\
		&\quad+D_{b_0} \dbar u^{b_0} \big)+2 K^{b_1 b_0} \big(-u_{b_0} u_{b_1}+\dbar K_{b_1 b_0} +D_{b_1} u_{b_0} \big) \Big)
	\end{align}
\end{subequations}
\begin{subequations}
\begin{align}
	\overset{0}{\mathcal{V}}_{2\, a_1 a_0} = \frac{1}{360}\Big(& - 3 \dbar^3 K_{a_1 a_0} + 6 K \dbar^2 K_{a_1 a_0} + 3 (7 - 40 \xi ) K_{a_1 a_0} \dbar^2 K \nonumber\\
	&\quad+ 3 \big( D_{(a_1} \dbar^2 u_{a_0)} - 4 u_{(a_1} \dbar^2 u_{a_0)} \big) + 3 D^{b_0}D_{b_0} \dbar K_{a_1 a_0} \nonumber\\
	&\quad+ 6 (3 - 20 \xi ) D_{(a_1} D_{a_0)} \dbar K - 6 u_{b_0} D_{b_0} \dbar K_{a_1 a_0} + 3 u^{b_0} D_{(a_1} \dbar K_{a_0) b_0} \nonumber\\
	&\quad- 12 u_{(a_1} D_{b_0} \dbar K_{b_0 a_0)} + 12 u_{(a_1} D_{a_0)} \dbar K - 3 K D_{(a_1} \dbar u_{a_0)} + 6 K_{b_0 (a_1} D_{b_0} \dbar u_{a_0)} \nonumber\\
	&\quad+ 6 ( - 3 + 20 \xi ) K_{a_1 a_0} D_{b_0} \dbar u_{b_0} + \big(5 \,^{(3)}\mathscr{R} - 6 u^2 - 30 M^2 + K^2 (2 - 30 \xi ) \nonumber\\
	&\quad+ 7 \big(K^{b_1 b_0}K_{b_1 b_0}\big) - 30 \xi \big(\,^{(3)}\mathscr{R} - 2 u^2 + K^{b_1 b_0}K_{b_1 b_0}\big) + ( - 6 + 60 \xi ) \dbar K \nonumber\\
	&\quad+ 12 (1 - 5 \xi ) D_{b_0} u_{b_0} \big) \dbar K_{a_1 a_0} + \big(14 K^{b_1}\,_{b_0}K_{b_1 (a_1} - 8 \big(K K_{b_0 (a_1} + \,^{(3)}\mathsf{R}_{b_0 (a_1}\big)\nonumber\\
	&\quad + 3 D_{b_0} u_{(a_1} \big) \dbar K^{b_0}\,_{a_0)} + 2 \big( - K_{b_1 (a_1} K_{a_0) b_0} + 4 ( 15 \xi - 2 ) K_{a_1 a_0} K_{b_1 b_0} \nonumber\\
	&\quad- \,^{(3)}R_{b_1 (a_1 a_0) b_0} \big) \dbar K^{b_1 b_0} + \big(( 60 \xi - 11 ) K K_{a_1 a_0} - 4 K^{b_0}\,_{(a_1} K_{a_0) b_0} \nonumber\\
	&\quad+ 10 (1 - 6 \xi ) \,^{(3)}\mathsf{R}_{a_1 a_0} + 6 ( - 3 + 10 \xi ) u_{(a_1} u_{a_0)} + 12 ( 1 - 5 \xi ) D_{(a_1} u_{a_0)} \big) \dbar K \nonumber\\
	&\quad+ \big(9 K u_{(a_1} - 6 u^{b_0} K_{b_0 (a_1} - 9 \dbar u_{(a_1} + 6 D_{(a_1} K - 6 D_{b_0} K^{b_0}\,_{(a_1} \big) \dbar u_{a_0)} \nonumber\\
	&\quad- 6 \big(K_{b_0 (a_1} u_{a_0)} - D_{(a_1} K_{a_0) b_0} + D_{b_0} K_{a_1 a_0} \big) \dbar u^{b_0} + 9 (7 - 40 \xi ) K_{a_1 a_0} u \dbar u \nonumber\\
	&\quad+ 3 \big( \dbar^2 \,^{(3)}\mathsf{R}_{a_1 a_0} - K \dbar \,^{(3)}\mathsf{R}_{a_1 a_0} + ( 20 \xi - 3 ) K_{a_1 a_0} \dbar \,^{(3)}\mathscr{R} \big)
	\Big) + \,^{0}\mathfrak{V}_{a_1 a_0}
\end{align}
\begin{align}
	\overset{1}{\mathcal{V}_{2\, a_1}} = \frac{1}{360} \Big(& 3 (7 - 40 \xi ) D_{a_1} \dbar^2 K - 3 D_{b_0} \dbar^2 K^{b_0}\,_{a_1} - 6 K_{b_0 a_1} u_{b_1} \dbar K^{b_1 b_0} - 360 \xi \dbar u D_{a_1} u \nonumber\\
	&- 18 \dbar K^{b_1 b_0} D_{a_1} K_{b_1 b_0} + 120 \xi \dbar K^{b_1 b_0} D_{a_1} K_{b_1 b_0} + 120 K \xi D_{a_1} \dbar K + 60 \xi D_{a_1} \dbar \,^{(3)}\mathscr{R} \nonumber\\
	&- 360 u \xi D_{a_1} \dbar u + 3 u^{b_0} \big(2 \dbar \,^{(3)}\mathsf{R}_{b_0 a_1} + D_{a_1} \dbar u_{b_0} \big) + 60 \xi \dbar K D_{b_0} K^{b_0}\,_{a_1} \nonumber\\
	&- 4 \dbar K^{b_1 b_0} D_{b_0} K_{b_1 a_1} + 21 K^{b_0}\,_{a_1} D_{b_0} \dbar K - 120 \xi K^{b_0}\,_{a_1} D_{b_0} \dbar K + 3 K D_{b_0} \dbar K^{b_0}\,_{a_1}\nonumber\\
	&+ K^{b_1}\,_{a_1} \big(6 u^{b_0} \dbar K_{b_1 b_0} - 3 D_{b_0} \dbar K^{b_1}\,_{b_0} \big) + 120 \xi D_{b_0} D_{a_1} \dbar u^{b_0} + \dbar K^{b_1 b_0} D_{b_1} K_{b_0 a_1} \nonumber\\
	&- 3 \big(( - 7 + 40 \xi ) \,^{(3)}\mathsf{R}_{b_0 a_1} \dbar u^{b_0} - \dbar u \big(u\,  u_{a_1} + 19 D_{a_1} u \big) + 7 K D_{a_1} \dbar K + 3 D_{a_1} \dbar \,^{(3)}\mathscr{R} \nonumber\\
	&- 19 u D_{a_1} \dbar u + \dbar K \big((4 - 20 \xi ) D_{a_1} K + 3 D_{b_0} K^{b_0}\,_{a_1} \big) + \dbar u_{a_1} \big(u^2 + D_{b_0} u^{b_0} \big) \nonumber\\
	&+ 2 u_{a_1} D_{b_0} \dbar u^{b_0} + 6 D_{b_0} D_{a_1} \dbar u^ {b_0} - \dbar K^{b_0}\,_{a_1} \big(D_{b_0} K + D_{b_1} K^{b_1}\,_{b_0} \big)\big) \nonumber\\
	&+ 3 K^{b_1 b_0} \big( - 2 u_{b_1} \dbar K_{b_0 a_1} + 2 u_{b_0} \dbar K_{b_1 a_1} + K_{b_1 a_1} \dbar u_{b_0} - K_{b_0 a_1} \dbar u_{b_1} \nonumber\\
	&\qquad + ( - 7 + 40 \xi ) D_{a_1} \dbar K_{b_1 b_0} + D_{b_1} \dbar K_{b_0 a_1} \big) \Big) + \,^{1}\mathfrak{V}_{a_1},
\end{align}
\begin{align}
	\overset{3}{\mathcal{V}_{2}} = \frac{1}{360} \Big(& 3(7 - 40 \xi ) \dbar^3 K + 3 K ( - 7 + 40 \xi ) \dbar^2 K + 24 ( - 1 + 5 \xi ) K^{b_1 b_0} \dbar^2 K_{b_1 b_0} \nonumber\\
	&+ 60 \xi \big(\dbar^2 \,^{(3)}\mathscr{R} - 2 \big(3 u \dbar^2 u + u^{b_0} \dbar^2 u_{b_0} - D_{b_0} \dbar^2 u^{b_0} \big)\big) + 21 \big(3 u \dbar^2 u + u^{b_0} \dbar^2 u_{b_0} \nonumber\\
	&\qquad - D_{b_0} \dbar^2 u^{b_0} \big) - 3 D^b D_b \dbar K + 12 ( - 1 + 5 \xi ) \dbar K^2 - 72 K u \dbar u - 9 \dbar^2 \,^{(3)}\mathscr{R} \nonumber\\
	& + 360 K u \xi \dbar u + 63 \dbar u^2 - 360 \xi \dbar u^2 + 6 K^{b_2 b_0} K_{b_2}\,^{b_1} \dbar K_{b_1 b_0} - 2 \,^{(3)}\mathsf{R}^{b_1 b_0} \dbar K_{b_1 b_0} \nonumber\\
	&+ 60 u^{b_0} u^{b_1} \dbar K_{b_1 b_0} - 240 \xi u^{b_0} u^{b_1} \dbar K_{b_1 b_0} - 18 \dbar K^{b_1 b_0} \dbar K_{b_1 b_0}+ 120 \xi \dbar K^{b_1 b_0}\dbar K_{b_1 b_0} \nonumber\\
	&+ 48 \dbar u^{b_0} D_{b_0} K - 240 \xi \dbar u^{b_0} D_{b_0} K - 48 \dbar u_{b_1} D_{b_0} K^{b_1 b_0} + 240 \xi \dbar u_{b_1} D_{b_0} K^{b_1 b_0} \nonumber\\
	&+ \dbar K \big( - 5 \,^{(3)}\mathscr{R} - 36 u^2 + 30 M^2 + 5 K^2 ( - 1 + 6 \xi ) - 13 \big(K^{b_1 b_0}K_{b_1 b_0}\big) \nonumber\\
	&\qquad + 30 \xi \big(\,^{(3)}\mathscr{R} + 4 u^2 + 3 \big(K^{b_1 b_0}K_{b_1 b_0}\big)\big) + 12 ( - 1 + 10 \xi ) D_{b_0} u^{b_0} \big) + 54 u^{b_0} D_{b_0} \dbar K \nonumber\\
	&- 240 \xi u^{b_0} D_{b_0} \dbar K - 33 u_{b_1} D_{b_0} \dbar K^{b_1 b_0} + 120 \xi u_{b_1} D_{b_0} \dbar K^{b_1 b_0} + 3 K D_{b_0} \dbar u^{b_0} \nonumber\\
	&+ 3 ( - 11 + 40 \xi ) \dbar K^{b_1 b_0} D_{b_1} u_{b_0} + 2 K^{b_1 b_0} \big(2 K \dbar K_{b_1 b_0} - 5 K^{b_2}\,_{b_0} \dbar K_{b_2 b_1} \nonumber\\
	&\qquad - 3 ( - 7 + 40 \xi ) \big(u_{b_1} \dbar u_{b_0} + u_{b_0} \dbar u_{b_1} - D_{b_1} \dbar u_{b_0} \big)\big) \Big) + \,^{3}\mathfrak{V},
\end{align}
\end{subequations}
\begin{align}
\tilde{V}_{1} = \frac{1}{720} \Big(& 6 ( - 1 + 5 \xi ) \big(-2\dbar^3 K  +4 K \dbar^2 K + \dbar^2 \,^{(3)}\mathscr{R} + 2 K^{b_1 b_0} \dbar^2 K_{b_1 b_0} - 2 \big(3 u \dbar^2 u + u^{b_0} \dbar^2 u_{b_0}\nonumber\\
&\qquad - D_{b_0} \dbar^2 u^{b_0} \big)\big) + 60 \xi D^b D_b \dbar K + ( - 3 + 60 \xi ( - 1 + 6 \xi )) (\dbar K)^2 - 30 K \xi \dbar \,^{(3)}\mathscr{R} \nonumber\\
&- 72 K u \dbar u + 360 K u \xi \dbar u + 36 (\dbar u)^2 - 180 \xi (\dbar u)^2 + 2 \,^{(3)}\mathsf{R}^{b_1 b_0} \dbar K_{b_1 b_0} + 30 u^{b_0} u^{b_1} \dbar K_{b_1 b_0} \nonumber\\
&- 120 \xi u^{b_0} u^{b_1} \dbar K_{b_1 b_0} - 9 \dbar K^{b_1 b_0} \dbar K_{b_1 b_0}+ 60 \xi \dbar K^{b_1 b_0}\dbar K_{b_1 b_0} - 120 \xi \dbar u^{b_0} D_{b_0} K \nonumber\\
&+ 120 \xi \dbar u_{b_1} D_{b_0} K^{b_1 b_0} + 6 \big( - 2 D^b D_b \dbar K + K \dbar \,^{(3)}\mathscr{R} + 4 \dbar u^{b_0} D_{b_0} K \nonumber\\
&\qquad- 4 \dbar u_{b_1} D_{b_0} K^{b_1 b_0} \big) + 2 \dbar K \big(K^2 - 5 \,^{(3)}\mathscr{R} - 3 u^2 - 4 K^{b_1 b_0}K_{b_1 b_0} \nonumber\\
&\qquad+ 30 \big(M^2 (1 - 6 \xi ) + \xi \big(K^2 (1 - 6 \xi ) + 2 u^2 ( - 1 + 6 \xi ) \nonumber\\
&\qquad\qquad- 2 ( - 1 + 3 \xi ) \big(\,^{(3)}\mathscr{R} + K^{b_1 b_0}K_{b_1 b_0}\big)\big)\big) - 3 (3 + 40 \xi ( - 1 + 3 \xi )) D_{b_0} u^{b_0} \big) \nonumber\\
&+ 24 u^{b_0} D_{b_0} \dbar K - 120 \xi u^{b_0} D_{b_0} \dbar K - 12 u_{b_1} D_{b_0} \dbar K^{b_1 b_0} + 60 \xi u_{b_1} D_{b_0} \dbar K^{b_1 b_0} \nonumber\\
&+ 12 K D_{b_0} \dbar u^{b_0} - 60 K \xi D_{b_0} \dbar u^{b_0} + 6 ( - 3 + 10 \xi ) \dbar K^{b_1 b_0} D_{b_1} u_{b_0} \nonumber\\
&- 2 K^{b_1 b_0} \big(K ( - 7 + 30 \xi ) \dbar K_{b_1 b_0} + 4 K^{b_2}\,_{b_0} \dbar K_{b_2 b_1} + 12 ( - 1 + 5 \xi ) \big(u_{b_1} \dbar u_{b_0} + u_{b_0} \dbar u_{b_1} \nonumber\\
&\qquad\qquad- D_{b_1} \dbar u_{b_0} \big)\big) \Big) + \mathfrak{V},
\end{align}
where  $\,^{0}\mathfrak{V}_{a_1 a_0}$, $\,^{1}\mathfrak{V}_{a_1}$, $\,^{3}\mathfrak{V}$ and $\mathfrak{V}$ are symmetric tensors that do not involve normal derivatives. This illustrates the fact  that four time derivatives of the induced metric on $\mathcal{C}$ are required as initial data to define the second order approximation $\tilde{V}(\x,\x')$.

\subsection{Surface approximation for \texorpdfstring{$\Sigma$}{S}(x,x')} \label{sec:HypSigma}

The surface Taylor expansion for $\Sigma(\x,\x')$ on $\mathcal{C}$ is

\small{\begin{align} \label{HypApproxS0}
	\tilde{\Sigma} (\x,\x') =& \overset{0}{S}(\x) + \overset{1}{S}^a (\x) D_a \sigma (\x,\x') + \frac{1}{2} \overset{2}{S}^{a_1 a_0} (\x) D_{a_1} \sigma (\x,\x') D_{a_0} \sigma (\x,\x') \nonumber \\
	& + \frac{1}{6} \overset{3}{S}^{a_2 a_1 a_0} (\x) D_{a_2} \sigma (\x,\x') D_{a_1} \sigma (\x,\x') D_{a_0} \sigma (\x,\x')\nonumber \\
	& + \frac{1}{24} \overset{4}{S}^{a_3 a_2 a_1 a_0} (\x) D_{a_3} \sigma (\x,\x') D_{a_2} \sigma (\x,\x') D_{a_1} \sigma (\x,\x') D_{a_0} \sigma (\x,\x') \nonumber \\
	& + \frac{1}{5!} \overset{5}{S}^{a_4 a_3 a_2 a_1 a_0} (\x) D_{a_4} \sigma (\x,\x') D_{a_3} \sigma (\x,\x') D_{a_2} \sigma (\x,\x') D_{a_1} \sigma (\x,\x') D_{a_0} \sigma (\x,\x') \nonumber \\
	& + \frac{1}{6!} \overset{6}{S}^{a_5 a_4 a_3 a_2 a_1 a_0} (\x) D_{a_5} \sigma (\x,\x') D_{a_4} \sigma (\x,\x') D_{a_3} \sigma (\x,\x') D_{a_2} \sigma (\x,\x') D_{a_1} \sigma (\x,\x') D_{a_0} \sigma (\x,\x'),
\end{align}}
where $\sigma(\x,\x')$ is half the squared surface geodesic distance on $\mathcal{C}$ from $\x$ to $\x'$. 

The elements required to construct a covariant Taylor series \eqref{CovTaylorGen} for a bi-scalar $A(\x,\x')$, particularly its expansion coefficients \eqref{CovTaylorCoefs02}, are the coincidence limits of covariant derivatives acting on both $\Sigma(\x,\x')$ and $A(\x,\x')$. We followed \cite{DeWittBrehme} to compute the required coincidence limits for $\Sigma(\x,\x')$, and included them in appendix \ref{AppSigma}. In the following, we will simply quote said results for better readability.

When considering surface Taylor series like \eqref{HypApproxS0}, most of the results from covariant Taylor series translate directly in terms of its surface counterparts, i.e., we have analog expressions for the \textit{coefficients} in \eqref{CovTaylorCoefs02} by taking limits of surface-tangent derivatives $D_a$ instead of spacetime covariant derivatives $\nabla_a$. These are computed in appendix \ref{App:3p1SigmaLimits}, where coincidence limits of successive derivatives of $\Sigma(\x,\x')$ in $3+1$ form are set equal to the $3+1$ expansions of the corresponding limits for the covariant derivatives of $\Sigma(\x,\x')$ computed in \ref{AppSigma}. The general procedure is the following: Suppose we have an expression for the limits of up to two derivatives of a biscalar $A(\x,\x')$. Let us omit the explicit indication of the arguments in the functions involved as these can be inferred from the context. According to the $3+1$ formalism, the coincidence limit of a first derivative for $A$ is expressed as
\begin{equation}
	[\nabla_a A] = [D_a A] + n_a [\dbar A],\label{DcovFull3p1}
\end{equation}
where we have used the notation introduced in Appendix \ref{App:3p1} for normal derivatives, 
\begin{equation}
	\dbar A_{a_k \dots a_0}\equiv - h_{a_k}\,^{a'_k} \dots h_{a_0}\,^{a'_0} n^b \nabla_b A_{a'_k \dots a'_0}.
\end{equation}
Given the limit $\mathcal{A}_a \equiv [\nabla_a A]$, we can compute the derivatives $[D_a A]$ and $[\dbar A]$ by writing the full $3+1$ decomposition of $\mathcal{A}_a$,
\begin{equation}
	\mathcal{A}_a = \overset{0}{\mathcal{A}}_a + n_a \overset{1}{\mathcal{A}}, \label{CurlyA}
\end{equation}
where
\begin{subequations}
\begin{align}
	\overset{0}{\mathcal{A}}_a &\equiv h_{a}\,^{a'} \mathcal{A}_{a'},\\
	\overset{1}{\mathcal{A}} &\equiv -n_a n^{a'} \mathcal{A}_{a'},
\end{align}
\end{subequations}
and identifying the terms in \eqref{DcovFull3p1} with those in \eqref{CurlyA} to obtain
\begin{subequations} \label{FirstDerivativeLimits}
\begin{align}
	[D_a A] &= \overset{0}{\mathcal{A}}_a, \\
	[\dbar A] &= \overset{1}{\mathcal{A}}.
\end{align}
\end{subequations}
However, this result is deceptively simple. In fact at the next order of derivatives, we have from \eqref{SecondDerPhi3p1},
\begin{align}
	\nabla_{a_1} \nabla_{a_0} A &= (K_{a_1 a_0} \dbar A + D_{a_1} D_{a_0} A) + n_{a_0} \left(D_{a_1} \dbar A + K_{a_1}\,^{b} D_b A \right) \nonumber \\
	& \quad + n_{a_1} \left(D_{a_0} \dbar A + K_{a_0}\,^{b} D_b A \right) + n_{a_0} n_{a_1} \left(\dbar^2 A + u^b D_b A \right), \label{DDcovFull3p1}
\end{align}
where $u_a \equiv \dbar n_a$. Then, upon taking the limit and considering all the $3+1$ projections of $\mathcal{B}_{a_1 a_0}\equiv [\nabla_{a_1} \nabla_{a_0} A]$, we get the following system of equations,
\begin{subequations}\label{CurlyBProjections}\begin{align}
	\overset{0}{\mathcal{B}}_{a_1 a_0} &= K_{a_1 a_0} [\dbar A] + [D_{a_1} D_{a_0} A], \label{CurlyB0}\\
	\overset{1}{\mathcal{B}}_{a_0} &= [D_{a_1} \dbar A] + K_{a_1}\,^{b} [D_b A] \\
	\overset{2}{\mathcal{B}}_{a_1} &= [D_{a_0} \dbar A] + K_{a_0}\,^{b} [D_b A] \\
	\overset{3}{\mathcal{B}} &=[\dbar^2 A] + u^b [D_b A],
\end{align}\end{subequations}
where we have used the notation defined in \ref{App:3p1} for the $3+1$ \textit{projection components} of $\mathcal{B}_{a_1 a_0}$.

Therefore, for each projection component of $\mathcal{B}_{a_1 a_0}$ in \eqref{CurlyBProjections} we have to solve for the \textit{unknown} limit of derivatives within each projection. For example, in \eqref{CurlyB0}, the unknown limit is $[D_{a_1} D_{a_0} A]$, as we already know $[\dbar A]$ from \eqref{FirstDerivativeLimits}, and $\overset{0}{\mathcal{B}}_{a_1 a_0}$ is the tangent projection of the known limit $\mathcal{B}_{a_1 a_0}$. Thus, we obtain,
\begin{subequations}\label{SecondDerivativeLimits}
	\begin{equation}
		[D_{a_1} D_{a_0} A] = \overset{0}{\mathcal{B}}_{a_1 a_0} - K_{a_1 a_0} \overset{1}{\mathcal{A}},
	\end{equation}
	and from the remaining $3+1$ \textit{projection components}, the corresponding limits of derivatives are
	\begin{align}
		[D_{a} \dbar A] &= \overset{1}{\mathcal{B}}_{a} - K_{a}\,^{b} \overset{0}{\mathcal{A}}_b,\\
		[\dbar^2 A] &= \overset{3}{\mathcal{B}} - u^b \overset{0}{\mathcal{A}}_b.
	\end{align}
\end{subequations}

We now recall the results for the limits of the covariant derivatives of $\Sigma$ computed in appendix \ref{AppSigma}. There we have shown that $[\nabla_a \Sigma] =0 $ \eqref{LimDSigma}, so that \eqref{FirstDerivativeLimits} yield
\begin{subequations}
\begin{align}
	[D_a \Sigma] &= 0,\\
	[\dbar \Sigma] &= 0.
\end{align}
\end{subequations}
We also have $[\nabla_{a_1} \nabla_{a_0}\Sigma] = g_{a_1 a_0}$ \eqref{LimDDSigma}. The complete $3+1$ expansion of $g_{ab}$ is given by \eqref{InducedMetric}, and therefore
	\begin{equation}
		[\nabla_{a_1} \nabla_{a_0}\Sigma] = h_{a_1 a_0}-n_{a_1} n_{a_0},
	\end{equation}
so that \eqref{SecondDerivativeLimits} yield
\begin{subequations}
\begin{align}
	[D_{a_1} D_{a_0}\Sigma] &= h_{a_1 a_0},\\
	[D_{a_1} \dbar \Sigma] &= 0,\\
	[\dbar^2 \Sigma] &= -1.
\end{align}
\end{subequations}
In appendix \ref{App:3p1SigmaLimits} we compute all the required projections and solve for the limits of all the tangent, normal and mixed derivatives included in the $3+1$ expansions of the coincidence limits $[\nabla_{a_k} \dots \nabla_{a_0} \Sigma]$.

With these results at hand, it follows that the first \textit{coefficients} in \eqref{HypApproxS0} are
\begin{subequations}
\begin{align}
	\overset{0}{S} &= 0,\\
	\overset{1}{S}_a &= 0,\\
	\overset{2}{S}_{a_1 a_0} &= h_{a_1 a_0},
\end{align}
\end{subequations}
and then we can rewrite \eqref{HypApproxS0} as
\begin{align} \label{HypApproxS2}
		\tilde{\Sigma} =& \sigma + \frac{1}{6} \overset{3}{S}^{a_2 a_1 a_0} (D_{a_2} \sigma) (D_{a_1} \sigma) (D_{a_0} \sigma) + \frac{1}{24} \overset{4}{S}^{a_3 a_2 a_1 a_0} (D_{a_3} \sigma) (D_{a_2} \sigma) (D_{a_1} \sigma) (D_{a_0} \sigma) \nonumber \\
		& + \frac{1}{5!} \overset{5}{S}^{a_4 a_3 a_2 a_1 a_0} (D_{a_4} \sigma) (D_{a_3} \sigma) (D_{a_2} \sigma) (D_{a_1} \sigma) (D_{a_0} \sigma) \nonumber\\
		&+ \frac{1}{6!} \overset{6}{S}^{a_5 a_4 a_3 a_2 a_1 a_0} (D_{a_5} \sigma) (D_{a_4} \sigma) (D_{a_3} \sigma) (D_{a_2} \sigma) (D_{a_1} \sigma) (D_{a_0} \sigma),
\end{align}
where we have used the differential equation satisfied by $\sigma$ on $\mathcal{C}$, analog to \eqref{pdeSigma},
\begin{equation}
	h^{ab}(D_a \sigma) (D_b \sigma) = 2 \sigma.
\end{equation}

Taking the coincidence limits of three tangent derivatives of \eqref{HypApproxS2} yields
\begin{equation}
	[D_{a_2}D_{a_1}D_{a_0}\Sigma] = [D_{a_2}D_{a_1}D_{a_0}\sigma] + \overset{3}{S}_{a_2 a_1 a_0},
\end{equation}
Where, by analogy with \eqref{LimDDDSigma}, $[D_{a_2}D_{a_1}D_{a_0}\sigma]=0$ and by \eqref{LimD3Sigma} we find
\begin{equation}
	\overset{3}{S}_{a_2 a_1 a_0} = 0.
\end{equation}

At the fourth order limit we arrive to
\begin{equation}
	[D_{a_3}D_{a_2}D_{a_1}D_{a_0}\Sigma] = [D_{a_3} D_{a_2}D_{a_1}D_{a_0}\sigma] + \overset{4}{S}_{a_3 a_2 a_1 a_0},
\end{equation}
and with aid of \eqref{LimD4Sigma} and the surface version of \eqref{limSigma4}, we obtain
\begin{equation}
	\overset{4}{S}_{a_3 a_2 a_1 a_0} = K_{(a_3 a_2 }K_{ a_1 a_0)}. \label{4thOrderTermSigma}
\end{equation}

At the fifth order, we have
\begin{equation}
	\overset{5}{S}_{a_4 a_3 a_2 a_1 a_0} = [D_{a_4}D_{a_3}D_{a_2}D_{a_1}D_{a_0}\Sigma] - [D_{a_4}D_{a_3}D_{a_2}D_{a_1}D_{a_0}\sigma] - 5 D_{(a_4} \overset{4}{S}_{a_3 a_2 a_1 a_0)},
\end{equation}
where, by means of \eqref{Lim5} and the surface version of \eqref{LimD5Sigma}, we find
\begin{equation}
	[D_{a_4}D_{a_3}D_{a_2}D_{a_1}D_{a_0}\Sigma] - [D_{a_4}D_{a_3}D_{a_2}D_{a_1}D_{a_0}\sigma] = 5 D_{(a_4} K_{a_3 a_2} K_{a_1 a_0)},
\end{equation}
which means that
\begin{equation}
	\overset{5}{S}_{a_4 a_3 a_2 a_1 a_0} = -5 (D_{(a_4} K_{a_3 a_2}) (K_{a_1 a_0)}). \label{5thOrderTermSigma}
\end{equation}

At the sixth order, we have
\begin{align}
	{\overset{6}{S}}_{(a_5 a_4 a_3 a_2 a_1 a_0)}&= [D_{a_5}D_{a_4}D_{a_3} D_{a_2} D_{a_1}D_{a_0} \Sigma] - [D_{a_5}D_{a_4}D_{a_3} D_{a_2} D_{a_1}D_{a_0} \sigma ] \nonumber\\
	&\qquad - 15 D_{(a_5} D_{a_4}  \overset{4}{S}_{a_3 a_2 a_1 a_0)} - 6 D_{(a_5}\overset{5}{S}_{a_4 a_3 a_2 a_1 a_0)},
\end{align}
where substituting \eqref{4thOrderTermSigma}, \eqref{5thOrderTermSigma}, the surface version of \eqref{LimD6Sigma} and the limit \eqref{MonsterLim6}, we obtain
\begin{align}
	{\overset{6}{S}}_{(a_5 a_4 a_3 a_2 a_1 a_0)} =& -3 K_{(a_5 a_4 }K_{a_3 a_2}\dbar K_{a_1 a_0 )}+9K_{(a_5a_4 } D_{a_3} D_{a_2} K_{ a_1 a_0)} +8\left(D_{(a_5 }K_{a_4 a_3} \right) \left( D_{a_2} K_{ a_1a_0)}\right) \nonumber\\
	&\, + 3 K_{(a_5a_4 }K_{a_3a_2} D_{a_1} u_{a_0)} + 4 K_{(a_5 a_4 } K_{a_3 a_2} K_{a_1}\,^b K_{ b a_0)} - 3 K_{(a_5 a_4 } K_{ a_3 a_2 } u_{a_1} u_{a_0)}. \label{6thOrderTermSigma}
\end{align}
This means that the required approximation for $\tilde{\Sigma}$ includes at most second order \textit{time} derivatives of the intrinsic metric $h_{ab}$ (cf. Appendix \ref{App:3p1LieTime}).

\subsection{Surface expansion for normal derivatives of \texorpdfstring{$\Sigma$}{S}}\label{sec:SigmaDerivatives}

As previously discussed, the biscalars $\dbar \Sigma$ and $\dbar^2 \Sigma$ cannot be computed as \textit{formal normal derivatives} of \eqref{HypApproxS2} because $\sigma$ is not defined for points in different hypersurfaces. However, both $\dbar \Sigma$ and $\dbar^2 \Sigma$ are biscalar functions that, if evaluated at points within the same hypersurface, admit a surface Taylor expansion just like $\Sigma$. In this section we compute such approximations.

We propose that the approximation for $\dbar \Sigma$ be given by the following surface Taylor expansion,
\begin{align}
	\dbar\tilde{\Sigma} =& \overset{0}{\mathsf{A}} + \overset{1}{\mathsf{A}}^{a} D_a \sigma + \frac{1}{2!}\overset{2}{\mathsf{A}}^{a_1 a_0} (D_{a_1}\sigma) (D_{a_0}\sigma) + \frac{1}{3!}\overset{3}{\mathsf{A}}^{a_2 a_1 a_0} (D_{a_2}\sigma)(D_{a_1}\sigma) (D_{a_0}\sigma) \nonumber\\ 
	&+\frac{1}{4!}\overset{4}{\mathsf{A}}^{a_3 a_2 a_1 a_0}(D_{a_3}\sigma) (D_{a_2}\sigma)(D_{a_1}\sigma) (D_{a_0}\sigma) +\frac{1}{5!}\overset{5}{\mathsf{A}}^{a_4 a_3 a_2 a_1 a_0}(D_{a_3}\sigma) (D_{a_2}\sigma)(D_{a_1}\sigma)(D_{a_0}\sigma).\label{ExpansiondbarSigma}
\end{align}

Following a procedure completely analog to the one we used to compute the approximation $\tilde{\Sigma}$, we find 
\begin{equation}
	\overset{0}{\mathsf{A}}=0,
\end{equation}
due to \eqref{LimDNSigma}. Similarly, due to \eqref{DdbarSigma},
\begin{equation}
	\overset{1}{\mathsf{A}}_a = 0.
\end{equation}
The first nonvanishing term is
\begin{equation}
	\overset{2}{\mathsf{A}}_{a_1 a_0} = - K_{a_1 a_0}, \label{A2dsigma}
\end{equation}
due to \eqref{LimD2dbarSigma}. 

From the surface analogue of \eqref{Taylor3} and \eqref{LimD3dbarSigma} we have
\begin{equation}
	\overset{3}{\mathsf{A}}_{a_2 a_1 a_0} =D_{(a_2} K_{a_1 a_0)}. \label{A3dsigma}
\end{equation}
For the fourth order term we have
\begin{equation}
	\overset{4}{\mathsf{A}}_{a_3 a_2 a_1 a_0} = [D_{a_3} D_{a_2} D_{a_1} D_{a_0} \dbar \Sigma] - 6 D_{(a_3} D_{a_2} \overset{2}{\mathsf{A}}_{a_1 a_0)} - 4 D_{(a_3} \overset{3}{\mathsf{A}}_{a_2 a_1 a_0)}. \label{Taylor4dS}
\end{equation}
Substituting \eqref{A2dsigma}, \eqref{A3dsigma} and \eqref{D4dbarSigma} on \eqref{Taylor4dS} we obtain 
\begin{align} 
	\overset{4}{\mathsf{A}}_{a_3 a_2 a_1 a_0} =& K_{(a_3 a_2} \dbar K_{a_1 a_0)} -D_{(a_3} D_{a_2} K_{a_1 a_0)} - K_{(a_3 a_2} D_{a_1} u_{a_0)}\nonumber\\
	&\quad - 2 K_{b(a_3} K^{b}\,_{a_2} K_{a_1 a_0)} + K_{(a_3 a_2} u_{a_1} u_{a_0)}. \label{A4dsigma}
\end{align}

For the fifth coefficient, we have 
\begin{align}
	{\overset{5}{\mathsf{A}}}_{a_4 a_3 a_2 a_1 a_0} = &[D_{a_4} D_{a_3} D_{a_2}D_{a_1}D_{a_0} \dbar \Sigma] -10 D_{(a_4} D_{a_3} D_{a_2} \overset{2}{\mathsf{A}}_{a_1 a_0)} -10 D_{(a_4}D_{a_3} \overset{3}{\mathsf{A}}_{a_2 a_1 a_0)} \nonumber\\
	&\, - 5D_{(a_4}\overset{4}{\mathsf{A}}_{a_3 a_2 a_1 a_0)} ,\label{PREdbarSigma5}
\end{align}
where substituting \eqref{A2dsigma}, \eqref{A3dsigma} and \eqref{A4dsigma} yields
\begin{align}
	{\overset{5}{\mathsf{A}}}_{a_4 a_3 a_2 a_1 a_0} = & -2 K_{ (a_4 a_3 } D_{a_2}\dbar K_{a_1 a_0 )} - \frac{7}{3} \dbar K_{ (a_4 a_3 } D_{a_2}K_{ a_1 a_0) } + D_{( a_4 }D_{a_3}D_{a_2}K_{ a_1 a_0 )} \nonumber\\
	&\quad + 2 K_{( a_4 a_3 } D_{a_2}D_{a_1}u_{ a_0 ) } + \frac{16}{3} K_{ (a_4 }\,^b K_{|b|a_3} D_{a_2}K_{ a_1 a_0) } - \frac{5}{3} K_{( a_4 a_3 } K_{a_2}\,^b D_{|b|}K_{ a_1 a_0) } \nonumber\\
	&\quad + \frac{26}{3} K_{ (a_4 a_3 } K_{a_2}\,^b D_{a_1} K_{ a_0) b} + \frac{7}{3} (D_{ (a_4 }u_{a_3})(D_{a_2} K_{ a_1 a_0) }) - \frac{7}{3} u_{ (a_4 }u_{a_3} D_{a_2} K_{ a_1 a_0) }\nonumber\\
	&\quad - 4 K_{ (a_4 a_3 }u_{a_2}D_{a_1}u_{ a_0) }. \label{dbarSigma5}
\end{align}

Now we propose as a fourth-order approximation for $\dbar^2 \Sigma$ the following surface expansion,
\begin{align}
	\dbar^2\tilde{\Sigma} =& {\overset{0}{\mathsf{B}}} + {\overset{1}{\mathsf{B}}}^a D_a \sigma + \frac{1}{2!}{\overset{2}{\mathsf{B}}}^{a_1 a_0} (D_{a_1} \sigma)(D_{a_0} \sigma) + \frac{1}{3!}{\overset{3}{\mathsf{B}}}^{a_2 a_1 a_0} (D_{a_2} \sigma)(D_{a_1} \sigma)(D_{a_0} \sigma) \nonumber\\
	&\quad + \frac{1}{4!}{\overset{4}{\mathsf{B}}}^{a_3 a_2 a_1 a_0} (D_{a_3} \sigma)(D_{a_2} \sigma)(D_{a_1} \sigma)(D_{a_0} \sigma). \label{ExpansionddbarSigma}
\end{align}

The first term is
\begin{equation}
	{\overset{0}{\mathsf{B}}} = -1,
\end{equation}
due to \eqref{ddbarSigma}. The next term involves the coincidence limit \eqref{LimDddbarSigma},
\begin{equation}
	{\overset{1}{\mathsf{B}}}_a = -u_a. \label{B1}
\end{equation}
The second order Taylor coefficient comes from \eqref{LimDDddbarSigma},
\begin{equation}
	{\overset{2}{\mathsf{B}}}_{a_1 a_0} = \frac{2}{3}\left( - \dbar K_{a_1 a_0} + K_{a_1}\,^{b} K_{b a_0} - u_{a_1} u_{a_0}  + D_{a_1} u_{a_0}\right) .
\end{equation}
By taking the coincidence limit of the third derivative of \eqref{ExpansionddbarSigma} we get from (the surface analog of) \eqref{Taylor3},\eqref{B1}, and \eqref{LimD3dbar2Sigma},
\begin{equation}
	{\overset{3}{\mathsf{B}}}_{a_2 a_1 a_0} = [D_{a_2}D_{a_1}D_{a_0}\dbar^2 \Sigma] - 3 D_{(a_2}D_{a_1} {\overset{1}{\mathsf{B}}}_{ a_0)} - 3 D_{(a_2}  {\overset{2}{\mathsf{B}}}_{a_1 a_0)},
\end{equation}
Substituting our previous results,
\begin{align}
	{\overset{3}{\mathsf{B}}}_{a_2 a_1 a_0} = & \frac{1}{2} D_{(a_2} \dbar K_{a_1 a_0)} - \frac{1}{2} D_{(a_2} D_{a_1} u_{a_0)} + u_{(a_2} D_{a_1} u_{a_0)} + (D_{b}K_{(a_2 a_1})K^b\,_{a_0)} - 2 K^b\,_{(a_2} D_{a_1} K_{a_0) b} \nonumber\\
	&\,- K_{(a_2 a_1} K^b\,_{a_0)} u_b.
\end{align}

At fourth order we have 
\begin{align}
	{\overset{4}{\mathsf{B}}}_{a_3 a_2 a_1 a_0} =& \frac{17}{5} K_{(a_3 a_2} \dbar^2 K_{a_1 a_0)}+ \frac{12}{5} D_{(a_3} D_{a_2} \dbar K_{a_1 a_0)}- \frac{17}{5} K_{(a_3 a_2} D_{a_1} \dbar u_{a_0)}+ \frac{24}{5} K_{b(a_3} K^b\,_{a_2} \dbar K_{a_1 a_0)} \nonumber\\
	&\,- \frac{214}{15} K_{(a_3 a_2}K_{a_1}\,^b \dbar K_{a_0) b}+ \frac{51}{5} K_{(a_3 a_2}u_{a_1} \dbar u_{a_0)}+ 28 K_{(a_3}\,^{b_1} K_{a_2}\,^{b_2} R^{(3)}_{a_1 |b_1| a_0 ) b_0} \nonumber\\
	&\,-\frac{72}{5} K_{(a_3}\,^b D_{a_2} D_{a_1} K_{a_0)b}+\frac{48}{5} K_{b(a_3} D^{b} D_{a_2} K_{a_1 a_0)}-\frac{48}{5} (D_{(a_3}K_{a_2}\,^b) (D^{a_1} K_{a_0)b}) \nonumber\\
	&\,+\frac{24}{5} (D_{(a_3}K_{a_2}\,^{b}) (D_{|b|} K_{a1 a_0)})+\frac{46}{5} K_{(a_3 a_2} u^b D_{|b|}K_{a_1 a_0)}-\frac{41}{5} K_{(a_3 a_2} u^b D_{a_1}K_{a_0)b} \nonumber\\
	&\,+2 K_{b(a_3} u^b D_{a_2}K_{a_1 a_0)}-\frac{12}{5} D_{(a_3} D_{a_2} D_{a_1} u_{a_0)}+\frac{24}{5} u_{(a_3} D_{a_2} D_{a_1} u_{a_0)} \nonumber\\
	&\,+\frac{24}{5} (D_{(a_3} u_{a_2}) (D_{a_1} u_{a_0)})-\frac{24}{5} K_{b(a_3} K^b\,_{a_2} D_{a_1} u_{a_0)}+\frac{61}{15} K^b\,_{a_3} k_{a_2 a_1} D_{|b|} u_{a_0)} \nonumber\\
	&\,-4 K^b\,_{(a_3} K_{a_2 |b|} R^{(3)}\,_{a_1 a_0)}+\frac{24}{5} K^b\,_{(a_3} K_{a_2 |b|} u_{a_1} u_{a_0)}-\frac{61}{15} K^b\,_{(a_3} K_{a_2 a_1} u_{a_0)} u_{b} \nonumber\\
	&\,-\frac{17}{5} u^2 K_{(a_3 a_2} K_{a_1 a_0)}+\frac{28}{5} K^{b_1}\,_{b_0} K_{b_1 (a_3} K^{b_0}\,_{a_2} K_{a_1 a_0)}-\frac{48}{5} K^{b_1}\,_{(a_3} K_{|b_1| a_2} K^{b_0}\,_{a_1} K_{|b_0| a_0)} \nonumber\\
	&\,+\frac{20}{3} K K^{b}\,_{(a_3} K_{a_2 a_1} K_{a_0)b}.
\end{align}

Following Theorem \ref{th:SigmaSurfaceApprox}, we also need approximations for $\dbar'\Sigma$ and $\dbar \dbar' \Sigma$ up to orders $5$ and $4$, respectively. In order to build these, we follow the same procedure used in the case of  $\dbar\Sigma$ and $\dbar^2\Sigma$, just considering the coincidence limits for tangent derivatives of $\dbar'\Sigma$ and $\dbar \dbar' \Sigma$ up to fifth and fourth order. In order to get these limits, it is necessary to compute first the coincidence limits for up to five covariant derivatives of $\nabla_{c'}\Sigma$ at $\x$. This computation is subtle, and is carried on in appendix \ref{LimitsMixedDerivativesSigma}. Afterwards, these limits need to be expanded in $3+1$ form to finally obtain the required derivatives, which is done in appendix \ref{3p1Hybrid}. 

Let us denote by $\zeta$ the surface approximation for $\dbar'\Sigma$ to be given by
\begin{align}
	\zeta =& \frac{1}{2!}\overset{2}{\mathsf{z}}^{a_1 a_0} (D_{a_1}\sigma) (D_{a_0}\sigma) + \frac{1}{3!}\overset{3}{\mathsf{z}}^{a_2 a_1 a_0} (D_{a_2}\sigma)(D_{a_1}\sigma) (D_{a_0}\sigma) \nonumber\\ 
	&+\frac{1}{4!}\overset{4}{\mathsf{z}}^{a_3 a_2 a_1 a_0}(D_{a_3}\sigma) (D_{a_2}\sigma)(D_{a_1}\sigma) (D_{a_0}\sigma) +\frac{1}{5!}\overset{5}{\mathsf{z}}^{a_4 a_3 a_2 a_1 a_0}(D_{a_3}\sigma) (D_{a_2}\sigma)(D_{a_1}\sigma)(D_{a_0}\sigma).\label{ApproxDbarPrimeSigma}
\end{align}
By taking the coincidence limits $[\zeta]=0$ and $[D_{a_1}\zeta]=0$ we verify they match those of $[\dbar'\Sigma]=0$ and $[D_{a_1}\dbar'\Sigma]=0$ (cf. \eqref{DprimeSigma} and \eqref{3p1HybridDerivativesSigma1}). The coefficients then are,
\begin{subequations}
	\begin{align}
		\overset{2}{\mathsf{z}}_{a_1 a_0} &= -K_{a_1 a_0},\\
		\overset{3}{\mathsf{z}}_{a_2 a_1 a_0} &= 2 D_{(a_2} K_{a_1 a_0)} ,\\
		\overset{4}{\mathsf{z}}_{a_3 a_2 a_1 a_0} &= \frac{1}{2} K_{(a_3 a_2} \left(-3 K^{b_0}\,_{a_1} K_{|b_0|a_0)}+u_{a_1} u_{a_0)} +\dbar K_{a_1 a_0)} - D_{a_1} u_{a_0)} \right)-3
		D_{(a_3} D_{a_2} K_{a_1 a_0)},\\
		\overset{5}{\mathsf{z}}_{a_4 a_3 a_2 a_1 a_0} &= -\frac{4}{3} \dbar K_{(a_4 a_3} D_{a_2} K_{a_1 a_0)} + \mathfrak{Z}_{a_4 a_3 a_2 a_1 a_0},
	\end{align}
\end{subequations}
where $\mathfrak{Z}_{a_4 a_3 a_2 a_1 a_0}$ is a symmetric tensor that does not include normal derivatives. 

For $\dbar\dbar'\Sigma$, let us consider the surface approximation
\begin{align}
	\chi =& \overset{0}{\mathsf{C}}+ \overset{1}{\mathsf{C}}^{a_0} D_{a_0}\sigma + \frac{1}{2!}\overset{2}{\mathsf{C}}^{a_1 a_0} (D_{a_1}\sigma) (D_{a_0}\sigma) + \frac{1}{3!}\overset{3}{\mathsf{C}}^{a_2 a_1 a_0} (D_{a_2}\sigma)(D_{a_1}\sigma) (D_{a_0}\sigma) \nonumber\\ 
	&+\frac{1}{4!}\overset{4}{\mathsf{C}}^{a_3 a_2 a_1 a_0}(D_{a_3}\sigma) (D_{a_2}\sigma)(D_{a_1}\sigma) (D_{a_0}\sigma) . \label{ApproxDbarDbarPrimeSigma}
\end{align}
The coefficents in this case are
\begin{subequations}
	\begin{align}
		\overset{0}{\mathsf{C}} =& 1,\\
		\overset{1}{\mathsf{C}}_{a_0} =& 0,\\
		\overset{2}{\mathsf{C}}_{a_1 a_0} =& \frac{1}{6}\left(7 K_{b_0 (a_1}K^{b_0}\,_{a_0)} - u_{a_1} u_{a_0}-\dbar K_{a_1 a_0} + D_{a_1} u_{a_0} \right),\\
		\overset{3}{\mathsf{C}}_{a_2 a_1 a_0} =&\frac{1}{6} \left(4 D_{(a_2} \dbar K_{a_1 a_0)} - 31 K_{b_0 (a_2} D_{a_1} K^{b_0}\,_{a_0)} +8 u_{(a_2} D_{a_1} u_{a_0)} \right.\nonumber\\ 
		&\qquad \left.-4 D_{(a_2} D_{a_1} u_{a_0)} +5 D_{b_0} K_{(a_2 a_1} K^{b_0}\,_{a_0)} \right) ,\\
		\overset{4}{\mathsf{C}}_{a_3 a_2 a_1 a_0} =& -\frac{1}{15} \Big(K_{(a_3 a_2} \dbar^2 K_{a_1 a_0)} + 22 K_{b_0 (a_3} K^{b_0}\,_{a_2} \dbar K_{a_1 a_0)} +K_{(a_3 a_2} K^{b_0}\,_{a_1} \dbar K_{a_0) b_0} \nonumber\\
		&\qquad\qquad +3 K_{(a_3 a_2} u_{a_1} \dbar u_{a_0)} \Big) + \mathfrak{C}_{a_3 a_2 a_1 a_0}
	\end{align}
\end{subequations}
where $\mathfrak{C}_{a_3 a_2 a_1 a_0}$ is a symmetric tensor that does not include normal derivatives. 

Therefore, according to Theorem \ref{th:SigmaSurfaceApprox}, we have provided all the pieces necessary to compute the exact expectation value of the renormalized stress-energy tensor at the \textit{initial surface} $\mathcal{C}$.

\section{Connection to the initial value problem in semiclassical gravity} 
\label{sec:SurfaceHadamard}

We now turn to discuss the relation between the present work and Conjecture 3.7 of \cite{JKMS2023:Programmatic} on the well-posedness of the initial value problem for semiclassical gravity. We begin by providing the following refined version of Theorem \ref{th:SigmaSurfaceApprox}.

\begin{thm}\label{th:SigmaSurfaceApproxTight}
	Let $\x$ and $\x'$ be points contained on a spacelike surface $\mathcal{C}$ of the spacetime $(M,g_{ab})$ and also contained in a convex normal neighborhood $\mathcal{D}$. One can construct approximations {\bf (a)} $\tilde{U}(\x,\x')$, $\dbar \tilde{U}(\x,\x')$, $\dbar \dbar' \tilde{U}(\x,\x')$ and $\dbar^2 \tilde{U}(\x,\x')$ for the Hadamard coefficient $U(\x,\x')$ and its corresponding normal derivatives, {\bf (b)} $\tilde{V}(\x,\x')$, $\dbar \tilde{V}(\x,\x')$, $\dbar \dbar' \tilde{V}(\x,\x')$ and $\dbar^2 \tilde{V}(\x,\x')$ for the Hadamard coefficient $V(\x,\x')$ and its corresponding normal derivatives, and {\bf (c)}  $\tilde{\Sigma}$, $\dbar \tilde{\Sigma}$, $\dbar \dbar'\tilde{\Sigma}$ and $\dbar^2 \tilde{\Sigma}$ for $\Sigma(\x,\x')$, $\dbar \Sigma (\x,\x')$, $\dbar \dbar' \Sigma (\x,\x')$ and $\dbar^2 \Sigma(\x,\x')$, using the intrinsic metric on $\mathcal{C}$ and up to four ``time derivatives" of the metric off the surface, such that the surface-defined approximate $\tilde{H}^\ell$ defines correctly the expectation value of the renormalized stress-energy tensor on $\mathcal{C}$ by Eq. \eqref{TrenByApprox2}.
\end{thm}

Theorem \ref{th:SigmaSurfaceApproxTight} is a precise statement of the ``second issue" discussed for the initial value problem of semiclassical gravity in \cite{JKMS2023:Programmatic} (on p.\ 31 of the published version).   Following Conjecture 3.7 in that paper, the point for the initial value problem of semiclassical gravity is that the ``missing" piece of initial data, i.e.\ the fourth ``time" derivative off the initial surface, ought to be obtained by imposing the semiclassical Einstein equations on the initial surface. Indeed, obtaining this piece of data is precisely the ``zeroth stage" alluded to in Conjecture 3.7 in \cite{JKMS2023:Programmatic}, and yields the ``crudest" approximation to the surface Hadamard condition, Def. 3.5 in \cite{JKMS2023:Programmatic}, which is in turn defined by all stages ($n = 0, 1, \ldots$) of that conjecture being satisfied.\footnote{One could stop at the crudest zeroth (or any later finite) stage at the cost of obtaining semiclassical gravity solutions with lower differentiability, see Footnote 5 in \cite{JKMS2023:Programmatic}. One would expect that the spacetimes in these solutions would be $C^4$ in the crudest approximation, and the states would be adiabatic states of finite order, in the sense of \cite{Junker-Schrohe}, with adiabatic order $k = 2$ in the crudest approximation.}

In terms of the developments in the present paper, Conjecture 3.7 in \cite{JKMS2023:Programmatic} can be recast in the following way:

\begin{conj} \label{IntrinsicInitialData}
	Self-consistent initial data for the fourth order \textit{time derivatives} of the metric can be obtained out of the \textit{classical} initial data $h_{ab}$, $K_{ab}$, $\dbar K_{ab}$ and $\dbar^2 K_{ab}$ on $\mathcal{C}$ (Equivalently in terms of \textit{time derivatives} of the metric, $g_{ab}|_{\mathcal{C}}$, $\dot g_{ab}|_{\mathcal{C}}$, $g^{(2)}_{ab}|_{\mathcal{C}}$, $g^{(3)}_{ab}|_{\mathcal{C}}$), by constructing the approximations for $\Sigma$, $\dbar \Sigma$, $\dbar^2 \Sigma$, $U$ and $V$, solving the fourth order terms (ie, $\nabla_{n}^{3} K_{a'b'}\sim g^{(4)}_{ab}|_{\mathcal{C}}$) using the semiclassical Einstein equation \eqref{SemiclassEinsteinEq}. Furthermore, higher order time derivatives of the induced metric can be obtained by formal derivation of the semiclassical Einstein equation by the same procedure.
\end{conj}

\section{Discussion}\label{sec:Discuss}

The initial value problem for semiclassical gravity is substantially more complicated than its classical counterpart  due to the fact  that renormalization is required for defining the expectation value of the stress-energy tensor, as one simultaneously solves for the semiclassical Einstein equations. In the Hadamard subtraction scheme -- suitable in curved spacetimes -- the renormalization procedure relies on removing the singular structure of the two-point function, which is typically characterized covariantly in terms of the {\it spacetime} geometry (and field parameters). This renders the initial-value formulation of semiclassical gravity highly non-trivial, since one only has {\it a priori} finitely many pieces of functional components  of the initial data on a spacelike surface, rather than the full spacetime geometry.

Motivated by the former, in this  paper we have studied in the context of a fixed background curved spacetime how the Hadamard property restricts to a (Cauchy) spacelike surface in terms of intrinsic geometric properties of the surface. While it is clear that the Hadamard property on an initial surface depends on the induced geometry and all possible normal derivatives of the induced metric off the surface, here we have provided a sufficiently accurate approximation to the Hadamard property that depends only on finitely many normal derivatives, such that a surface-renormalization prescription yields the exact expectation value of the stress-energy tensor in the surface (for states satisfying this approximate property).

In order to achieve this, we introduced a notion of \textit{order}-$n$ approximation for biscalars, and determined a sufficient order of approximation for the Hadamard coefficients, half the squared geodesic distance and suitable normal derivatives thereof in the initial surface in terms of the intrinsic surface geometry and only {\it finitely many normal derivatives} of the induced metric off the surface. This analysis relies crucially on obtaining a surface Taylor series for the spacetime geodesic distance between two points on the spacelike surface in terms of the  geodesic distance,   within the   surface,  between the said points.

The sufficient approximations computed in this work were expansions of order $4$ and $2$, respectively, for the regular biscalars $U(\x,\x')$ and $V(\x,\x')$ appearing on the Hadamard singular structure, and surface Taylor expansions for half the squared geodesic distance $\Sigma(\x,\x')$ and its normal derivatives,
$\dbar \Sigma (\x,\x')$, $\dbar' \Sigma (\x,\x')$, $\dbar \dbar' \Sigma (\x,\x')$ and $\dbar^2 \Sigma (\x,\x')$, of orders $6$, $5$, $5$, $4$ and $4$, respectively. The higher order terms of these expansions are fourth order normal derivatives of the metric. This can be troublesome given that semiclassical gravity, seen as an initial value problem, would seem to require initial data of the same derivative orders than the semiclassical Einstein equation, which is a fourth order system. In order to deal with this potential issue, we made contact with the \textit{surface Hadamard} condition given in \cite{JKMS2023:Programmatic}, which implies not only a property for data of the quantum field in the initial surface, but also the requirement that, if we are given data for a fourth order problem (ie., up to third order time derivatives), the \textit{missing} data for the metric normal derivatives can be obtained implicitly from the semiclassical Einstein equation. This sets a preliminary but necessary foundation for the study of the initial value problem of semiclassical gravity. Nevertheless, the initial value problem for semiclassical gravity remains unsolved, and it must be noted that our proposals to sort some of the immediate difficulties rely on (yet) unproven conjectures at this point.

In \cite{JKMS2023:Programmatic}, we have discussed other topics in semiclassical gravity, such as the possibility to implement a \textit{perturbative} expansion in $\hbar$ satisfying a classical correspondence limit criteria that might enable to compute solutions and corrections, order by order, out of \textit{classical} initial data $h_{ab}$ and $K_{ab}$. There we also incorporate a scheme to model quantum collapse in semiclassical gravity, as a possible path to address some conceptual issues of the theory, and even established some continuity conditions at order $\hbar^0$ for such setting, i.e., if quantum matter has \textit{jumps}, what jumps in the gravitational sector?

Further research is required in order to establish the general conditions for which semiclassical gravity represent a valid approach to describe the interaction between quantum matter and spacetime geometry.

\section{Acknowledgements}
B.A.J-A. was supported by a CONAHCYT (formerly CONACYT) Postdoctoral Fellowship held at the Instituto de Ciencias Nucleares of UNAM, Mexico.  The author is currently supported by \linebreak EPSRC, through the EPSRC Open Fellowship EP/Y014510/1, held at the Department of Mathematics of the University of York, UK. B.A.J.A, T.M. and D.S. acknowledge the support of CONAHCYT project 140630. B.A.J.-A. and D.S. acknowledge partial financial support from  PAPIIT-UNAM, Mexico (Grant No. IG100124 ). D.S. also thanks the Foundational Questions Institute (Grant No. FQXi-MGB-1928) and the Fetzer Franklin Fund, a donor advised by the Silicon Valley Community Foundation.

\appendix

\section{Half the Squared Geodesic Distance}\label{AppSigma}
In this section we reproduce a brief version of the procedure by DeWitt \& Brehme in \cite{DeWittBrehme}. Readers familiar with these techniques may jump straight to our results \eqref{LimD5Sigma} and \eqref{LimD6Sigma}.

Consider a globally hyperbolic spacetime $(M,g_{ab})$, and let $\mathcal{D}\subset M$ be an open convex normal neighborhood, i.e., for every pair of points $\x,\x' \in \mathcal{D}$, there is a unique geodesic from $\x$ to $\x'$ contained in $\mathcal{D}$. We denote by $s(\x,\x')$ the geodesic distance for points $\x,\x'\in \mathcal{D}$. Fixing $\x'=\x_0\in\mathcal{D}$, a field $s_{\x_0}(\x) \equiv s(\x,\x_0)$ can be defined on $\mathcal{D}$. Except for the lightcone of $\x_0$, the vector field $g^{ab}\nabla_b s_{\x_0}$ is normalized, given that $s_{\x_0}$ is the affine parameter for every geodesic emanating from $\x_0$, i.e.,
\begin{equation}
	g^{ab} \nabla_a s_{\x_0} \nabla_b s_{\x_0} = \pm 1,\label{UnitS}
\end{equation}
where the positive sign is for spacelike geodesics, and the negative sign is for timelike geodesics. We define half the squared geodesic distance by
\begin{equation}
	\Sigma (\x,\x') \equiv \pm \frac{1}{2} s^2(\x,\x'), \label{SigmaDef}
\end{equation}
where the plus sign is used for spacelike related $\x,\x'$, and the minus sign is used for timelike related $\x,\x'$. $\Sigma$ vanishes when the geodesic distance vanishes, in particular
\begin{equation}
	\lim_{\x'\to\x} \Sigma(\x,\x')=0. \label{coincidenceSigma}
\end{equation}
By fixing $\x'=\x_0$, a scalar field $\Sigma_{\x_0}(\x)\equiv \Sigma (\x,\x_0)$ in $\mathcal{D}$ is defined. In the following we will omit the explicit reference to $\x$ and $\x'=\x_0$. Derivation on $\x$ yields
\begin{equation}
	\nabla_a \Sigma = \pm s \nabla_a s. \label{dSigma}
\end{equation}
It then follows 
\begin{equation}
	\lim\limits_{\x \rightarrow \x_0} \nabla_a \Sigma = 0.\label{LimDSigma}
\end{equation}
From \eqref{dSigma}, \eqref{SigmaDef} and \eqref{UnitS} one can readily obtain the following differential equation for $\Sigma$,
\begin{equation}
	g^{ab} (\nabla_a \Sigma) (\nabla_b \Sigma) = 2 \Sigma, \label{pdeSigmaApp}
\end{equation}
which we recall as \eqref{pdeSigma}.
This equation allows to extract information on the derivatives of $\Sigma$. Using the abbreviated notation\footnote{Note that this notation is similar to the semicolon notation for derivatives, although in an inverse order: $$\nabla_a \nabla_b \phi=\phi_{ab}=\phi_{;ba}$$} $\Sigma_a = \nabla_a \Sigma$, $\Sigma_{ab} = \nabla_a \nabla_b \Sigma$, etc, one obtains
\begin{align}
	\Sigma_a &= \Sigma^b \Sigma_{ab},\\
	\Sigma_{ba} &= \Sigma_b\,^c \, \Sigma_{ac} + \Sigma^c \, \Sigma_{bac} \label{DDSigma},
\end{align}
so that Eq. \eqref{DDSigma} in the coincidence limit reads
\begin{equation}
	\lim\limits_{\x\rightarrow\x_0} \Sigma_{ba} = \lim\limits_{\x\rightarrow\x_0} \Sigma_b\,^c \, \Sigma_{ac},
\end{equation}
where it is inferred
\begin{equation}
	\lim\limits_{\x\rightarrow\x_0} \Sigma_{ab} = g_{ab}. \label{LimDDSigma}
\end{equation}
In the following we compute the coincidence limits of up to six derivatives of $\Sigma$ with help of Eq. \eqref{pdeSigmaApp} and the limits \eqref{coincidenceSigma}, \eqref{LimDSigma} and \eqref{LimDDSigma}, following DeWitt \& Brehme method devised in \cite{DeWittBrehme}. From now on, we will use square brackets $[\quad ]$ to indicate the coincidence limit $\x\rightarrow\x_0$. 

The coincidence limit of the covariant derivative of \eqref{DDSigma} yield
\begin{align}
	[\Sigma_{cba}] &= [\Sigma_{cba}] + [\Sigma_{cab}] + [\Sigma_{bac}]. \label{Master3}
\end{align}
Substituting in \eqref{Master3} the no-torsion condition $[\Sigma_{cba}]=[\Sigma_{cab}]$ and the limit  $[\Sigma_{bca}] = [\Sigma_{cba}]$ of the identity (cf. \eqref{RiemannDef})
\begin{align}
	\Sigma_{bca} &= R_{bca}\,^d \Sigma_d + \Sigma_{cba}, \label{RiemSig3}
\end{align}
yield
\begin{equation}
	[\Sigma_{cba}] = 0. \label{LimDDDSigma}
\end{equation}

In general, computing $m$ derivatives of \eqref{pdeSigmaApp}, and taking the coincidence limit yields an equation containing a combination of terms $[\Sigma_{a_m \dots a_0}]$ with switched indexes. Let us call this the \textit{master equation of order $m$}. In order to \textit{exchange} those indices, it is necessary to use the general rule in terms of the Riemann tensor,
\begin{equation}
	\nabla_{a_n} \nabla_{a_{n-1}} \omega_{a_{n-2}\dots a_1} -\nabla_{a_{n-1}} \nabla_{a_{n}} \omega_{a_{n-2}\dots a_1} = \sum_{k=1}^{n-2} R_{a_{n} a_{n-1} a_k }\,^{b} \, \omega_{a_{n-2} \dots a_{k+1} b\, a_{k-1} \dots a_1},
\end{equation}
at every order $2 < n \leq m$, and take $m-n$ more derivatives, resulting in $m-2$ index exchange equations (exchange of the first two indexes is \textit{free} given there is no torsion). Then, take the coincidence limit for every index exchange equation, obtaining a set of rules for the exchange of every two adjacent indexes for the terms $[\Sigma_{a_m \dots a_0}]$. These index exchange rules are then used to solve the limit $[\Sigma_{a_m \dots a_0}]$ from the master equation of order $m$. This summarize the method devised by DeWitt \& Brehme in \cite{DeWittBrehme}.
They have already computed this limits up to the fourth derivative, obtaining
\begin{equation}
	[\Sigma_{dbca}] = \frac{2}{3} R_{d (c a) b}. \label{limSigma4}
\end{equation}

Following their procedure, we get the following limits,
\begin{align}
	[\Sigma_{edcba}] &= \frac{1}{2}\left(\nabla_e R_{bcda} +\nabla_d R_{beca} + \nabla_c R_{bdea}\right) \nonumber\\
	&= \frac{3}{2} \nabla_{(e}R_{ |b| d c) a} \label{LimD5Sigma}
\end{align}

\begin{align}
	[\Sigma_{a_5 a_4 a_3 a_2 a_1 a_0}] &= -\frac{12}{5} \nabla_{(a_5} \nabla_{a_4} R_{a_3 | a_1 |a_2) a_0} + \frac{1}{45} \mathfrak{F}_{a_5 a_4 a_3 a_2 a_1 a_0}, \label{LimD6Sigma}
\end{align}
where
\begin{align}
	\mathfrak{F}_{a_5 a_4 a_3 a_2 a_1 a_0} =& 20 \Big( 2 R^{b}\,_{(a_1 a_0) (a_5} R_{a_4) (a_3 a_2) b} + R^{b}\,_{(a_1 a_0) a_3} R_{a_5 (a_2 a_4) b} + R^{b}\,_{(a_1 a_0) a_2} R_{a_5 (a_3 a_4) b} \Big) \nonumber  \\
	&+ \Bigg( 
	R_{a_5 a_4 a_0 }\,^b \Big(16 R_{b a_1 a_3 a_2 } + 17 R_{b a_3 a_2 a_1 } \Big) + R_{a_4 a_0 a_5 }\,^b \Big(17 R_{b a_1 a_3 a_2 } +  4 R_{b a_3 a_2 a_1 } \Big) \nonumber \\
	&\quad + R_{a_5 a_3 a_0 }\,^b \Big(16 R_{b a_1 a_4 a_2 } + 17 R_{b a_4 a_2 a_1 } \Big) + R_{a_3 a_0 a_5 }\,^b \Big(17 R_{b a_1 a_4 a_2 } +  4 R_{b a_4 a_2 a_1 } \Big) \nonumber \\
	&\quad + R_{a_5 a_2 a_0 }\,^b \Big(16 R_{b a_1 a_4 a_3 } + 17 R_{b a_4 a_3 a_1 } \Big) + R_{a_2 a_0 a_5 }\,^b \Big(17 R_{b a_1 a_4 a_3 } +  4 R_{b a_4 a_3 a_1 } \Big) \nonumber \\
	&\quad + R_{a_3 a_0 a_4 }\,^b \Big(17 R_{b a_1 a_5 a_2 } +  4 R_{b a_5 a_2 a_1 } \Big) + R_{a_4 a_3 a_0 }\,^b \Big(16 R_{b a_1 a_5 a_2 } + 17 R_{b a_5 a_2 a_1 } \Big) \nonumber \\
	&\quad + R_{a_2 a_0 a_4 }\,^b \Big(17 R_{b a_1 a_5 a_3 } +  4 R_{b a_5 a_3 a_1 } \Big) + R_{a_4 a_2 a_0 }\,^b \Big(16 R_{b a_1 a_5 a_3 } + 17 R_{b a_5 a_3 a_1 } \Big) \nonumber \\
	&\quad + R_{a_2 a_0 a_3 }\,^b \Big(17 R_{b a_1 a_5 a_4 } +  4 R_{b a_5 a_4 a_1 } \Big) + R_{a_3 a_2 a_0 }\,^b \Big(16 R_{b a_1 a_5 a_4 } + 17 R_{b a_5 a_4 a_1 } \Big)
	\Bigg).
\end{align}

\subsection{Limits of mixed derivatives} \label{LimitsMixedDerivativesSigma}
Since $\Sigma(\x,\x')$ is a two-point function (defined within a convex normal neighborhood), it is possible to compute derivatives on either points, which yield \textit{hybrid} tensors defined by sums of products of tensors defined at different points, like, for example, $\nabla_{a} \nabla_{c'} \Sigma(\x,\x') = \Sigma_{a c'}(\x,\x')$. We shall  compute the coincidence limits of this hybrid derivatives using  a procedure  analog to  that  employed when analysing the  derivatives of $\Sigma$ at a  single point.  Before   that one   must  establish  the  order of coincidence  limits that are required.  As $\Sigma(\x,\x')$ is symmetric under the  exchange of $\x$ and $\x'$,  we have  at hand the coincidence limits of the first six derivatives of $\Sigma$ at one point, regardless of whether it is at $\x$ or $\x'$. This allow us to set 
\begin{subequations}\label{FirstLimitHybrids}
	\begin{equation}
		[\Sigma_{c'}]=0.
	\end{equation}
Application of Synge's rule for $\nabla_a [\nabla_c \Sigma]$ yields,
\begin{equation}
	[\Sigma_{ac'}] = - g_{ac}.\label{limDDHybridSigma}
\end{equation}
\end{subequations}
We note that \textit{commutation} of derivatives at $\x$ and $\x'$ is trivial, as the respective tensor components are defined in their own tangent spaces: relative to the tangent space at $\x$, $\Sigma_{c'}$ can be regarded as a scalar, and only in the coincidence limit it becomes a vector at that space. This also means that the first two derivatives of $\Sigma_{c'}$ at $\x$ commute, i.e.
\begin{equation}
	\Sigma_{a_1 a_2 c'} = \Sigma_{a_2 a_1 c'}. \label{TorsionlessHybrid}
\end{equation}
 With this idea in mind, the relevant differential equation for $\Sigma_{c'}$ is obtained by taking a derivative of \eqref{pdeSigmaApp} at $\x'$,
\begin{equation} \label{pdeSigmaHybrid}
	\Sigma_{c'} = \Sigma_{b c'}\Sigma^{b}.
\end{equation}
While the direct computation of the limit of \eqref{pdeSigmaHybrid} and its first derivative is consistent with \eqref{FirstDerivativeLimits}, at the second derivative we get
\begin{equation}
	[\Sigma_{a_2 a_1 c'}] = [\Sigma_{a_2 a_1 c'}] + [\Sigma_{a_1 a_2 c'}].
\end{equation}
With \eqref{TorsionlessHybrid} we get
\begin{equation}
	[\Sigma_{a_2 a_1 a'_0}] = 0.
\end{equation}
Following the usual procedure of taking derivatives of \eqref{pdeSigmaHybrid}, and substituting the commutation rules (obtained from \eqref{TorsionlessHybrid}), as well as the limits for derivatives of $\Sigma$ at a single point, we get the following limits,
\begin{subequations}
\begin{align}
	[\Sigma_{a_3 a_2 a_1 a'_0}] &= -\frac{2}{3}R_{a_3 (a_2 a_1) a_0},\\
	[\Sigma_{a_4 a_3 a_2 a_1 a'_0}] &= -\frac{1}{6} \left(\nabla _{a_2}R_{a_4a_0a_3a_1}+\nabla _{a_3}R_{a_4a_0a_2a_1}+3 \nabla _{a_3}R_{a_4a_2a_1a_0}+\nabla _{a_4}R_{a_3a_0a_2a_1}+2 \nabla
	_{a_4}R_{a_3a_2a_1 a_0}\right)
\end{align}
\begin{align}
	[\Sigma_{a_5 a_4 a_3 a_2 a_1 a'_0}] = \frac{1}{90} (&2 R^{b_0}\,_{a_0a_5a_2} R_{b_0a_1a_4a_3}+11 R^{b_0}\,_{a_0a_4a_3} R_{b_0a_1a_5a_2}+11 R^{b_0}\,_{a_0a_4a_2} R_{b_0a_1a_5a_3} \nonumber\\
	&+11 R^{b_0}\,_{a_0a_3a_2} R_{b_0a_1a_5a_4}+11 R^{b_0}\,_{a_0a_5a_1} R_{b_0a_2a_4a_3}+11 R^{b_0}\,_{a_0a_4a_1}	R_{b_0a_2a_5a_3} \nonumber\\
	&+11 R^{b_0}\,_{a_0a_3a_1} R_{b_0a_2a_5a_4}-14 R^{b_0}\,_{a_2a_5a_4} R_{b_0a_3a_1a_0}-14 R^{b_0}\,_{a_1a_5a_4} R_{b_0a_3a_2a_0} \nonumber\\
	&+R^{b_0}\,_{a_0a_5a_4} \left(2 R_{b_0a_1a_3a_2}+R_{b_0a_3a_2a_1}\right)+11 R^{b_0}\,_{a_0a_2a_1} R_{b_0a_3a_5a_4} \nonumber\\
	&+22 R^{b_0}\,_{a_2a_1a_0} R_{b_0a_3a_5a_4}-14 R^{b_0}\,_{a_2a_5a_3} R_{b_0a_4a_1a_0}-14 R^{b_0}\,_{a_1a_5a_3} R_{b_0a_4a_2a_0} \nonumber\\
	&+R^{b_0}\,_{a_0a_5a_3} \left(2 R_{b_0a_1a_4a_2}+R_{b_0a_4a_2a_1}\right)-14 R^{b_0}\,_{a_1a_5a_2} R_{b_0a_4a_3a_0} \nonumber\\
	&+R^{b_0}\,_{a_0a_5a_2} R_{b_0a_4a_3a_1}+R^{b_0}\,_{a_0a_5a_1} R_{b_0a_4a_3a_2}-5 R^{b_0}\,_{a_2a_4a_3} R_{b_0a_5a_1a_0} \nonumber\\
	&-7 R^{b_0}\,_{a_4a_3a_2} R_{b_0a_5a_1a_0}-14 R^{b_0}\,_{a_1a_4a_3} R_{b_0a_5a_2a_0}-7 R^{b_0}\,_{a_4a_3a_1} R_{b_0a_5a_2a_0} \nonumber\\
	&+R^{b_0}\,_{a_0a_4a_3} R_{b_0a_5a_2a_1}-7 R^{b_0}\,_{a_4a_3a_0} R_{b_0a_5a_2a_1}-14 R^{b_0}\,_{a_1a_4a_2} R_{b_0a_5a_3a_0} \nonumber\\
	&-7 R^{b_0}\,_{a_4a_2a_1} R_{b_0a_5a_3a_0}+R^{b_0}\,_{a_0a_4a_2} R_{b_0a_5a_3a_1}-7 R^{b_0}\,_{a_4a_2a_0} R_{b_0a_5a_3a_1} \nonumber\\
	&+R^{b_0}\,_{a_0a_4a_1} R_{b_0a_5a_3a_2}-7 R^{b_0}\,_{a_4a_1a_0} R_{b_0a_5a_3a_2}-14 R^{b_0}\,_{a_1a_3a_2} R_{b_0a_5a_4a_0} \nonumber\\
	&-7 R^{b_0}\,_{a_3a_2a_1} R_{b_0a_5a_4a_0}+R^{b_0}\,_{a_0a_3a_2} R_{b_0a_5a_4a_1}-7 R^{b_0}\,_{a_3a_2a_0} R_{b_0a_5a_4a_1} \nonumber\\
	&+R^{b_0}\,_{a_0a_3a_1} R_{b_0a_5a_4a_2}-7 R^{b_0}\,_{a_3a_1a_0} R_{b_0a_5a_4a_2}+R^{b_0}\,_{a_0a_2a_1} R_{b_0a_5a_4a_3} \nonumber\\
	&-7 R^{b_0}\,_{a_2a_1a_0} R_{b_0a_5a_4a_3}-9 \nabla_{a_3}\nabla _{a_2}R_{a_5a_0a_4a_1}-9 \nabla _{a_4}\nabla _{a_2}R_{a_5a_0a_3a_1} \nonumber\\
	&-9 \nabla_{a_4}\nabla _{a_3}R_{a_5a_0a_2a_1}-36 \nabla _{a_4}\nabla _{a_3}R_{a_5a_2a_1a_0}-9 \nabla _{a_5}\nabla _{a_2}R_{a_4a_0a_3a_1} \nonumber\\
	&-9 \nabla_{a_5}\nabla _{a_3}R_{a_4a_0a_2a_1}-27 \nabla _{a_5}\nabla _{a_3}R_{a_4a_2a_1a_0}-9 \nabla _{a_5}\nabla _{a_4}R_{a_3a_0a_2a_1} \nonumber\\
	&-18 \nabla _{a_5}\nabla_{a_4}R_{a_3a_2a_1a_0} )
\end{align}
\end{subequations}
\section{3+1 formalism}\label{App:3p1}
We use the $3+1$ formulation introduced by Wald in \cite{WaldGR} and extended in \cite{Miramontes3p1}, which we briefly resume here.

Let $t:M\rightarrow \R$ a time function for the globally hyperbolic spacetime $(M,g_{ab})$, i.e., a function such that
	\begin{equation}
		g^{ab} \nabla_a t \nabla_b t < 0,
	\end{equation}
everywhere, and such that $\mathcal{C}_\tau=\{ p \in M : t(p)=\tau \}$ are Cauchy surfaces for $(M,g_{ab})$.

A foliation $\mathscr{F}$ of $(M,g_{ab})$ is defined by the set of Cauchy surfaces $\mathcal{C}_{t}$ for all $t\in \R$, so that the whole spacetime is covered by $\mathscr{F}$, identifying every $p\in M$ with a single surface $\mathcal{C}_{t(p)}$ that contains $p$. The vector field $n^a$ given by
	\begin{equation}
		n^a \equiv - \frac{g^{ab}\nabla_b t}{\sqrt{-g^{cd}(\nabla_c t)(\nabla_d t) }},
	\end{equation}
is normalized and normal to every $\mathcal{C}_t\in\mathscr{F}$. This can be directly verified given that in any tangent direction to $\mathcal{C}_t$ the derivative of $t$ vanishes. Let $t^a$ be a timelike vector field on $M$ such that
\begin{equation}
	t^a \nabla_a t =1. \label{CondTat1}
\end{equation}
Every point on each surface $\mathcal{C}_t$ can then be mapped to exactly one point at another surface $\mathcal{C}_{t'}$ by means of the one-parameter group of diffeomorphisms generated by $t^a$. This shows that all the surfaces of the foliation are diffeomorphic, and therefore $M$ is diffeomorphic to $\mathcal{C}_t\times \R$.
 
 Let $h_{ab}$ be a symmetric rank $(0,2)$ tensor on $M$ defined by
	\begin{equation}
		h_{ab} \equiv g_{ab} + n_a n_b. \label{InducedMetric}
	\end{equation}
It can be seen that $h_{ab}$ satisfies the properties of an euclidean $3-$metric on $\mathcal{C}_t$ for every $t$. The tensor $h_{ab}$ is called the \textit{induced metric} by $g_{ab}$ on $\mathcal{C}_t$. Note that for every $t$, $h_{ab}$ admits a representation as a $3-$dimensional symmetric matrix with entries $h_{ij}$ corresponding to the coordinate components $h_{\mu\nu}$, with $\mu>0$ in gaussian coordinates $x^\mu$ adapted to $\mathcal{C}_t$ such that $x^0=t$. These coordinates can be carried to every other surface $\mathcal{C}_t$ by the one-parameter group of diffeomorphisms generated by $t^a$. In the following we avoid such representation and we will not use any properties of the aforementioned coordinates.

The \textit{projectors} of vectors in the tangent and normal directions relative to $\mathcal{C}_t$ are respectively,
	\begin{align}
		h^a\,_{a'} &\equiv \delta^a\,_{a'} + n^a n_{a'},\\
		P_{\bot}^{a}\,_{a'} &\equiv -n^a n_{a'}.
	\end{align}
These objects allow to split the space of tangent vectors in $M$ in \textit{surface-tangent} and normal subspaces, which can be generalized directly to tensors of arbitrary rank. For example, the normal projection of the vector field $t^a$ 
	\begin{align}
		N\equiv -n_a t^a,
	\end{align}
is called \textit{lapse}, and the \textit{surface-}tangent projection,
\begin{align}
		\beta^a\equiv h^a\,_{a'} t^{a'}, \label{shift}
	\end{align}
is called \textit{shift vector}. Due to \eqref{CondTat1}, it follows that
	\begin{equation}
		N \equiv \frac{1}{n^a \nabla_a t} = \frac{1}{\sqrt{-g^{ab}(\nabla_a t)(\nabla_b t)}}.
	\end{equation}

The derivative operator associated with the intrinsic metric $h_{ab}$ will be denoted by $D_a$ and is given by \cite{WaldGR} (Lemma 10.2.1),
	\begin{equation}
		D_a T^{b_1 \dots b_k}\,_{c_1 \dots c_l} = h^{b_1}\,_{b'_1}\dots h^{b_1}\,_{b'_1} h_a\, ^{a'} h_{c_1}\,^{c'_1} \dots h_{c_l}\,^{c'_l}\nabla_{a'} T^{b'_1 \dots b'_k}\,_{c'_1 \dots c'_l}.
	\end{equation}
It is useful to define the \textit{surface-}tangent projection of the \textit{normal derivative} of a tensor by
	\begin{equation}
		\dbar T^{b_1 \dots b_k}\,_{c_1 \dots c_l} \equiv h^{b_1}\,_{b'_1}\dots h^{b_1}\,_{b'_1} h_{c_1}\,^{c'_1} \dots h_{c_l}\,^{c'_l}(-n^a \nabla_{a}) T^{b'_1 \dots b'_k}\,_{c'_1 \dots c'_l}.
		\label{NormalDerivative}
	\end{equation}
 We will show that the covariant derivative of any tensor field $T^{b_1 \dots b_k}\,_{c_1 \dots c_l}$ can be expressed in terms of \textit{surface-}tangent ($D_a T^{b_1 \dots b_k}\,_{c_1 \dots c_l}$), normal ($\dbar T^{b_1 \dots b_k}\,_{c_1 \dots c_l}$) and \textit{hybrid} projections of derivatives, along with surface-tangent tensors to be defined bellow. In the following we omit the prefix \textit{surface} regarding the tangent projections. 

For the normal vector itself we have
\begin{equation}
	\nabla_a n_b = K_{ab} + n_a u_b,
\end{equation}
where
	\begin{equation}
		K_{ab} \equiv D_a n_b
	\end{equation}
and 
	\begin{equation}
		u_a \equiv \dbar n_a.
	\end{equation}
The tensor $K_{ab}$ is tangent by construction and symmetric, and can be identified as the \textit{extrinsic curvature} of $\mathcal{C}_t$. In terms of the induced metric, it will be seen in the following subsection to be half the Lie derivative along $n^a$ of $h_{ab}$, i.e.,
\begin{equation}
	\Lie_n h_{ab} = 2 K_{ab}. \label{LieNH}
\end{equation}
It can also be seen that $u_a$ is related to $N$ by
	\begin{equation}
		u_a = - D_a \ln N. \label{DlnN}
	\end{equation}
From \eqref{DlnN} it follows that
	\begin{equation}
		D_a n_b = D_b n_a.
	\end{equation}

We say a tensor has been expanded in $3+1$ form when it is expressed in the form
	\begin{align}
		T^{a_1 \dots a_k}\,_{b_1\dots b_l} &= h^{a_1}\,_{a_1'} \dots h^{a_k}\,_{a_k'} h_{b_1}\,^{b_1'} \dots h_{b_l}\,^{b_l'} T^{a_1' \dots a_k'}\,_{b_1'\dots b_l'} \nonumber \\
		&\quad + h^{a_1}\,_{a_1'} \dots h^{a_{k-1}}\,_{a_{k-1}'} h_{b_1}\,^{b_1'} \dots h_{b_l}\,^{b_l'} n^{a_k} (-n_{a_k'}) T^{a_1' \dots a_k'}\,_{b_1'\dots b_l'} \nonumber\\
		&\quad + \dots \nonumber \\
		&\quad + n^{a_1} \dots n^{a_k} n_{b_1} \dots n_{b_l} (-n_{a_1'})\cdots (-n_{a_k'}) (-n^{b_1'}) \cdots (-n^{b_l'}) T^{a_1' \dots a_k'}\,_{b_1'\dots b_l'}. \label{31splitt}
	\end{align}

We now introduce a key aspect of our formulation of the $3+1$ formalism: the systematic labeling of the terms that constitute the complete $3+1$ expansion of tensors.

The $3+1$ expansion of a $(k,l)$ tensor $T^{a_1 \dots a_k}\,_{b_1\dots b_l}$ has $2^{k+l}$ terms. Each term contains products of the following types of elements: \textbf{a)} normal vectors $n^a$ and dual vectors $n_b$ adding up a total of $m$ \textit{free} normal \textit{elements} ie,
\begin{equation}
	n^{a_{q_{m}}}n^{a_{q_{m-1}}} \dots n_{b_{q_{2}}}n_{b_{q_{1}}},
\end{equation}
and \textbf{b)} a tensor with only tangent \textit{free} indexes, 
\begin{equation}
	\overset{\mathrm{m}}{T}^{a_{Q_{k+l-m}}a_{Q_{k+l-m-1}}\dots}\,_{\dots b_{Q_{2}} b_{Q_{1}}}.
\end{equation}
We will call these tensors \textit{tangent projection elements}. Therefore, every term is linearly independent of the other: there is only one complete $3+1$ expansion for every tensor in $M$ given a foliation.

We will number the indexes from right to left beginning at $0$, taking the original tensor as reference. We don't distinguish \textit{raised} or \textit{lowered} indexes in this labeling. For a given projection term of the $3+1$ expansion of a tensor we take all the numeric labels $q_{i}$ of every \textit{free} normal index so that summing all the terms of the form $2^{q_i}$ yield a unique label for each projection term. That is, if the normal dual vectors of some projection are $n_{a_{q_1}}, n_{a_{q_2}}\dots, n_{a_{q_m}}$, the numeric label we assing to the projection term is
\begin{equation}
\mathrm{m}\equiv \sum_{i=1}^{m} 2^{q_i}.
\end{equation}
This is a natural way to span all the possible projection terms based on an intuitive reasoning: the projection terms are labeled according to the binary positional representation of the indexes of their normal vectors (or dual vectors). For example, given a $(0,2)$ tensor $T_{ab} $ we have
	\begin{equation}
		T_{a_1 a_0} = \overset{0}{T}_{a_1 a_0} + n_{a_0} \overset{1}{T}_{a_1} + n_{a_1} \overset{2}{T}_{a_0} + n_{a_1} n_{a_0} \overset{3}{T},
	\end{equation}
where the tangent projection elements are
\begin{subequations} \label{ProjectionElements2}
	\begin{align}
		\overset{0}{T}_{a_1 a_0} &\equiv h_{a_1}\,^{a_1'} h_{a_0}\,^{a_0'} T_{a_1' a_0'}, & \sum_{q\in \emptyset} 2^q &= 0, \\
		\overset{1}{T}_{a_1} &\equiv h_{a_1}\,^{a_1'} (-n^{a_0'})T_{a_1' a_0'}, & \sum_{q\in \{0\}} 2^q &= 1,\\
		\overset{2}{T}_{a_0} &\equiv (-n^{a_1'}) h_{a_0}\,^{a_0'} T_{a_1' a_0'}, & \sum_{q\in \{1\}} 2^q &= 2,\\
		\overset{3}{T} &\equiv (-n^{a_1'}) (-n^{a_0'}) T_{a_1' a_0'}, & \sum_{q\in \{0,1\}} 2^q &= 3.
	\end{align}
\end{subequations}

Consider a scalar function $\phi: M\rightarrow \R$. The covariant derivative of $\phi$ expanded in $3+1$ form is
	\begin{equation}
		\nabla_{a_0} \phi = \phi_{a_0} = D_{a_0} \phi + n_{a_0} \dbar \phi. \label{DCovPhi}
	\end{equation}

If we consider a dual vector $v_a$, the $3+1$ expansion of its covariant derivative yield
\begin{align}
	\nabla_{a_1} v_{a_0} =&\, K_{a_1 a_0} \,^1 v + D_{a_1} \,^0 v_{a_0} + \left(K_{a_1}\,^b \,^0 v_b + D_{a_1} \,^1 v \right) n_{a_0} \nonumber\\
	 &+ \left(\dbar \,^0 v_{a_0} + u_{a_0} \,^1 v \right) n_{a_1} + \left(u^b \,^0 v_b \right) n_{a_1} n_{a_0}. \label{DcovV3p1}
\end{align}

By putting $v_a \equiv \nabla_a \phi$ in this expression we get the $3+1$ expansion for the second covariant derivative of $\phi$ (using the abbreviated notation for derivatives introduced in appendix \ref{AppSigma}),
\begin{align}
	\phi_{a_1 a_0} &= (K_{a_1 a_0} \dbar \phi + D_{a_1} D_{a_0} \phi) + n_{a_0} \left(D_{a_1} \dbar \phi + K_{a_1}\,^{b} D_b \phi \right) \nonumber \\
	& \quad + n_{a_1} \left(u_{a_0} \dbar \phi + \dbar D_{a_0} \phi \right) + n_{a_0} n_{a_1} \left(\dbar^2 \phi + u^b D_b \phi \right). \label{SecondDerPhi3p1}
\end{align}
Note that the tangent component of $\phi_{a_1 a_0}$ differs from the tangent derivative $D_{a_1} D_{a_0} \phi$ by a normal derivative of $\phi$ times the extrinsic curvature. This will keep happening at any arbitrary order $k$ of derivatives, which will include normal derivatives of $\phi$ up to order $k-1$.

The no-torsion condition ($\phi_{ab}=\phi_{ba}$) yields the identity
\begin{equation}
	\dbar D_a \phi - D_a \dbar \phi = K_a\,^b D_b \phi - u_a \dbar \phi. \label{NormalDerivativesExchange}
\end{equation}

Recall that for any dual vector $\omega_a$, the Riemann tensor is defined by
\begin{equation}
	R_{abc}\,^d \omega_d = \nabla_a \nabla_b \omega_c - \nabla_b \nabla_a \omega_c. \label{RiemannDef}
\end{equation}
Substituting the $3+1$ expansion of the second derivative of a dual vector $v_a$ in \eqref{RiemannDef}, and taking into account all its symmetries, we get that the nonvanishing tensor projection elements of the Riemann tensor are
\begin{subequations}
	\begin{align}
		\,^0 R_{a_3 a_2 a_1 a_0} &= \,^{(3)}R_{a_3 a_2 a_1 a_0} + K_{a_3 a_1} K_{a_2 a_0} - K_{a_2 a_1}K_{a_3 a_0}, \label{0Riem}\\
		\,^1 R_{a_3 a_2 a_1} = \,^2 R_{a_2 a_3 a_1} = \,^4 R_{a_1 a_2 a_2} = \,^8 R_{a_1 a_2 a_3} &= D_{a_2} K_{a_3 a_1} - D_{a_3} K_{a_2 a_1},\\
		\,^5 R_{a_3 a_1} = - \,^6 R_{a_3 a_1} = - \,^9 R_{a_3 a_1} = \,^{10} R_{a_3 a_1} &= u_{a_3} u_{a_1} + \dbar K_{a_3 a_1} - K_{a_3}\,^b K_{b a_1} -D_{a_3} u_{a_1}.
	\end{align}\label{RiemannComponents}
\end{subequations}
where $\,^{(3)}R_{a_3 a_2 a_1 a_0}$ is the Riemann tensor defined by the induced metric $h_{ab}$ on $\mathcal{C}_t$. We will call it the \textit{intrinsic} Riemann tensor.

Implementing \eqref{RiemannDef} for any \textit{tangent dual vector} $\omega_b$, we get
\begin{subequations}
	\begin{equation}
		\dbar D_a \omega_b - D_a \dbar \omega_b = K_a\,^{c} D_c \omega_b - u_a \dbar \omega_b + \left(K_{ab} u^c - K_{a}\,^{c} u_b + D_b K_{a}\,^{c} - D^c K_{ab}\right) \omega_c, \label{ConmutaDersNormTan1}
	\end{equation}
	and for any \textit{tangent tensor} $A_{a_1 a_0}$,
	\begin{align}
		\dbar D_{a_2} A_{a_1 a_0} - D_{a_2} \dbar A_{a_1 a_0} =& K_{a_2}\,^b D_b A_{a_1 a_0} - u_{a_2} \dbar A_{a_1 a_0} \nonumber\\
		&+ \left(K_{a_2 a_1} u^b - K_{a_2}\,^b u_{a_1} + D_{a_1}K_{a_2}\,^b-D^b K_{a_2 a_1}\right) A_{b a_0}\nonumber\\
		&+ \left(K_{a_2 a_0} u^b - K_{a_2}\,^b u_{a_0} + D_{a_0}K_{a_2}\,^b-D^b K_{a_2 a_0}\right) A_{a_1 b}, \label{ConmutaDersNormTan2}
	\end{align} 
\end{subequations}
and so on for higher order tensors. 

From \eqref{RiemannComponents} it is possible to compute the projection components for the Ricci tensor $\mathsf{R}_{ab}$ in terms of their intrinsic counterparts, the extrinsic curvature, the vector $u^a$ and their derivatives:
\begin{subequations}
	\begin{align}
		\,^0 \mathsf{R}_{a_1 a_0} &= \,^{(3)}\mathsf{R}_{a_1 a_0} + K\, K_{a_1 a_0}  + D_{a_1} u_{a_0} - u_{a_1} u_{a_0} - \dbar K_{a_1 a_0}, \\
		\,^1 \mathsf{R}_{a_1} = \,^2 \mathsf{R}_{a_1} &= D_{a_1} K - D_b K^b\,_{a_1},\label{GaussCodazzi2}\\
		\,^3 \mathsf{R} &= \dbar K + u^2 - K^{b_1 b_0} K_{b_1 b_0} - D_b u^b.
	\end{align}\label{Ricc3p1}
\end{subequations}
as well as the Ricci scalar $\mathscr{R}$,
\begin{equation}
	\mathscr{R} = \,^{(3)}\mathscr{R} + K^2 + K^{b_1 b_0} K_{b_1 b_0} + 2 \left( D_b u^b - u^2 - \dbar K \right), \label{RiccScalar}
\end{equation}
where
\begin{subequations}\begin{align}
		K &\equiv K^b\,_b,\\
		u^2 &\equiv u^b u_b.
\end{align}\end{subequations}
Equations \eqref{0Riem} and \eqref{GaussCodazzi2} are known as the Gauss-Codazzi relations.

This formalism allows to systematically express derivatives of arbitrary tensors. General formulas for covariant derivatives, products, contractions, symmetrization and anti-symmetrization of indexes can be given and programmed within a computer algebra system. However, those expressions require further definitions that are out of the scope of this exposition. An exposition of these formulas is available in \cite{Miramontes3p1}.

This $3+1$ formalism is a complete representation of $4-$dimensional tensors in terms of tangent tensor components on each of the surfaces $\mathcal{C}_t$ of a given foliation of spacetime.

In this formalism, Bianchi identity allows to solve $\dbar \,^{(3)}R_{a_3 a_2 a_1 a_0}$ in terms of $K_{a_1 a_0}$, $u_{a_0}$ and their tangent derivatives:
\begin{subequations} \label{BianchiDbarR}
	\begin{align}
		\dbar \,^{(3)}R_{a_3 a_2 a_1 a_0} =& 2 \Big( K_{b_0 [a_3} \,^{(3)}R^{b_0}\,_{a_2] a_1 a_0} + D_{a_3} D_{[a_1} K_{a_0] a_2} - D_{a_2} D_{[a_1} K_{a_0] a_3} - K_{a_1 [a_3} D_{a_2]} u_{a_0} \nonumber\\
		&\quad + K_{a_0 [a_3} D_{a_2]} u_{a_1}+ 2\big( (D_{[a_3} K_{a_2] [a_1}) u_{a_0]} + (D_{[a_1} K_{a_0] [a_3}) u_{a_2]} - u_{[a_3} K_{a_2] [a_1} u_{a_0]} \big)\Big). \label{dbarRiemann3}
	\end{align}
	Upon contraction, this also yield relations for the intrinsic Ricci tensor and scalar,
	\begin{align}
		\dbar \,^{(3)}\mathsf{R}_{a_1 a_0} =& 2 \big( K_{b_0 (a_1} \,^{(3)}\mathsf{R}^{b_0}\,_{a_0)} - D_{b_0} D_{(a_1} K_{a_0)}\,^{b_0} \big) + D_{a_1} D_{a_0}K + D^b D_b K_{a_1 a_0} \nonumber\\
		&\quad + 2 \Big( u^{b_0} \big(D_{(a_1} K_{a_0) b_0} - D_{b_0} K_{a_1 a_0} - K_{b_0 (a_1} u_{a_0)}\big) + u_{(a_1} \big(D^{b_0} K_{a_0)b_0} - D_{a_0)}K\big) \Big)   \nonumber\\
		&\quad - K D_{a_1} u_{a_0} +2 (D_{b_0} u_{(a_1}) K_{a_0)}\,^{b_0} - K_{a_1 a_0} D_{b_0} u^{b_0} + u^2 K_{a_1 a_0} + K u_{a_0} u_{a_1}. \label{dbarRicciTensor3}
	\end{align}
	\begin{align}
		\dbar \,^{(3)}\mathscr{R} = 2 \Big(& K^{b_1 b_0} \,^{(3)}\mathsf{R}_{b_1 b_0} + D^{b_0}\big(D_{b_0} K - D_{b_1} K^{b_1}\,_{b_0}\big) + 2 u^{b_0} \big(D_{b_1} K^{b_1}\,_{b_0} - D_{b_0}K\big) \nonumber\\
		&- K D_{b_0} u^{b_0} + K^{b_1 b_0} D_{b_1} u_{b_0} + K u^2 - K^{b_1 b_0} u_{b_1} u_{b_0} \Big) . \label{dbarRicciScalar3}
	\end{align}
\end{subequations}

As an example, we can compute the \textit{failure} of a geodesic on $\mathcal{C}_t$, defined by the tangent derivative operator $D_a$, to be a geodesic vector on $M$. From \eqref{DcovV3p1} we have that for a \textit{surface}-tangent vector $\xi^a$ which is tangent to an affine geodesic in $\mathcal{C}_t$ such that
\begin{equation}
	\xi^b D_b \xi^a =0,
\end{equation}
then 
\begin{equation}
	\xi^b \nabla_b \xi^a = n^a K^{b_1 b_0} \xi_{b_1} \xi_{b_0}. \label{3p1GeodesicId}
\end{equation}
This justifies the statement made in Section \ref{sec:HypSigma} that in general, geodesics intrinsically defined on $\mathcal{C}$, for points contained within a convex normal neighborhood $\mathscr{D}\subset\mathcal{C}$ will not coincide generically with the \textit{spacetime} geodesics for the same points. This also proves that they will coincide if the extrinsic curvature vanishes in $\mathscr{D}$.

\subsection{Lie derivatives and \textit{time} derivatives.} \label{App:3p1LieTime}
The Lie derivative $\Lie_v T^{b_{l-1}\dots b_0}\,_{a_{k-1}\dots a_0}$ represent the rate of change of the tensor $T^{b_{l-1}\dots b_0}\,_{a_{k-1}\dots a_0}$ along the uniparametric group of diffeomorphisms generated by the vector field $v^a$. In terms of the covariant derivative of the spacetime it can be seen \cite{WaldGR} that
\begin{align}
	\Lie_v T^{b_{l-1}\dots b_0}\,_{a_{k-1}\dots a_0} = v^c\nabla_c T^{b_{l-1}\dots b_0}\,_{a_{k-1}\dots a_0} -& \sum_{i=0}^{l-1} T^{b_{l-1}\dots c \dots b_0}\,_{a_{k-1}\dots a_0} \nabla_c v^{b_i} \nonumber\\
	&+ \sum_{j=0}^{k-1} T^{b_{l-1}\dots b_0}\,_{a_{k-1}\dots c \dots a_0} \nabla_{a_j} v^c.\label{LieCov}
\end{align}
For a scalar field $\phi$ this reduces to
\begin{equation}
	\Lie_v \phi = v^b \nabla_b \phi,
\end{equation}
which in $3+1$ form is given by the contraction of \eqref{DCovPhi} with $v$ expressed in $3+1$ form, that is,
\begin{equation}
	\Lie_v \phi =\,^0 v^{b_0} D_{b_0} \phi - \,^1 v \dbar \phi.
\end{equation}
Similarly, for a dual vector $\omega_a$ we have
\begin{equation}
	\Lie_v \omega_{a_0} = v^b \nabla_b \omega_{a_0} + \omega_{b} \nabla_{a_0} v^b,
\end{equation}
which in $3+1$ form reads
\begin{align}
	\Lie_v \omega_{a_0} = &- \overset{1}{v} \dbar \overset{0}{\omega}_{a_0} + \overset{1}{v} K_{a_0}\,^b \overset{0}{\omega}_b - \overset{1}{v}\overset{1}{\omega} u_{a_0} + \overset{0}{\omega}_{b} D_{a_0} \overset{0}{v}\,^b - \overset{1}{\omega}D_{a_0} \overset{1}{v} + \overset{0}{v}\,^b D_b \overset{0}{\omega}_{a_0} \nonumber\\
	&+n_{a_0}\left(- \overset{1}{v} \dbar \overset{1}{\omega} + \overset{0}{\omega}_b \dbar \overset{0}{v}\,^b + K^{b_1 b_0} \overset{0}{v}_{b_1} \overset{0}{\omega}_{b_0} - \overset{1}{\omega} \dbar \overset{1}{v}-\overset{1}{\omega}u_b \overset{0}{v}\,^b+ \overset{0}{v}\,^b D_b \overset{1}{\omega}\right). \label{Lie1Forma}
\end{align}
We use the vector $t^a$ defined in \eqref{CondTat1} in order to represent time derivatives by taking $t^a$ to be the generator of \textit{time flow} in this formalism. The $3+1$ expansion of $t^a$ is given in terms of the lapse function $N$ and the shift tangent vector $\beta^a$,
\begin{equation}
	t^{a_0} = \beta^{a_0} + n^{a_0} N.
\end{equation}
The time derivative of a tensor $T^{b_{l-1}\dots b_0}\,_{a_{k-1}\dots a_0}$ is then taken as the tangent projection of its Lie derivative along $t^a$, that is,
\begin{equation}
	\dot T^{b_{k-1}\dots b_k}\,_{a_{k-1}\dots a_k} \equiv h^{b_{l-1}}\,_{d_{l-1}} \dots h^{b_0}\,_{d_0} h_{a_{k-1}}\,^{c_{k-1}} \dots h_{a_0}\,^{c_0} \Lie_t T^{d_{k-1}\dots d_k}\,_{c_{k-1}\dots c_k}.
\end{equation}
Expanding this definition with aid of formula \eqref{LieCov}, we get
\begin{align}
	\dot T^{b_{l-1}\dots b_0}\,_{a_{k-1}\dots a_k} =& - N \dbar T^{b_{l-1}\dots b_0}\,_{a_{k-1}\dots a_k} + \beta^c D_c T^{b_{l-1}\dots b_0}\,_{a_{k-1}\dots a_k}\nonumber\\
	&- \sum_{i=0}^{l-1} T^{b_{l-1}\dots c \dots b_0}\,_{a_{k-1}\dots a_0} \left(N K_c\,^{b_i} + D_c \beta^{b_i}\right) \nonumber\\
	&+ \sum_{j=0}^{k-1} T^{b_{l-1}\dots b_0}\,_{a_{k-1}\dots c \dots a_0} \left(N K_{a_j}\,^c + D_{a_j} \beta^{c}\right).\label{LieTime}
\end{align}
For example, we can compute the time derivative of the induced metric,
\begin{equation}
	\dot h_{a_1 a_0} = 2 \left(N K_{a_1 a_0} + D_{(a_1} \beta_{a_0)}\right). \label{DotH}
\end{equation}
Consider the case where $\beta^a=0$ and $N=1$, then $n^a = t^a$ and \eqref{DotH} reduces to \eqref{LieNH}. This lies behind the intuitive notion that the extrinsic curvature is essentially (half) the \textit{time} derivative of the induced metric, even if such relationship is formally given by \eqref{DotH}. We note that in this expression, the operator $\dbar$ does not appear given that $\dbar h_{ab}=0$, however, for tensors other than $h_{ab}$, the order of time derivatives of tensors will coincide with the order of $\dbar$ derivatives. 

We can make the substitution 
\begin{align}
	\dbar T^{b_{l-1}\dots b_0}\,_{a_{k-1}\dots a_k} =& - \frac{1}{N}\dot T^{b_{l-1}\dots b_0}\,_{a_{k-1}\dots a_k} + \frac{1}{N}\beta^c D_c T^{b_{l-1}\dots b_0}\,_{a_{k-1}\dots a_k}\nonumber\\
	&- \sum_{i=0}^{l-1} T^{b_{l-1}\dots c \dots b_0}\,_{a_{k-1}\dots a_0} \left(K_c\,^{b_i} + \frac{1}{N} D_c \beta^{b_i}\right) \nonumber\\
	&+ \sum_{j=0}^{k-1} T^{b_{l-1}\dots b_0}\,_{a_{k-1}\dots c \dots a_0} \left(K_{a_j}\,^c + \frac{1}{N} D_{a_j} \beta^{c}\right),\label{DbarLie}
\end{align}
and for the case of $h_{ab}$, we can write first
\begin{equation}
	K_{a_1 a_0} = \frac{1}{N} \left(\frac{1}{2}\dot h_{a_1 a_0} - D_{(a_1} \beta_{a_0)}\right),
\end{equation}
and then use \eqref{DbarLie} to translate normal derivatives of $K_{ab}$ in terms of time derivatives of $h_{ab}$. That results in an expression of the form
\begin{equation}
	\dbar^n K_{a_1 a_0} = \frac{1}{2N (-N)^n} (\partial_t)^{n+1} h_{a_1 a_0} + \dots \label{DtHDk}
\end{equation}
where $(\partial_t)^n h_{a_1 a_0}$ represent the $n$-th time derivative of $h_{a_1 a_0}$. Therefore, the number of normal derivatives of the extrinsic curvature equals the number of time derivatives of the induced metric, plus one.

\subsection{3+1 expansion of \texorpdfstring{$[\Sigma_{a_k\dots a_0}]$}{Sak...a0}.}\label{App:3p1SigmaLimits}

We are interested in the explicit expression for the coincidence limits of every tangent, normal and hybrid derivatives of $\Sigma$ in terms of the $3+1$ expansion of data on any surface $\mathcal{C}_t$.

For the first derivative, we have from \eqref{LimDSigma}, 
\begin{equation}
	[\Sigma_{a_0}]= 0, \label{LimDSigmaShort}
\end{equation}
and given \eqref{DCovPhi}, we readily verify that
\begin{subequations}
	\begin{align}
		[D_a \Sigma] &= 0, \label{LimDtanSigma}\\
		[\dbar \Sigma]&= 0. \label{LimDNSigma}
	\end{align}
\end{subequations}

For the second order derivatives we have from \eqref{SecondDerPhi3p1}, \eqref{LimDtanSigma}, \eqref{LimDNSigma}, \eqref{LimDDSigma} and \eqref{InducedMetric},
\begin{align}
	[\Sigma_{a_1 a_0}] = h_{a_1 a_0} + n_{a_0} n_{a_1} (-1),
\end{align}
	which implies 
\begin{subequations}
\begin{align}
		[D_a D_b \Sigma] &= h_{ab},\\
		[D_{a} \dbar \Sigma] &= [\dbar D_{a} \Sigma] = 0, \label{DdbarSigma}\\
		[\dbar^2 \Sigma] &= -1. \label{ddbarSigma}
\end{align}
\end{subequations}
	
In general, the procedure for computing a derivative order $k > 2$ is the following: 
\begin{enumerate}
	\item Compute the $3+1$ expansion for the covariant derivative of $\Sigma_{a_{k-1}\dots a_0}$.
	\item Take the limit of said expansion and replace the previously found limits of all the derivatives of order $k-1$, i.e., $[D_{a_{k-1}} \dots D_{a_{0}} \Sigma]$, $[D_{a_{k-1}} \dots D_{a_{1}} \dbar \Sigma]$, $[D_{a_{k-1}} \dots D_{a_{2}} \dbar D_{a_0}\Sigma]$, $\dots$, $[\dbar^{k-1} \Sigma]$.
	\item Identify every projection element of the expansion of the derivative with the corresponding projection element of the $3+1$ expansion of the limit $[\Sigma_{a_k\dots a_0}]$ computed in Appendix \ref{AppSigma}.
	\item Solve for the limits $[D_{a_{k}} \dots D_{a_{0}} \Sigma]$, $[D_{a_{k}} \dots D_{a_{1}} \dbar \Sigma]$, $\dots$, $[\dbar^{k} \Sigma]$ in the equation for the tangent projection component $0, 1, \dots, 2^k-1$, respectively.
\end{enumerate}
We will denote this algorithm as $\star$.

\begin{subequations}
For $k=3$, the relevant coincidence limit is trivially $[\Sigma_{abc}]=0$, and following $\star$ we get
	\begin{align}
		[D_{a_2} D_{a_1} D_{a_0} \Sigma] &=0, \label{LimD3Sigma}\\
		[D_{a_2} D_{a_1} \dbar \Sigma] &= -K_{a_2 a_1}, \label{LimD2dbarSigma}\\
		[D_{a_2} \dbar D_{a_0}] &= 0,\\
		[D_{a_2} \dbar^2 \Sigma] &= - u_{a_2}, \label{LimDddbarSigma}\\
		[\dbar D_{a_1} D_{a_0} \Sigma] &= K_{a_1 a_0},\\
		[\dbar D_{a_1} \dbar \Sigma] &= 0, \\
		[\dbar^2 D_{a_0} \Sigma] &= u_{a_0},\\
		[\dbar^3\Sigma] &= 0.
	\end{align}
\end{subequations}
For $k=4$ we have a nontrivial coincidence limit, a combination of Riemann tensors with various index orderings, cf. Eq. \eqref{limSigma4}. The Riemann tensor has already been completely expressed in $3+1$ form in the previous section and in \cite{Miramontes3p1}, which also describes the techniques required for dealing with the exchange of indexes. Following the algorithm $\star$ we get
	\begin{subequations}
		\begin{align}
			[D_{a_3} D_{a_2} D_{a_1} D_{a_0} \Sigma] &= \frac{1}{3} \left(\,^{(3)} R_{a_3 a_1 a_0 a_2} + \,^{(3)}R_{a_3 a_0 a_1 a_2} + K_{a_3 a_0} K_{a_2 a_1} + K_{a_3 a_1} K_{a_2 a_0} + K_{a_3 a_2} K_{a_1 a_0} \right) \label{LimD4Sigma} \\
			[D_{a_3} D_{a_2} D_{a_1} \dbar \Sigma] &= -\frac{2}{3} \left(D_{a_1} K_{a_3 a_2} + D_{a_2} K_{a_3 a_1} + D_{a_3} K_{a_2 a_1}\right) \label{LimD3dbarSigma}\\
			[D_{a_3} D_{a_2} \dbar D_{a_0} \Sigma] &= \frac{1}{3} \left(3 K_{a_3 a_2} u_{a_0} - 2 D_{a_0} K_{a_3 a_2} + D_{a_2} K_{a_3 a_0} + D_{a_3} K_{a_2 a_0} \right)\\
			[D_{a_3} D_{a_2} \dbar^2 \Sigma] &= \frac{1}{3}\left(2 K_{a_3}\,^{b} K_{b a_2} - 2 u_{a_3} u_{a_2} - 2 \dbar K_{a_3 a_2} - 4 D_{a_3} u_{a_2}\right) \label{LimDDddbarSigma}
			\end{align}
			\begin{align}
			[D_{a_3} \dbar D_{a_1} D_{a_0} \Sigma] &= \frac{1}{3} \left(3 u_{a_3} K_{a_1 a_0} + D_{a_0} K_{a_3 a_1} + D_{a_1} K_{a_3 a_0} + D_{a_3} K_{a_1 a_0}\right)\\
			[D_{a_3} \dbar D_{a_1} \dbar \Sigma] &= \frac{1}{3} \left(u_{a_3} u_{a_1} -K_{a_3}\,^b K_{b a_1} -2\dbar K_{a_3 a_1} - D_{a_3} u_{a_1} \right)\\
			[D_{a_3} \dbar^2 D_{a_0} \Sigma] &= \frac{1}{3} \left(4 u_{a_3} u_{a_0} -K_{a_3}\,^b K_{b a_0} + \dbar K_{a_3 a_0} + 2 D_{a_3} u_{a_0} \right)\\
			[D_{a_3} \dbar^3 \Sigma] &= -\dbar u_{a_3}
		\end{align}
		\begin{align}
			[\dbar D_{a_2} D_{a_1} D_{a_0} \Sigma] &= \frac{1}{3} \left(D_{a_2} K_{a_1 a_0} + D_{a_1} K_{a_2 a_0} + D_{a_0} K_{a_2 a_1}\right)\\
			[\dbar D_{a_2} D_{a_1} \dbar \Sigma] &= \frac{1}{3}\left(u_{a_2} u_{a_1} - 4 K_{a_2}\,^{b} K_{b a_1} -2 \dbar K_{a_2 a_1} - D_{a_1} u_{a_2} \right)\\
			[\dbar D_{a_2} \dbar D_{a_0}\Sigma]&= \frac{1}{3} \left(u_{a_2} u_{a_0} - K_{a_2}\,^b K_{b a_0} + 2 D_{a_2} u_{a_0} \right) \\
			[\dbar D_{a_2} \dbar^2 \Sigma] &= -K_{a_2}\,^b u_b - \dbar u_{a_2}
			\end{align}
			\begin{align}
			[\dbar^2 D_{a_1} D_{a_0} \Sigma] &= \frac{1}{3}\left(2 K_{a_1}\,^b K_{b a_0} - 2 u_{a_1} u_{a_0} + 4 \dbar K_{a_1 a_0} + 2 D_{a_1} u_{a_0} \right)\\
			[\dbar^2 D_{a_1} \dbar \Sigma] &= -K_{a_1}\,^b u_b\\
			[\dbar^3 D_{a_0} \Sigma] &= 2 \dbar u_{a_0}\\
			[\dbar^4 \Sigma] &= -u^b u_b
		\end{align}
	\end{subequations}
We recall that $\,^{(3)}R_{a_3 a_2 a_1}\,^{a_0}$ is the Riemann tensor defined by the induced metric $h_{ab}$ on the surface.

The order $k=5$ involves the coincidence limit of a combination of derivatives of the Riemann tensor, Eq. \eqref{LimD5Sigma}, where we obtain $2^5=32$ limits of derivatives. It serves no practical purpose to write all these limits down here, so we only reproduce the limit of derivatives that will be used later,
\begin{align}
	[D_{a_4} D_{a_3} D_{a_2} D_{a_1} D_{a_0} \Sigma] &= \frac{1}{2} \left(D_{a_4} \,^{(3)}R_{a_1 a_2 a_3 a_0} +D_{a_3} \,^{(3)}R_{a_1 a_4 a_2 a_0} + D_{a_2} \,^{(3)}R_{a_1 a_3 a_4 a_0} \right) \nonumber\\
	&\quad + 5 K_{(a_4 a_3 } D_{a_2 } K_{a_1 a_0)}, \label{Lim5}\\
	[D_{a_4} D_{a_3} D_{a_2} D_{a_1} \dbar \Sigma] &= K_{(a_4 a_3} \dbar K_{a_2 a_1)} - 3 D_{(a_4} D_{a_3} K_{a_2 a_1)} - K_{(a_4 a_3} D_{a_2} u_{a_1)}\nonumber\\
	&\quad - 2 K_{b(a_4} K^{b}\,_{a_3} K_{a_2 a_1)} + K_{(a_4 a_3} u_{a_2} u_{a_1)} \nonumber\\
	&\quad +\frac{1}{6}\Big(- K_{a_4}\,^{b} [ 2 \,^{(3)}R_{a_3 b a_2 a_1} + \,^{(3)}R_{a_3 a_2 a_1 b} ] \nonumber\\
	&\qquad + K_{a_3}\,^{b} [ 5 \,^{(3)}R_{a_4 b a_2 a_1} + 4 \,^{(3)}R_{a_4 a_2 a_1 b} ] \nonumber\\
	&\qquad + K_{a_2}\,^{b} [ \,^{(3)}R_{a_4 b a_3 a_1} - 4 \,^{(3)}R_{a_4 a_3 a_1 b} ] \nonumber\\
	&\qquad - 2 K_{a_1}\,^{b} [ \,^{(3)}R_{a_4 b a_3 a_2} + 2 \,^{(3)}R_{a_4 a_3 a_2 b} ] \Big),\label{D4dbarSigma}
\end{align}
\begin{align}
	[D_{a_4} D_{a_3} D_{a_2} \dbar^2 \Sigma] &= -\frac{3}{2} D_{(a_4} \dbar K_{a_3 a_2)}-3 u_{(a_4} D_{a_3} u_{a_2)}-\frac{3}{2}D_{a_4} D_{a_3} u_{a_2} \nonumber\\
	&\quad + \frac{1}{6}\left(\,^{(3)}R_{a_3 a_2 a_4}\,^b + 2 \,^{(3)}R_{a_4 a_3 a_2}\,^b \right) u_b + K_{b(a_4} D^b K_{a_3 a_2)}\nonumber\\
	&\quad + 2 K_{(a_4}\,^{b} D_{a_3} K_{a_2) b} - K_{(a_4 a_3} K_{a_2)}\,^b u_b. \label{LimD3dbar2Sigma}
\end{align}

The coincidence limit for the order $6$ covariant derivative of $\Sigma$ is given by \eqref{LimD6Sigma}. The all-tangent derivative limit is given by:
\begin{align}
	[D_{a_5} D_{a_4} D_{a_3} D_{a_2} D_{a_1} D_{a_0} \Sigma] &= \frac{1}{15} \Big(11 K_{a_2 a_3} K_{a_4 a_5} \dbar K_{a_0 a_1} -K_{a_3 a_4} K_{a_1 a_5} \dbar K_{a_0 a_2} -K_{a_1 a_4} K_{a_3 a_5} \dbar K_{a_0 a_2} \nonumber \\
	&\, -K_{a_1 a_3} K_{a_4 a_5} \dbar K_{a_0 a_2} -K_{a_3 a_4} K_{a_0 a_5} \dbar K_{a_1 a_2} -K_{a_0 a_4} K_{a_3 a_5} \dbar K_{a_1 a_2} \nonumber \\
	&\,-K_{a_0 a_3}	K_{a_4 a_5} \dbar K_{a_1 a_2} -K_{a_1 a_2} K_{a_4 a_5} \dbar K_{a_0 a_3} -K_{a_0 a_2} K_{a_4 a_5} \dbar K_{a_1 a_3} \nonumber \\
	&\,-K_{a_1 a_4} K_{a_0 a_5} \dbar K_{a_2 a_3} -K_{a_0 a_4} K_{a_1 a_5} \dbar K_{a_2 a_3} -K_{a_0 a_1} K_{a_4 a_5} \dbar K_{a_2 a_3} \nonumber \\
	&\,-7 K_{a_2 a_3} K_{a_1 a_5} \dbar K_{a_0 a_4} -7 K_{a_1 a_2} K_{a_3 a_5} \dbar K_{a_0 a_4} -7 K_{a_2 a_3} K_{a_0 a_5} \dbar K_{a_1 a_4} \nonumber \\
	&\,-7 K_{a_0 a_2} K_{a_3 a_5} \dbar K_{a_1 a_4} -K_{a_1 a_3} K_{a_0 a_5} \dbar K_{a_2 a_4} -K_{a_0 a_3} K_{a_1 a_5} \dbar K_{a_2 a_4} \nonumber \\
	&\,-K_{a_0 a_1} K_{a_3 a_5} \dbar K_{a_2 a_4} -K_{a_1 a_2} K_{a_0 a_5} \dbar K_{a_3 a_4} -K_{a_0 a_2} K_{a_1 a_5} \dbar K_{a_3 a_4} \nonumber \\
	&\,-7 K_{a_2 a_3} K_{a_1 a_4} \dbar K_{a_0 a_5} -7 K_{a_1 a_2} K_{a_3 a_4} \dbar K_{a_0 a_5} -7 K_{a_2 a_3} K_{a_0 a_4} \dbar K_{a_1 a_5} \nonumber \\
	&\,-7 K_{a_0 a_2} K_{a_3 a_4} \dbar K_{a_1 a_5} -K_{a_1 a_3} K_{a_0 a_4} \dbar K_{a_2 a_5} -K_{a_0 a_3} K_{a_1 a_4} \dbar K_{a_2 a_5} \nonumber \\
	&\,-K_{a_0 a_1} K_{a_3 a_4} \dbar K_{a_2 a_5} -K_{a_1 a_2} K_{a_0 a_4} \dbar K_{a_3 a_5} -K_{a_0 a_2} K_{a_1 a_4} \dbar K_{a_3 a_5} \nonumber \\
	&\,+11 K_{a_1 a_2} K_{a_0 a_3} \dbar K_{a_4 a_5} +11 K_{a_0 a_2} K_{a_1 a_3} \dbar K_{a_4 a_5} \nonumber \\
	&\,+K_{a_2 a_5} \big[11 K_{a_3 a_4} \dbar K_{a_0 a_1} -K_{a_1 a_4} \dbar K_{a_0 a_3} -K_{a_0 a_4} \dbar K_{a_1 a_3} \nonumber \\
	&\qquad-7 K_{a_1 a_3} \dbar K_{a_0 a_4} -7 K_{a_0 a_3} \dbar K_{a_1 a_4} -K_{a_0 a_1} \dbar K_{a_3 a_4} \big]\nonumber \\
	&\,+K_{a_2 a_4} \big[11 K_{a_3 a_5} \dbar K_{a_0 a_1} -K_{a_1 a_5} \dbar K_{a_0 a_3} -K_{a_0 a_5} \dbar K_{a_1 a_3} \nonumber \\
	&\qquad-7 K_{a_1 a_3} \dbar K_{a_0 a_5} -7 K_{a_0 a_3} \dbar K_{a_1 a_5} -K_{a_0 a_1} \dbar K_{a_3 a_5} \big]\nonumber \\
	&\,+11 K_{a_0 a_1} K_{a_2 a_3} \dbar K_{a_4 a_5} \Big)+ \,^0 F_{a_5 a_4 a_3 a_2 a_1 a_0} \label{MonsterLim6}
\end{align}
where $\,^0 F_{a_5 a_4 a_3 a_2 a_1 a_0}$ is a tangent tensor which does not involve normal derivatives of the extrinsic curvature.

At least two other projections are of explicit interest,

\begin{align}
	[D_{a_5} D_{a_4} D_{a_3} D_{a_2} D_{a_1} \dbar \Sigma] &= \frac{1}{5} (\dbar K_{a_5 a_1}) D_{a_2} K_{a_4 a_3} - \frac{16}{45} (\dbar K_{a_4 a_1}) D_{a_2} K_{a_5 a_3} + (\dbar K_{a_4 a_2}) D_{a_3} K_{a_5 a_1} \nonumber \\
	&\,+ \frac{19}{15} (\dbar K_{a_5 a_4}) D_{a_3} K_{a_2 a_1} - \frac{2}{5} (\dbar K_{a_5 a_2}) D_{a_3} K_{a_4 a_1} - \frac{2}{9} (\dbar K_{a_5 a_1}) D_{a_3} K_{a_4 a_2}\nonumber \\
	&\,+ \frac{4}{9} (\dbar K_{a_4 a_1}) D_{a_3} K_{a_5 a_2} - \frac{41}{45} (\dbar K_{a_2 a_1}) D_{a_3} K_{a_5 a_4} + \frac{2}{3} (\dbar K_{a_5 a_3}) D_{a_4} K_{a_2 a_1} \nonumber \\
	&\,+ \frac{2}{3} (\dbar K_{a_5 a_2}) D_{a_4} K_{a_3 a_1} - \frac{14}{45} (\dbar K_{a_5 a_1}) D_{a_4} K_{a_3 a_2} + \frac{4}{45} (\dbar K_{a_3 a_1}) D_{a_4} K_{a_5 a_2} \nonumber \\
	&\,+ \frac{22}{45} (\dbar K_{a_2 a_1}) D_{a_4} K_{a_5 a_3} + \frac{2}{3} (\dbar K_{a_4 a_3}) D_{a_5} K_{a_2 a_1} + \frac{5}{3} (\dbar K_{a_4 a_2}) D_{a_5} K_{a_3 a_1} \nonumber \\
	&\,+ \frac{46}{15} (\dbar K_{a_3 a_2}) D_{a_5} K_{a_4 a_1} - \frac{28}{45} (\dbar K_{a_4 a_1}) D_{a_5} K_{a_3 a_2} - \frac{20}{9} (\dbar K_{a_3 a_1}) D_{a_5} K_{a_4 a_2}\nonumber \\
	&\,- \frac{113}{45} (\dbar K_{a_2 a_1}) D_{a_5} K_{a_4 a_3} + \,^1 F_{a_5 a_4 a_3 a_2 a_1} \label{MonsterLim61}
\end{align}
\begin{align}
	[D_{a_5} D_{a_4} D_{a_3} D_{a_2} \dbar^2 \Sigma] &= \frac{29}{10} K_{a_5 a_4} \dbar^2 K_{a_3 a_2} + \frac{5}{2} K_{a_5 a_3} \dbar^2 K_{a_4 a_2} + \frac{1}{10} K_{a_4 a_3} \dbar^2 K_{a_5 a_2} \nonumber \\
	&\quad +\frac{1}{2} K_{a_5 a_2} \dbar^2 K_{a_4 a_3} + \frac{1}{2} K_{a_4 a_2} \dbar^2 K_{a_5 a_3} + \frac{11}{10} K_{a_3 a_2} \dbar^2 K_{a_5 a_4}\nonumber\\
	&\quad + \frac{8}{45} \,^{(3)}R_{a_4 a_3 a_5}\,^b \dbar K_{b a_2} + \frac{52}{45} \,^{(3)}R_{a_5 a_4 a_3}\,^b \dbar K_{b a_2} - \frac{1}{5} \,^{(3)}R_{a_3 a_2 a_5}\,^b \dbar K_{b a_4} \nonumber\\
	&\quad + \frac{3}{5} \,`{(3)}R_{a_5 a_3 a_2}\,^b \dbar K_{b a_4} + K_{a_4}\,^b K_{b a_3} \dbar K_{a_5 a_2} + K_{a_5}\,^b K_{b a_2} \dbar K_{a_4 a_3} \nonumber\\
	&\quad + K_{a_4}\,^b K_{b a_2} \dbar K_{a_5 a_3} + 3 K_{a_5}\,^b K_{b a_3} \dbar K_{a_4 a_2} - \frac{509}{90} K_{a_5}\,^b K_{a_4 a_3} \dbar K_{b a_2} \nonumber\\
	&\quad -\frac{167}{99} K_{a_4}\,^b K_{a_5 a_3} \dbar K_{b a_2} -\frac{329}{90} K_{a_3}\,^b K_{a_5 a_4} \dbar K_{b a_2} -\frac{39}{10} K_{a_5}\,^b K_{a_4 a_2} \dbar K_{b a_3} \nonumber\\
	&\quad -\frac{5}{6} K_{a_4}\,^b K_{a_4 a_2} \dbar K_{b a_3} - \frac{97}{30} K_{a_5}\,^b K_{a_5 a_4} \dbar K_{b a_3} - \frac{37}{10} K_{a_5}\,^b K_{a_3 a_2} \dbar K_{b a_4} \nonumber\\
	&\quad -\frac{49}{30} K_{a_3}\,^b K_{a_5 a_2} \dbar K_{b a_4} -\frac{91}{30} K_{a_2}\,^b K_{a_5 a_3} \dbar K_{b a_4} -\frac{21}{10} K_{a_4}\,^b K_{a_3 a_2} \dbar K_{b a_5} \nonumber\\
	&\quad -\frac{31}{30} K_{a_3}\,^b K_{a_4 a_2} K_{b a_5} - \frac{13}{30} K_{a_2}\,^b K_{a_4 a_3} \dbar K_{b a_5} + \frac{29}{5} K_{a_5}\,^b K_{b a_4} \dbar K_{a_3 a_2} \nonumber\\
	&\quad -\frac{1}{3}K_{a_4}\,^b K_{b a_5} \dbar K_{a_3 a_2} + \frac{5}{3} K_{a_5}\,^b K_{b a_3} \dbar K_{a_4 a_2} - \frac{17}{15} K_{a_4}\,^b K_{b a_3} \dbar K_{a_5 a_2} \nonumber\\
	&\quad -\frac{1}{3} K_{a_5}\,^b K_{b a_2} \dbar K_{a_4 a_3} - \frac{1}{3} K_{a_4}\,^b K_{b a_2} \dbar K_{a_5 a_3} + \frac{11}{5} K_{a_3}\,^b K_{b a_2} \dbar K_{a_5 a_4} \nonumber\\
	&\quad -\frac{1}{3} K_{a_3}\,^b K_{b a_2} \dbar K_{a_5 a_4} + \,^3 F_{a_5 a_4 a_3 a_2} \label{MonsterLim63}
\end{align}
where $\,^1 F_{a_5 a_4 a_3 a_2 a_1}$ and $\,^3 F_{a_5 a_4 a_3 a_2}$ are tangent tensors which do not involve normal derivatives of the extrinsic curvature. Note that both the sixth order tangent derivative of $\Sigma$ and the fifth order tangent derivative of $\dbar \Sigma$ involve normal derivatives of the extrinsic curvature $K_{ab}$ of first order, i.e., second order \textit{normal Lie derivatives} of the intrinsic metric $h_{ab}$, but the fourth order derivative of $\dbar^2 \Sigma$ already includes second order derivatives of $K_{ab}$, implying third order \textit{normal Lie derivatives} of the intrinsic metric $h_{ab}$. This fact represents a complication in treating the initial value formulation of semiclassical gravity, a subject that was extensively discussed in \cite{JKMS2023:Programmatic}.

\subsection{Projection of hybrid derivatives of \texorpdfstring{$\Sigma$}{S}} \label{3p1Hybrid}
Just as explained in Appendix \ref{LimitsMixedDerivativesSigma}, for many purposes $\nabla_{a'_0} \Sigma$ can be regarded as a scalar. We can assume this literally in the case of the $3+1$ decomposition of its derivatives, working with the \textit{scalar} $\zeta=\nabla_{c'} \Sigma$ within our $3+1$ formalism, and only when limits are taken, does the vector nature of $\zeta_{c'}$ emerge, \textit{shifting} the labeling of the components by \textit{one bit left}. For example, we will represent $\zeta$ by the scalar
\begin{equation}
	\zeta = \mathscr{D} + \mathscr{N},
\end{equation}
where $\mathscr{D}$ is identified with $D_{c'}\Sigma$ and $\mathscr{N}$ is identified with $n_{a'}\dbar'\Sigma=-n_{c'} n^{b'}\nabla_{b'} \Sigma$. In the coincidence limit, $\zeta \to [\Sigma_{a'_0}]$, which is a vector, so in the representation of the zeroth-order tensor (ie., scalar) $\zeta$, the $3+1$ projections go from $2^0$ to $2^1$, that is,
\begin{align}
		\lim\limits_{\x'\to\x} \zeta \mapsto \,^0[\zeta]_{a_0}+ n_{a_0}\,^1[\zeta],
\end{align}
with
\begin{subequations}\label{DprimeSigma}
	\begin{align}
		\,^0[\zeta]_{a_0} &= [D_{a'_0}\Sigma] = 0, \\
		\,^1[\zeta] &= [\dbar'\Sigma]= 0.
	\end{align}
\end{subequations}
The \textit{one-bit left-shift} multiplies the components by $2$, so for each component labeled $m$, there will be two components, $2m$ and $2m +1$. Component $2m$ will result from the substitutions $\mathscr{D}\to D_{a'_0} \Sigma$ and $\mathscr{N}\to 0$ followed by taking the coincidence limit, while the component $2m+1$ will result from the substitution of $\mathscr{N}\to \dbar'\Sigma$ and $\mathscr{D}\to 0$, followed by taking the coincidence limit. Then, for the derivative of $\zeta$ we have
\begin{align}
		\nabla_{a_0} \zeta &= D_{a_0}\zeta + n_{a_0} \dbar \zeta,
\end{align}
that is, tangent projection component $m=0$ is $D_{a_0}\zeta$ and the tangent projection component $m=1$ is $\dbar \zeta$. When we compute the coincidence limit we will have:
\begin{subequations} \label{3p1HybridDerivativesSigma1}
	\begin{align}
		m&=0, \quad [D_{a_1} D_{a'_0} \Sigma] = -h_{a_1 a_0},\\
		m&=1, \quad [D_{a_1} \dbar' \Sigma] = 0,\\
		m&=2, \quad [\dbar D_{a'_0} \Sigma] = 0,\\
		m&=3, \quad [\dbar \dbar' \Sigma] = 1,
	\end{align}
\end{subequations}
where the right hand comes from the $3+1$ projection of \eqref{limDDHybridSigma}. Note that the label of each index increased one unit to fit the \textit{new} index $a_0$ at the right. By following this procedure, and considering the $3+1$ projections of the limits obtained in Appendix \ref{LimitsMixedDerivativesSigma}, we obtain the following limits for hybrid derivatives in $3+1$ form, for the second derivative of $\zeta$,
\begin{subequations} \label{3p1HybridDerivativesSigma2}
\begin{align}
	[D_{a_2} D_{a_1}D_{a_0}]&=0, \\
	[D_{a_2} D_{a_1}\dbar'\Sigma]&=-K_{a_2a_1}, \\
	[D_{a_2} \dbar D_{a_0}\Sigma] &=K_{a_2a_0}, \\
	[D_{a_2} \dbar\dbar'\Sigma] &=0, \\
	[D_{a_1} \dbar D_{a_0}\Sigma] &=0, \\
	[D_{a_1} \dbar\dbar'\Sigma] &=-u_{a_1}, \\
	[\dbar^2 D_{a_0}\Sigma]&=u_{a_0}, \\
	[\dbar^2 \dbar'\Sigma]&=0,
\end{align}
\end{subequations}
for the third derivative of $\zeta$,
\begin{subequations} \label{3p1HybridDerivativesSigma3}
\begin{align}
	[D_{a_3} D_{a_2} D_{a_1}D_{a'_0} \Sigma]  &=-\frac{2}{3} K_{a_2a_1} K_{a_3a_0}-\frac{1}{6} K_{a_2a_0} K_{a_3a_1}-\frac{1}{6} K_{a_1a_0} K_{a_3a_2}-\frac{1}{6} \,^{(3)}R_{a_3a_0a_2a_1}-\frac{1}{3} \,^{(3)}R_{a_3a_2a_1a_0}, \\
	[D_{a_3} D_{a_2} D_{a_1}\dbar' \Sigma]  &=-\frac{1}{6} D_{a_1} K_{a_3a_2} -\frac{1}{6} D_{a_2} K_{a_3a_1} -\frac{2}{3} D_{a_3} K_{a_2a_1}, \\
	[D_{a_3} D_{a_2} \dbar D_{a'_0} \Sigma]  &=-\frac{1}{6} D_{a_0} K_{a_3a_2} +\frac{1}{3} D_{a_2} K_{a_3a_0} +\frac{5}{6} D_{a_3} K_{a_2a_0}, \\
	[D_{a_3} D_{a_2} \dbar \dbar' \Sigma]  &=\frac{7}{6} K^{b_0}\,_{a_2} K_{b_0a_3}-\frac{1}{6} u_{a_2} u_{a_3}-\frac{1}{6} \dbar  K_{a_3a_2} +\frac{1}{6} D_{a_3} u_{a_2}, \\
	[D_{a_3} D_{a_1} \dbar D_{a'_0} \Sigma]  &=-K_{a_3a_0} u_{a_1}-\frac{1}{6} D_{a_0} K_{a_3a_1} +\frac{1}{3} D_{a_1} K_{a_3a_0} -\frac{1}{6} D_{a_3} K_{a_1a_0}, \\
	[D_{a_3} D_{a_1} \dbar \dbar' \Sigma]  &=\frac{1}{6} K^{b_0}\,_{a_1} K_{b_0a_3}-\frac{1}{6} u_{a_1} u_{a_3}-\frac{1}{6} \dbar  K_{a_3a_1} -\frac{5}{6} D_{a_3} u_{a_1}, \\
	[D_{a_3} \dbar ^2D_{a'_0} \Sigma] &=-\frac{1}{3} K^{b_0}\,_{a_0} K_{b_0a_3}+\frac{1}{3} u_{a_0} u_{a_3}+\frac{1}{3} \dbar  K_{a_3a_0} +\frac{2}{3} D_{a_3} u_{a_0}, \\
	[D_{a_3} \dbar ^2\dbar' \Sigma] &=K^{b_0}\,_{a_3} u_{b_0}\\
	[D_{a_2} D_{a_1} \dbar D_{a'_0} \Sigma]  &=-K_{a_2a_1} u_{a_0}+\frac{1}{3} D_{a_0} K_{a_2a_1} -\frac{1}{6} D_{a_1} K_{a_2a_0} -\frac{1}{6} D_{a_2} K_{a_1a_0}, \\
	[D_{a_2} D_{a_1} \dbar \dbar' \Sigma]  &=-\frac{1}{3} K^{b_0}\,_{a_1} K_{b_0a_2}+\frac{1}{3} u_{a_1} u_{a_2}-\frac{2}{3} \dbar  K_{a_2a_1} -\frac{1}{3} D_{a_2} u_{a_1}, \\
	[D_{a_2} \dbar ^2D_{a'_0} \Sigma] &=\frac{1}{6} K^{b_0}\,_{a_0} K_{b_0a_2}-\frac{1}{6} u_{a_0} u_{a_2}+\frac{5}{6} \dbar  K_{a_2a_0} +\frac{1}{6} D_{a_2} u_{a_0}, \\
	[D_{a_2} \dbar ^2\dbar' \Sigma] &=K^{b_0}\,_{a_2} u_{b_0}\\
	[\dbar ^2 D_{a_1}D_{a'_0} \Sigma] &=\frac{1}{6} K^{b_0}\,_{a_0} K_{b_0a_1}-\frac{7}{6} u_{a_0} u_{a_1}-\frac{1}{6} \dbar  K_{a_1a_0} +\frac{1}{6} D_{a_1} u_{a_0}, \\
	[\dbar ^2 D_{a_1}\dbar' \Sigma] &=-\dbar  u_{a_1}, \\
	[\dbar ^3D_{a'_0} \Sigma]&=\dbar  u_{a_0}, \\
	[\dbar ^3\dbar' \Sigma]&=u^2,
\end{align}
\end{subequations}
for the fourth derivatives of $\zeta$, we will only write the terms relevant for the construction of the surface Taylor series of $\dbar'\Sigma$ and $\dbar \dbar' \Sigma$:
\begin{subequations}
	\begin{align}
		[D_{a_4} D_{a_3} D_{a_2} D_{a_1}\dbar' \Sigma] =& -\frac{1}{3} K_{a_4 a_3} K^{b_0}\,_{a_1} K_{b_0 a_2}+\frac{1}{6} K_{a_4 a_2} K^{b_0}\,_{a_1} K_{b_0 a_3}-\frac{5}{6} K_{a_4 a_1} K^{b_0}\,_{a_2} K_{b_0 a_3} \nonumber\\
		&+\frac{1}{2} K_{a_3 a_2} K^{b_0}\,_{a_1} K_{b_0 a_4} -\frac{1}{2} K_{a_3 a_1} K^{b_0}\,_{a_2} K_{b_0 a_4}-\frac{1}{2} K_{a_2 a_1} K^{b_0}\,_{a_3} K_{b_0 a_4} \nonumber\\
		&+\frac{1}{6} K_{b_0 a_4} \,^{(3)}R^{b_0}\,_{a_1 a_3 a_2}+\frac{1}{2} K_{b_0 a_3} \,^{(3)}R^{b_0}\,_{a_1 a_4 a_2} +\frac{1}{2} K_{b_0 a_2} \,^{(3)}R^{b_0}\,_{a_1 a_4 a_3} \nonumber\\
		&+\frac{1}{2} K_{b_0 a_1} \,^{(3)}R^{b_0}\,_{a_2 a_4 a_3}-\frac{1}{6} K_{b_0 a_4} \,^{(3)}R^{b_0}\,_{a_3 a_2 a_1}+\frac{1}{3} K_{a_4 a_3} u_{a_1} u_{a_2}-\frac{1}{6} K_{a_4 a_2} u_{a_1} u_{a_3} \nonumber\\
		&+\frac{1}{3} K_{a_4 a_1} u_{a_2} u_{a_3}-\frac{1}{3} K_{a_3 a_2} u_{a_1} u_{a_4}+\frac{1}{6} K_{a_3 a_1} u_{a_2} u_{a_4}+\frac{1}{6} K_{a_2 a_1} u_{a_3} u_{a_4} \nonumber\\
		&+\frac{1}{3} K_{a_4 a_3} \dbar K_{a_2 a_1} -\frac{1}{6} K_{a_4 a_2} \dbar K_{a_3 a_1} +\frac{1}{3} K_{a_4 a_1} \dbar K_{a_3 a_2} -\frac{1}{3} K_{a_3 a_2} \dbar K_{a_4 a_1}  \nonumber\\
		&+\frac{1}{6} K_{a_3 a_1} \dbar K_{a_4 a_2} +\frac{1}{6} K_{a_2 a_1} \dbar K_{a_4 a_3} -\frac{1}{3} K_{a_4 a_3} D_{a_2} u_{a_1} +\frac{1}{6} K_{a_4 a_2} D_{a_3} u_{a_1}  \nonumber\\
		&-\frac{1}{3} K_{a_4 a_1} D_{a_3} u_{a_2} -\frac{1}{2} D_{a_3} D_{a_2} K_{a_4 a_1} +\frac{1}{3} K_{a_3 a_2} D_{a_4} u_{a_1} -\frac{1}{6} K_{a_3 a_1} D_{a_4} u_{a_2} \nonumber\\
		&-\frac{1}{6} K_{a_2 a_1} D_{a_4} u_{a_3}+\frac{1}{6} D_{a_4} D_{a_1} K_{a_3 a_2} -\frac{1}{3} D_{a_4} D_{a_2} K_{a_3 a_1} -\frac{1}{3} D_{a_4} D_{a_3} K_{a_2 a_1}, \\
		[D_{a_4} D_{a_3} D_{a_2}\dbar \dbar' \Sigma] =& \frac{1}{6} K^{b_0}\,_{a_3} D_{a_2} K_{b_0 a_4} +\frac{1}{2} K^{b_0}\,_{a_4} D_{a_3} K_{b_0 a_2} +\frac{1}{6} K^{b_0}\,_{a_2} D_{a_3} K_{b_0 a_4} +\frac{1}{2} K^{b_0}\,_{a_3} D_{a_4} K_{b_0 a_2}  \nonumber\\
		&+\frac{1}{2} K^{b_0}\,_{a_2} D_{a_4} K_{b_0 a_3} +\frac{1}{6} u_{a_3} D_{a_4} u_{a_2} +\frac{1}{6} u_{a_2} D_{a_4} u_{a_3} +\frac{1}{6} D_{a_4} \dbar K_{a_3 a_2}  \nonumber\\
		&-\frac{1}{6} D_{a_4} D_{a_3} u_{a_2} +\frac{1}{2} K^{b_0}\,_{a_4} D_{b_0} K_{a_3 a_2} +\frac{1}{6} K^{b_0}\,_{a_3} D_{b_0} K_{a_4 a_2} +\frac{1}{6} K^{b_0}\,_{a_2} D_{b_0} K_{a_4 a_3}.
	\end{align}
\end{subequations}
For the fifth derivatives of $\zeta$, we will only write the leading order terms in normal derivatives relevant for the construction of the surface Taylor series of $\dbar'\Sigma$ and $\dbar \dbar' \Sigma$:
\begin{subequations}
	\begin{align}
		[D_{a_5} D_{a_4} D_{a_3} D_{a_2} D_{a_1}\dbar' \Sigma] = &-\frac{1}{6} \dbar  K_{a_5 a_1}  D_{a_2} K_{a_4 a_3} -\frac{11}{90} \dbar  K_{a_4 a_1}  D_{a_2} K_{a_5 a_3} +\frac{1}{6} \dbar  K_{a_5 a_4}  D_{a_3} K_{a_2 a_1} \nonumber\\
		&+\frac{2}{9} \dbar  K_{a_5 a_2}  D_{a_3} K_{a_4 a_1} +\frac{2}{9} \dbar  K_{a_5 a_1}  D_{a_3} K_{a_4 a_2} -\frac{1}{10} \dbar  K_{a_4 a_2}  D_{a_3} K_{a_5 a_1} \nonumber\\
		&-\frac{2}{5} \dbar  K_{a_4 a_1}  D_{a_3} K_{a_5 a_2} +\frac{11}{18} \dbar  K_{a_2 a_1}  D_{a_3} K_{a_5 a_4} +\frac{1}{6} \dbar  K_{a_5 a_3}  D_{a_4} K_{a_2 a_1} \nonumber\\
		&+\frac{2}{45} \dbar  K_{a_5 a_2}  D_{a_4} K_{a_3 a_1} +\frac{1}{9} \dbar  K_{a_5 a_1}  D_{a_4} K_{a_3 a_2} +\frac{1}{90} \dbar  K_{a_3 a_2}  D_{a_4} K_{a_5 a_1} \nonumber\\
		&-\frac{4}{45} \dbar  K_{a_2 a_1}  D_{a_4} K_{a_5 a_3} +\frac{1}{6} \dbar  K_{a_4 a_3}  D_{a_5} K_{a_2 a_1} +\frac{1}{6} \dbar  K_{a_4 a_2}  D_{a_5} K_{a_3 a_1} \nonumber\\
		&-\frac{1}{90} \dbar  K_{a_4 a_1}  D_{a_5} K_{a_3 a_2} +\frac{23}{90} \dbar  K_{a_3 a_2}  D_{a_5} K_{a_4 a_1} -\frac{1}{30} \dbar  K_{a_3 a_1}  D_{a_5} K_{a_4 a_2} \nonumber\\
		&-\frac{1}{18} \dbar  K_{a_2 a_1}  D_{a_5} K_{a_4 a_3} + \,^1\mathfrak{Q}_{a_5 a_4 a_3 a_2 a_1},
	\end{align}
	\begin{align}
		[D_{a_5} D_{a_4} D_{a_3} D_{a_2}\dbar \dbar' \Sigma] =
		&-\frac{1}{15} K_{a_5 a_4} \dbar^2 K_{a_3 a_2} -\frac{1}{10} K_{a_5 a_3} \dbar^2 K_{a_4 a_2} +\frac{1}{10} K_{a_4 a_3} \dbar^2 K_{a_5 a_2} \nonumber\\
		&-\frac{3}{10} K^{b_0}\,_{a_4} K_{b_0 a_5} \dbar K_{a_3 a_2} -\frac{8}{15} K^{b_0}\,_{a_3} K_{b_0 a_5} \dbar K_{a_4 a_2} -\frac{1}{3} K^{b_0}\,_{a_2} K_{b_0 a_5} \dbar K_{a_4 a_3} \nonumber\\
		&+\frac{1}{30} K^{b_0}\,_{a_3} K_{b_0 a_4} \dbar K_{a_5 a_2} -\frac{1}{6} K^{b_0}\,_{a_2} K_{b_0 a_4} \dbar K_{a_5 a_3} -\frac{1}{6} K^{b_0}\,_{a_2} K_{b_0 a_3} \dbar K_{a_5 a_4} \nonumber\\
		&-\frac{17}{45} K_{a_5 a_4} K_{b_0 a_3} \dbar K^{b_0}\,_{a_2} -\frac{11}{45} K_{a_5 a_3} K_{b_0 a_4} \dbar K^{b_0}\,_{a_2} -\frac{26}{45} K_{a_4 a_3} K_{b_0 a_5} \dbar K^{b_0}\,_{a_2} \nonumber\\
		&+\frac{1}{45} \,^{(3)}R^{b_0}\,_{a_3 a_5 a_4} \dbar K^{b_0}\,_{a_2} +\frac{1}{9} \,^{(3)}R^{b_0}\,_{a_5 a_4 a_3} \dbar K_{b_0 a_2} -\frac{4}{15} K_{a_5 a_4} K^{b_0}\,_{a_2} \dbar K_{b_0 a_3} \nonumber\\
		&+\frac{1}{6} K_{a_5 a_2} K^{b_0}\,_{a_4} \dbar K_{b_0 a_3} -\frac{1}{10} K_{a_4 a_2} K^{b_0}\,_{a_5} \dbar K_{b_0 a_3} +\frac{4}{15} K_{a_5 a_3} K^{b_0}\,_{a_2} \dbar K_{b_0 a_4} \nonumber\\
		&+\frac{4}{15} K_{a_5 a_2} K^{b_0}\,_{a_3} \dbar K_{b_0 a_4} +\frac{1}{6} K_{a_3 a_2} K^{b_0}\,_{a_5} \dbar K_{b_0 a_4} -\frac{1}{10} \,^{(3)}R^{b_0}\,_{a_2 a_5 a_3} \dbar K_{b_0 a_4} \nonumber\\
		&+\frac{7}{30} K_{a_4 a_3} K^{b_0}\,_{a_2} \dbar K_{b_0 a_5} +\frac{7}{30} K_{a_4 a_2} K^{b_0}\,_{a_3} \dbar K_{b_0 a_5} +\frac{1}{6} K_{a_3 a_2} K^{b_0}\,_{a_4} \dbar K_{b_0 a_5} \nonumber\\
		&+\frac{1}{10} \,^{(3)}R^{b_0}\,_{a_2 a_4 a_3} \dbar K_{b_0 a_5} -\frac{1}{10} K_{b_0 a_3} \dbar \,^{(3)}R^{b_0}\,_{a_2 a_5 a_4} -\frac{2}{15} K_{a_5 a_4} u_{a_3} \dbar u_{a_2} \nonumber\\
		&-\frac{1}{5} K_{a_5 a_3} u_{a_4} \dbar u_{a_2} +\frac{1}{5} K_{a_4 a_3} u_{a_5} \dbar u_{a_2} -\frac{1}{15} K_{a_5 a_4} u_{a_2} \dbar u_{a_3} -\frac{1}{10} K_{a_5 a_3} u_{a_2} \dbar u_{a_4} \nonumber\\
		&+\frac{1}{10} K_{a_4 a_3} u_{a_2} \dbar u_{a_5} + \,^2\mathfrak{Q}_{a_5 a_4 a_3 a_2}
	\end{align}
\end{subequations}
where $\,^1\mathfrak{Q}_{a_5 a_4 a_3 a_2 a_1}$ and $\,^2\mathfrak{Q}_{a_5 a_4 a_3 a_2}$ are tangent tensors which do not involve normal derivatives.

\end{document}